%% file: chapter-2-for-arxiv.tex
\documentclass[runningheads]{llncs}
\usepackage[T1]{fontenc}
\usepackage{tabularx}
\usepackage{multirow}
\usepackage{amsmath}
\usepackage{mathptmx}
\usepackage{helvet}
\usepackage{courier}
\usepackage{type1cm}         
\usepackage{makeidx}         
\usepackage{graphicx}        
\usepackage{multicol}        
\usepackage[bottom]{footmisc}
\usepackage{newtxtext}       %
\usepackage[varvw]{newtxmath}       
\usepackage{hyperref}
\usepackage[sort&compress,numbers]{natbib}
\usepackage{subfig}
\usepackage{verbatim}
\usepackage{rotating}
\usepackage{xcolor}
\usepackage{tikz}
\usepackage{setspace}
\def\drawcross#1#2#3{
	\draw[thick,#3] (#1-.1,#2+.1) -- (#1+.1,#2-.1) (#1-.1,#2-.1) -- (#1+.1,#2+.1);
}
\input{2013-10-19-Macros-common.tex}

\usetikzlibrary {arrows.meta,shapes.misc}
\date{}

\begin{document}
\author{David Meredith}
\title{A General Theoretical Framework for Transformed Geometric Pattern Matching in Music}
\titlerunning{Framework for Geometric Pattern Matching in Music}
\institute{Aalborg University, Aalborg, Denmark\\{\small \email{dave@create.aau.dk}}\\{\small \email{dave@titanmusic.com}}\\{\small \url{https://www.titanmusic.com}}}

\maketitle
\input{abstract02}
\input{chapter-02-theoretical-framework}

\end{document}

%% file: 2013-10-19-Macros-common.tex
\usepackage{musicography} 

\newcommand{\colWidth}{\ \ \ \ }
\newcommand{\firstColWidth}{\ \ \ \ \ \ \ \ }

\newcommand{\keyword}[1]{\textrm{\bfseries #1}}

\newcommand{\dmreturn}{\keyword{return\ }}

\newcommand{\dmif}{\keyword{if\ }}
\newcommand{\dmto}{\keyword{\ to\ }}
\newcommand{\dmfor}{\keyword{for\ }}

\newcommand{\dmwhile}{\keyword{while\ }}
\newcommand{\dmelse}{\keyword{else\ }}

\newcommand{\becomes}{\leftarrow}
\newcommand{\dmcomment}[1]{{\color{gray}\qquad\blacktriangleright\;\textsl{#1}}}
\newcommand{\dmcommentnotab}[1]{{\color{gray}\blacktriangleright\;\textsl{#1}}}

\newcommand{\mathsc}[1]{{\normalfont\textsc{#1}}}

\newcommand{\fun}[1]{\mathsc{#1}}

\newcommand{\sorted}[1]{\mathbf{#1}}

\newcommand{\point}{p}

\newcommand{\concat}{\oplus}
\newcommand{\bigconcat}{\bigoplus}
\newcommand{\dataset}{D}

\newcommand{\covsetfun}{\fun{COV}}

\newcommand{\GetCompressionFactor}{\fun{CF}}
\newcommand{\compfacfun}{\GetCompressionFactor}

\newcommand{\pattern}{P}

\newcommand{\reals}{\mathbb{R}}
\newcommand{\integers}{\mathbb{Z}}

\newcommand{\transformation}{f}

\newcommand{\tc}{F}
\newcommand{\tcFun}[1]{\tc_{\mathrm{#1}}}
\newcommand{\FTwoT}{\tcFun{2T}}
\newcommand{\FTwoTR}{\tcFun{2TR}}
\newcommand{\FTwoST}{\tcFun{2ST}}
\newcommand{\FTwoSTR}{\tcFun{2STR}}
\newcommand{\FTwoSTRMod}[1]{\tcFun{2STRmod{#1}}}
\newcommand{\ftwostrmu}{\FTwoSTRMod{\mu}}

\newcommand{\basisSize}{\beta}

\newcommand{\basis}{\mathbf{b}}
\newcommand{\objectBasis}{{\basis}_{\mathrm{obj}}}
\newcommand{\imageBasis}{{\basis}_{\mathrm{img}}}

\newcommand{\BTSF}{\mathcal T}

\newcommand{\maxmatchfun}{\fun{M}_{\max}}
\newcommand{\maxpatfun}{\fun{P}_{\max}}
\newcommand{\occsetfun}{\fun{OS}}

\newcommand{\phiftwostd}{\phi_{\FTwoST'}}
\newcommand{\phiftwostrd}{\phi_{\FTwoSTR'}}

\newcommand{\allsigmas}{\mathcal U_\sigma}

\newcommand{\len}{\fun{Len}}

\newcommand{\setfun}{\fun{Set}}

\newcommand{\llpar}{\left(}
\newcommand{\rrpar}{\right)}

\newcommand{\llfloor}{\left\lfloor}
\newcommand{\rrfloor}{\right\rfloor}



%% file: abstract02.tex
\begin{abstract}
In this paper, we present a framework of concepts, terminology and notation for geometric pattern matching, centred around the idea of a {\em transformation class}. After formally defining the concept of a transformation class, we provide detailed definitions of a number of specific transformation classes that are particularly relevant to pattern matching and pattern discovery in music, which is the primary application domain within which the framework has been developed. We then provide formal definitions of the notions of {\em match}, {\em occurrence}, {\em maximal match\/} and {\em maximal occurrence\/} and prove some results relating to these concepts. We then define and explore in detail the idea of an {\em occurrence set\/} and the more general notion of a {\em heterogeneous occurrence set}. Finally, we explore in detail the concept of {\em transformational equivalence}, including a consideration of {\em transformation groups}, along with the group-theoretic and transformational equivalence properties of selected transformation classes. 
\end{abstract}

%% file: chapter-02-theoretical-framework.tex
\section{Introduction}

In this paper, we present a framework of concepts, terminology and notation for geometric pattern matching, centred around the idea of a {\em transformation class}. After defining some preliminary ideas, notation and terminology in Section~\ref{preliminary-concepts-section}, we provide a formal definition of the concept of a transformation class in Section~\ref{section-transformation-classes}. 
We then provide detailed definitions of a number of specific transformation classes that are particularly relevant to pattern matching and pattern discovery in music, as this has been the primary application domain within which the framework has been developed. In Section~\ref{matches-and-occurrences-section}, we provide formal definitions of the notions of {\em match}, {\em occurrence}, {\em maximal match\/} and {\em maximal occurrence\/} and prove some results relating to these concepts. In Section~\ref{main-occurrence-sets-section}, we define and explore in detail the notion of an {\em occurrence set\/} and the more general idea of a {\em heterogeneous occurrence set}. Finally, in Section~\ref{transformational-equivalence-section}, we explore in detail the concept of {\em transformational equivalence}, including a consideration of {\em transformation groups}, along with the group-theoretic and transformational equivalence properties of selected transformation classes.

Although the framework described in this paper has been developed primarily with music in mind, it is actually quite general and could, in principle, be applied in any domain where data takes the form of multi-dimensional point sets and meaningful relationships between sets of data points can be modelled by geometric transformations.
The framework presented here generalizes and develops further that proposed by \citet{MeredithArxiv2022,MeredithHCII2023,MeredithMML2024}. The transformation classes discussed in this paper were applied in \cite{MeredithArxivChapter12026} to the analysis of the relationships between occurrences of the so-called ``HAYDN'' theme in Maurice Ravel's {\em Menuet sur le nom d'Haydn}.

\section{Preliminary definitions of concepts, terminology and notation}
\label{preliminary-concepts-section}
In this section we define some basic concepts, terminology and notation that will be used throughout the remainder of this paper.

\subsection{Sets and ordered sets}
\label{sets-and-ordered-sets-section}
A {\em set}, $S=\lbrace s_1, s_2, \ldots s_n\rbrace$, is a collection of {\em distinct\/} objects. If $S$ is a set, then the {\em cardinality\/} of $S$, denoted by $|S|$, is the number of elements that it contains. Thus, if $S=\lbrace s_1, s_2, \ldots s_n\rbrace$, then $|S|=n$. If $\mathit\Psi=\lbrace S_1, S_2,\ldots S_m\rbrace$, where each $S_i$ is a set, then we define the following shorthand notations:
$\bigcup_{S\in\mathit\Psi} S= S_1\cup S_2 \cup\ldots \cup S_{m}$
and
$\bigcap_{S\in\mathit\Psi} S= S_1\cap S_2 \cap\ldots \cap S_{m}$.
Note that we will use braces, $\lbrace \cdot \rbrace$, to denote sets and name them with capital italic letters. We will reserve the operators, $\subset$ and $\supset$, to mean {\em proper\/} subset and superset, respectively. That is, $A\subset B\Rightarrow A\neq B$ and $A\supset B\Rightarrow A\neq B$. If we want to express that a set, $A$, is either a proper subset of $B$ or is equal to $B$, then we write $A\subseteq B$.

An {\em ordered set}, $\sorted A=\langle a_1, a_2, \ldots a_k \rangle$, is a collection of (not necessarily distinct) objects where the order in which the objects appear in the collection matters. Examples of ordered sets include sequences, vectors, lists, arrays and tuples. Thus, the ordered sets, $\sorted A=\langle 1, 2, 3\rangle$ and $\sorted B=\langle 2, 1, 3\rangle$, are distinct, whereas the sets, $X=\lbrace 1, 2, 3\rbrace$ and $Y=\lbrace 2, 1, 3\rbrace$, are equal. Note that we use angle brackets, $\langle \cdot \rangle$, to denote ordered sets and name them with bold letters. When an ordered set represents a vector or a tuple, we typically name it with a lower-case letter, whereas when it represents a sequence, array or list, we typically name it with a capital letter. We denote the set of distinct elements in an ordered set, $\sorted A$, by $\setfun(\sorted A)$ and define it as follows: $\setfun(\sorted A)=\bigcup_{a\in \sorted A}\lbrace a\rbrace$.

Given an ordered set, $\sorted A=\langle a_1, a_2, \ldots a_k \rangle$, then $\sorted A[i]$ denotes the $(i+1)$th element in $\sorted A$ (i.e., we use zero-based indexing). For example, if $\sorted B=\langle 2, 1, 3\rangle$, then 
$\sorted B[0]= 2$, $\sorted B[1]= 1$ and $\sorted B[2]= 3$. 
We define the {\em length\/} or {\em size\/} of $\sorted A$, denoted by $\len(\sorted A)$, to be the number of elements in $\sorted A$---for example, if $\sorted A=\langle a_1, a_2, \ldots a_k \rangle$, then $\len(\sorted A)=k$ and $\sorted A[\len(\sorted A)-1]$ is the last element in $\sorted A$.
We define the {\em cardinality\/} of $\sorted A$, denoted by $|\sorted A|$, to be the number of {\em distinct\/} elements in $\sorted A$. Thus $|\sorted A|=\left|\bigcup_{i=0}^{\len(A)-1}\lbrace\sorted A[i]\rbrace\right|$.
An ordered set, $\sorted A$ can be represented as a set of ordered pairs, $\bigcup_{i=0}^{\len(\sorted A)-1}\lbrace\langle  i,\sorted A[i]\rangle\rbrace$.

We define that, if $\sorted A$ is an ordered set, and $i$ and $j$ are integers such that $0\le i\le j\le\len(\sorted A)$, then $\sorted A[i,j]=\langle \sorted A[i],\sorted A[i+1],\ldots \sorted A[j-1] \rangle$. Note that $\len(\sorted A[i,j])=j-i$. For example, if $\sorted B=\langle 5, 6, 7\rangle$, then $\sorted B[0,2]=\langle5,6\rangle$, $\sorted B[1,3]=\langle6,7\rangle$ and $\sorted B[0,3]=\sorted B$. 

If $\sorted A=\langle a_1, a_2, \ldots a_k \rangle$ and $\sorted B=\langle b_1, b_2, \ldots b_\ell \rangle$ then the {\em concatenation of\/} $\sorted B$ {\em onto} $\sorted A$ is the ordered set,
$\sorted A\oplus \sorted B=\langle a_1, a_2, \ldots a_k, b_1, b_2, \ldots b_\ell\rangle$. It follows that, for any ordered set, $\sorted A= \sorted A[0,k]\concat\sorted A[k,\len(\sorted A)]$ (note that $\sorted A[i,j]$ is only defined for $i,j\in\integers$ and $0\le i\le j\le\len(\sorted A)$).

If we want to express that a sequence is the result of concatenating some sequence of other sequences together, then we can use the $\bigoplus$ operator. For example, $\bigoplus_{i=1}^{5}\langle0\rangle=\langle0\rangle\concat\langle0\rangle\concat\langle0\rangle\concat\langle0\rangle\concat\langle0\rangle=\langle0,0,0,0,0\rangle$ and $\bigoplus_{i=1}^{5}\langle i\rangle=\langle1,2,3,4,5\rangle$. Note that the argument expression of $\bigoplus_{i=j}^k$ must be an ordered set.
If $\sorted \mathit\Phi=\langle \sorted A_1,\sorted A_2,\ldots \sorted A_n\rangle$, where each $\sorted A_i$ is an ordered set, then we define the following shorthand notation:
$
\bigoplus_{i=1}^{n} \sorted A_i = \bigoplus_{\sorted A\in\sorted \mathit\Phi}\sorted A=\sorted A_1\oplus \sorted A_2\oplus\ldots \oplus \sorted A_{n}
$.

If $S$ is a set of numbers, then we define that $\min(S)$ and $\max(S)$ return the minimum and maximum values in $S$, respectively. 

If $\sorted A$ and $\sorted B$ are two ordered sets, then we define that $\sorted A<\sorted B$ if and only if $\sorted A$ is {\em lexicographically less than\/} $\sorted B$. That is, $\sorted A<\sorted B$ if and only if
\begin{enumerate}
\item there exists an integer, $0\le i<\min(\lbrace\len(\sorted A),\len(\sorted B)\rbrace)$ such that $\sorted A[i] < \sorted B[i]$ and $\sorted A[j] = \sorted B[j]$ for all $0 \le j < i$; or
\item $\len(\sorted A)<\len(\sorted B)$ and $\sorted A = \sorted B[0,\len(\sorted A)]$.
\end{enumerate}

\subsection{Pseudocode}
\label{pseudocode-section}
Algorithms are described in this paper using a form of pseudocode that we have used in several previous publications \citep[e.g., ][]{MeredithJNMR2015,Meredith2006, Meredith2006b, MeredithLemstromWiggins2002, Meredith2007b}. This pseudocode closely resembles that used by \citet{CormenEtAl2009} in terms of the assumptions that it makes about, for example, the machine model and data structures. In this pseudocode, data-structures are typically represented as sets and ordered sets. It is generally assumed that ordered sets will be implemented as arrays or some other array-like data structures that support constant-time random access. Sets are assumed to be implemented by data structures that allow for dictionary operations (add, delete and test for membership) to be carried out in $O(\log n)$ time in the worst case (e.g., by using a binary tree) or $O(1)$ time on average (e.g., by using a hash table), where $n$ is the number of elements in the set. If a specific data structure implementation is intended, then this will be explicitly stated. 

As an example, Fig.~\ref{insertion-sort-algorithm} shows pseudocode for the well-known insertion-sort algorithm. As can be seen in this example, an algorithm is expressed as a function which begins with a header line (unnumbered) in which the name of the function is printed flush-left in $\fun{CapitalizedSmallCaps}$. The name of the function is followed on the same line by a list of parameters enclosed in parentheses (in this example, simply $\mathbf A$). The numbered lines following the header line form the body of the function and describe the steps carried out by the algorithm on its parameters. 

\begin{figure}[t]
\begin{center}
\resizebox*{.7\columnwidth}{!}{
\begin{minipage}{\linewidth}
\begin{tabbing}
\input{ColumnDefinitions}
$\fun{InsertionSort}(\mathbf A)$\\
\>$\dmcommentnotab{$\mathbf A$ is an ordered set of elements in an arbitrary order.}$\\
1\>$\dmfor j \becomes 1 \dmto \len(\mathbf A)-1$\\
2\>\>$k\becomes\mathbf A[j]$\\
\>\>$\dmcommentnotab{Following inserts $\mathbf A[j]$ into the sorted sequence $\mathbf A[0,j]$}$\\
3\>\>$i\becomes j-1$\\
4\>\>$\dmwhile i \ge 0 \land \mathbf A[i] > k$\\
5\>\>\>$\mathbf A[i+1]\becomes \mathbf A[i]$\\
6\>\>\>$i\becomes i -1$\\
7\>\>$\mathbf A[i+1] \becomes k$\\
\end{tabbing}
\end{minipage}
}
\end{center}
\caption{Pseudocode for the insertion-sort algorithm (based on \citet[p.~18]{CormenEtAl2009}).}
\label{insertion-sort-algorithm}
\end{figure}

As exemplified in Fig.~\ref{insertion-sort-algorithm}, each line starts with a number and then the line itself is printed, indented by at least one tab stop. The looping constructs, {\bfseries for} and {\bfseries while}, and the conditional constructs, {\bfseries if} and {\bfseries else}, have the same meaning as they do in languages like C, C++, Java and C\#. Block structure is indicated solely by indentation, which reduces clutter and dispenses with the need for reserved keywords such as {\bfseries do}, {\bfseries begin}, {\bfseries end} and {\bfseries then}. It also removes the need for punctuation such as semicolons to mark the ends of statements. Note that comments are written in grey italics and start with a grey, right-pointing arrowhead (e.g., $\dmcommentnotab{A comment.}$). Anything to the right of such an arrowhead on a line does not form part of the algorithm. Lines that only contain comments are not numbered.

In this pseudocode, the assignment operator is $\becomes$ and the equality operator is $=$. Variables and parameters in a function should generally be considered local to the function, unless otherwise stated, and variables are not declared or typed before they are used. Typically, as in languages like Java and C\#, we assume that, when provided as the argument of a function, a primitive type value (e.g., a basic integer or floating-point value) is passed by value, while more complex structures and objects (e.g., sets and ordered sets) are passed by reference. If this distinction is significant in any specific case, then it will be noted in the text. As an example, the parameter {\bfseries A} in Fig.~\ref{insertion-sort-algorithm} is assumed to be passed by reference. This means that there is no need for the algorithm to return a new ordered set of values, as the elements in the ordered set passed to the algorithm in the parameter {\bfseries A} will be in sorted order after the algorithm has completed. Other features of the pseudocode will be introduced at the points in the paper where they are first required. 

\begin{figure}
\begin{center}
\resizebox*{.7\columnwidth}{!}{
\includegraphics{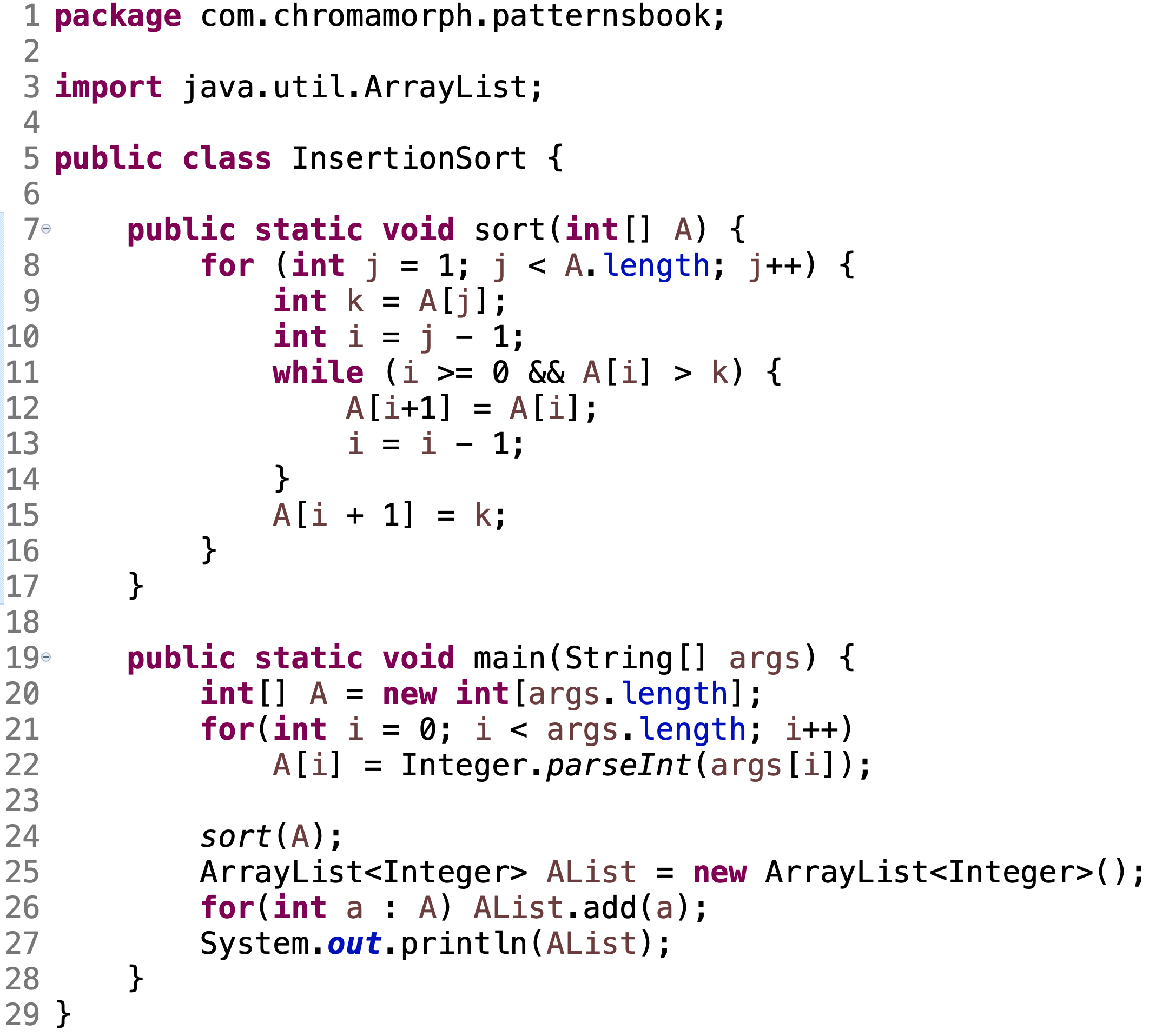}
}
\end{center}
\caption{A straightforward implementation in Java of the insertion-sort pseudocode shown in Fig.~\ref{insertion-sort-algorithm}.}
\label{insertion-sort-java}
\end{figure}

As a concrete example of how such pseudocode might be implemented, Fig.~\ref{insertion-sort-java} shows a straightforward Java implementation of the insertion sort algorithm shown in Fig.~\ref{insertion-sort-algorithm}.

\subsection{Datasets and patterns}
\label{datasets-and-patterns-section}
In this paper, we suppose that we are provided with a set of points, $\dataset\subset\mathbb{R}^k$, that we call a {\em dataset}. If $\dataset\subset\mathbb{R}^k$, then the {\em dimensionality\/} of $\dataset$ is $k$, so we can call it a {\em $k$-dimensional dataset}. We denote the {\em cardinality\/} or {\em size\/} of $\dataset$ by $|\dataset|$. Typically, in the context of an application in computational music analysis or music information retrieval, a dataset represents a piece of music, part of a piece of music, or a collection of pieces of music, with each point representing a tone, in the sense defined in footnote~2 on page~7 of \cite{MeredithArxivChapter12026}---that is, either an untied note in a score or a sequence of tied notes. Most commonly in the literature, the datasets considered are two-dimensional, with pitch represented on the $y$-axis and time represented on the $x$-axis. The eight types of two-dimensional representation introduced in \cite{MeredithArxivChapter12026}---that is, CPO, CPM, MPO, MPM, CO, CM, MO, and MM---are examples of this type of dataset. Figures~8 to 15 in \cite{MeredithArxivChapter12026} are therefore examples of datasets.
%
%

We use the term {\em pattern\/} to refer to any finite point set, $\pattern\subset\mathbb{R}^k$, usually (but not always) under consideration in the context of some specified dataset. For example, if we are interested in knowing if a point set, $A$, occurs within a dataset, then we might refer to $A$ as a ``query pattern'' or just a ``pattern'', even if it does not actually occur within the dataset. 
\subsection{Transformations}
\label{transformations-section}

In the context of a dataset, $\dataset\subset\mathbb{R}^k$,
we define a {\em transformation\/} to be a function, $\transformation:\mathbb{R}^k\rightarrow\mathbb{R}^k$.  
If a transformation, $\transformation$, maps a point $p$ onto a point $q$, then we can write $q=\transformation(p)$.
If $\transformation$ has a unique inverse, then we denote this inverse by $\transformation^{-1}$ and $\transformation^{-1}(p) = q \textrm{ if and only if } \transformation(q) = p$ for all $p,q\in\mathbb{R}^k$. 
If we have a $k$-dimensional point set, $S\subset\mathbb{R}^k$, and a transformation, $\transformation:\mathbb{R}^k\rightarrow\mathbb{R}^k$, then, for notational convenience, we define that $\transformation(S)=\bigcup_{p\in S}\lbrace \transformation(p)\rbrace$. Similarly, if $\mathbf S$ is an ordered set of $k$-dimensional points, then we define that $\transformation(\mathbf S)=\bigconcat_{i=0}^{\len(\mathbf S)-1} \langle \transformation(\mathbf S[i])\rangle$.

A transformation, $f$, is thus a binary relation, $f\subset\mathbb{R}^k\times\mathbb{R}^k$, such that for every point, $p\in\mathbb{R}^k$, there exists precisely one point, $q\in\mathbb{R}^k$, such that $(p,q)\in f$. 
If $(p,q)\in f$, then we can write $q=f(p)$ and we call $p$ the {\em argument\/} of $f$ and $q$ the {\em value\/} of $f$ for $p$. We can also say that $q$ is the {\em image\/} or {\em image point\/} of $p$ (which is the {\em object point}) under the transformation $f$.\footnote{We follow closely the terminology and notation used by \citet[pp.~1163--1168]{CormenEtAl2009}.} For example, if $k=2$ and $f$ is the transformation consisting of a translation by the vector, $\langle1,2 \rangle$, then
$
f = \left\lbrace(p,q):p,q\in\mathbb{R}^2\textrm{ and } q=p+\langle1,2 \rangle\right\rbrace
$,
where, if $p$ is a $k$-dimensional point and $v$ is a $k$-dimensional vector, then 
$
p+v=\bigoplus_{i=0}^{k-1}\langle p[i]+v[i]\rangle
$.

Note that, notwithstanding the notational shortcut just defined above, a transformation is defined to be a function that maps a {\em point\/} onto a {\em point}, not a point {\em set\/} onto a point {\em set}. This means that a mapping of a pattern that, for example, removes points from the pattern or that scales the pattern but applies different scale factors to different points in the pattern cannot be modelled (directly) within the framework presented in this paper, since it would involve applying different transformations to different points within the pattern.

Note also that the definition of transformation used in this paper is more general than the definitions given by \citet{MeredithArxiv2022,MeredithHCII2023,MeredithMML2024}, where transformations were restricted to being bijections.

\subsection{Modular arithmetic}
We define the $\bmod$ operator as follows.
\begin{definition}[$\bmod$]
We define that $a\bmod \mu=a - \mu\left\lfloor\frac a{\mu} \right\rfloor$ where $a,b,\mu\in\reals$, and $\lfloor x \rfloor$ is the largest integer less than or equal to $x$.
\label{mod-def}
\end{definition}
We also need the following lemma.
\begin{lemma}
If $a,b,\mu\in\reals$, then
$$a\bmod\mu=b\bmod\mu\iff(-a)\bmod\mu=(-b)\bmod\mu\,.$$
\begin{proof}
We first prove that the left-hand side implies the right-hand side. From Def.~\ref{mod-def}, it follows that
\begin{eqnarray}
a\bmod\mu=b\bmod\mu&\Rightarrow& a-\mu\llfloor\frac a\mu\rrfloor = b-\mu\llfloor\frac b\mu\rrfloor\nonumber\\
&\Rightarrow&a-b=\mu\llpar \llfloor\frac a\mu\rrfloor-\llfloor\frac b\mu\rrfloor\rrpar\,.\label{modlem1-eq1}
\end{eqnarray}
\begin{equation}
\llfloor\frac a\mu\rrfloor,\llfloor\frac b\mu\rrfloor\in\integers\Rightarrow \llpar \llfloor\frac a\mu\rrfloor-\llfloor\frac b\mu\rrfloor\rrpar\in\integers\,.
\label{modlem1-eq2}
\end{equation}
\begin{equation}
\text{\eqref{modlem1-eq1} and \eqref{modlem1-eq2}}\Rightarrow\exists i\in\integers:a-b=i\mu\Rightarrow-a-(-b)=-i\mu\,.
\label{modlem1-eq3} 
\end{equation}
Therefore the difference between $-a$ and $-b$ is an integer multiple of $\mu$. Therefore, $(-a)\bmod\mu=(-b)\bmod\mu$, which proves that
\begin{equation}
a\bmod\mu=b\bmod\mu\Rightarrow(-a)\bmod\mu=(-b)\bmod\mu\,.
\label{modlem1-eq4}
\end{equation}
From \eqref{modlem1-eq4}, it follows directly that
\begin{eqnarray}
(-a)\bmod\mu=(-b)\bmod\mu&\Rightarrow&(-(-a))\bmod\mu=(-(-b))\bmod\mu\nonumber\\
&\Rightarrow&a\bmod\mu=b\bmod\mu\,.
\label{modlem1-eq5}
\end{eqnarray}
Equations \eqref{modlem1-eq4} and \eqref{modlem1-eq5} prove the lemma.\qed
\end{proof}
\label{mod-lemma-1}
\end{lemma}

\section{Transformation classes}
\label{section-transformation-classes}

In \cite{MeredithArxivChapter12026}, we informally presented the concept of a {\em transformation class}, as well as the specific classes, $\FTwoSTR$, $\FTwoSTRMod{7}$ and $\FTwoSTRMod{12}$. In this section, we begin by formally defining the concept of a transformation class. We then provide formal definitions of a number of specific transformation classes that are particularly relevant to pattern matching and pattern discovery in music.

\subsection{Definition of transformation class}

Figure~\ref{RAVEL-HAYDN-GROUND-TRUTH-CPO-GT} shows a geometric representation of Ravel's {\em Menuet sur le nom d'Haydn\/} as a set of 2-dimensional points in which each point, $\langle t,p\rangle$, represents a tone in the score, where $t$ is the onset time of the tone in tatums (in this case, quavers) and $p$ is the chromatic pitch of the note (reproduced from \cite[Fig.~8 on p.~11]{MeredithArxivChapter12026}). We call such a geometric, point-set representation a {\em chromatic-pitch--onset\/} (or CPO) representation. The figure shows 31 transformed occurrences of the so-called ``HAYDN'' theme, B-A-D-D-G (see \cite{MeredithArxivChapter12026} for details).

\begin{sidewaysfigure}
\resizebox{1\linewidth}{!}{\includegraphics{"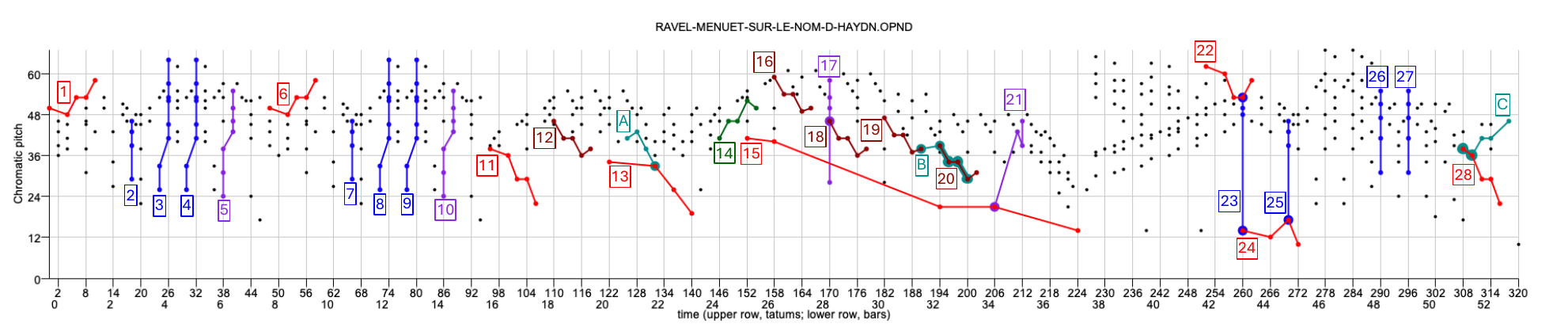"}}
\caption{A {\em chromatic-pitch--onset\/} (or CPO) representation of Ravel's {\em Menuet sur le nom d'Haydn}. Each point represents a distinct untied note or sequence of tied notes in the score. The $x$-coordinate of a point gives the onset time of the corresponding note in tatums (in this case, quavers) and the $y$--coordinate gives the chromatic pitch of the note. The occurrences of the HAYDN theme identified in Figs.~5 and 7 of \cite{MeredithArxivChapter12026} are shown, labelled with the same numbers and letters and coloured using the same colours as in those figures. When a point occurs in two patterns, it is drawn twice with different diameters and colours so that its presence in both patterns is visible. Reproduced from \cite[Fig.~8 on p.~11]{MeredithArxivChapter12026}.}
	\label{RAVEL-HAYDN-GROUND-TRUTH-CPO-GT}
\end{sidewaysfigure}

\input{RAVEL-HAYDN-GROUND-TRUTH-CPO-F2STRBook-IPTG.tex}

Table~\ref{RAVEL-HAYDN-GROUND-TRUTH-CPO-IPTG} (reproduced from \cite[p.~17]{MeredithArxivChapter12026}) shows the inter-pattern transformation graph for the transformation class, $\FTwoSTR$, and the patterns shown in Fig.~\ref{RAVEL-HAYDN-GROUND-TRUTH-CPO-GT}. 
$\FTwoSTR$ is the class of 2-dimensional transformations that consist of a scaling parallel to the $x$-axis by a real-valued scale factor, $f$, denoted by $S_f$, followed by a translation by a vector, $v$, denoted by $T_v$, optionally followed by a reflection in the $x$-axis, denoted by $R_x$.
An entry in this table in the row labelled $A$ and the column labelled $B$ gives the transformations within $\FTwoSTR$ that map pattern $A$ onto pattern $B$.\footnote{See \cite[Section~3.4, pp.~21 ff.]{MeredithArxivChapter12026} for details on inter-pattern transformation graphs.} This inter-pattern transformation graph was generated using an algorithm that, when given two patterns, $A$ and $B$, and a transformation class, $\tc$, can determine the transformations in $\tc$ (if any) that map $A$ onto $B$. In order to achieve this, a transformation class, $\tc$, must be defined in such a way that, given two patterns, $A$ and $B$, both of which are of some minimum size, $\beta(\tc)$, there must exist a finite upper limit on the number of transformations in $\tc$ that map $A$ onto $B$.

Suppose we have a transformation class, $\tc$, and two {\em ordered\/} sets of points, $\basis_1$ and $\basis_2$, such that $|\basis_1| = |\basis_2| = \len(\basis_1)=\len(\basis_2)=\basisSize(\tc)$, and there exists at least one transformation, $\transformation\in\tc$, such that $\basis_2[i]=\transformation(\basis_1[i])$ for all $0\le i<\basisSize(\tc)$. In addition, we require that all the points in $\basis_1$ are distinct and that all the points in $\basis_2$ are distinct. For a given $\tc$, $\basisSize(\tc)$ is the minimum value for which there is a finite maximum number, $\tau(\tc)$, of distinct transformations, $\transformation\in\tc$, for which this is true. For example, for the transformation class, $\FTwoT$, which consists of the set of all 2-dimensional translations, $\basisSize(\FTwoT)=1$ and $\tau(\FTwoT)=1$, since there is only one translation vector that maps any given single point onto any other single point. 

On the other hand, if we define that $\FTwoTR$ is the class of 2-dimensional transformations that consist of a translation, optionally followed by a reflection in the $x$-axis, then, given two distinct points, $p$ and $q$, there are at most {\em two\/} transformations, $g$, in $\FTwoTR$, such that $q=g(p)$---one that simply translates $p$ by the vector $q-p$, and a second that translates $p$ by $\langle q[0]-p[0],-p[1]-q[1]\rangle$ and then reflects in the $x$-axis (see Fig.~\ref{inter-point-transformations-example-1}). Therefore, $\beta(\FTwoTR)=1$ and $\tau(\FTwoTR)=2$.

\begin{figure}[t]
	\begin{center}
		{\resizebox{5cm}{!}{
			\begin{tikzpicture}
			
			\draw[thick,-{>[length=3mm,width=4mm]}] (-.1,4) -- (9.5,4);
			\draw[thick,{<[length=3mm,width=4mm]}-{>[length=3mm,width=4mm]}] (0,-.5) -- (0,8.5);		
			
			\draw (10,4) node{\LARGE $x$};
			\draw (-1,8) node{\LARGE $y$};
			
			\foreach \x in {1,...,8} {
				\draw[very thin,gray] (\x,0) -- (\x,8);
				\draw[thick] (\x,-.1) -- (\x,0);
				\draw[very thin,gray] (0,\x) -- (9,\x);
				\draw[thick] (-.1,\x) -- (0,\x);
			}
			\draw[very thin, gray](0,0) -- (9,0);
			\draw[very thin, gray](9,0) -- (9,8);
			
			\foreach \x in {-4,-3,...,4} {
				\draw (-.2,\x+4) node[anchor=east]{\x};
			}
			\foreach \x in {0,1,...,9} {
				\draw (\x+.1,4) node[anchor=north]{\x};
			}

			\filldraw  (1,5) circle[radius=.1]
				(7,7) circle[radius=.1]
				(7,1) circle[radius=.1];

			\draw (1,5) node[anchor=south west]{\LARGE $a$};
			\draw (7,7) node[anchor=south west]{\LARGE $b$};
			\draw (7,1) node[anchor=south west]{\LARGE $c$};
			
			 \draw[-{Stealth[length=7mm,width=4mm]}] (1,5) -- (7,7);
			 \draw (4,6) node[anchor=north west]{\LARGE $T_{\langle6,2\rangle}$};

			 \draw[-{Stealth[length=7mm,width=4mm]}] (1,5) -- (7,1);
			 \draw (4.4,3) node[anchor=north east]{\LARGE $T_{\langle6,-4\rangle}$};

			 \draw[-{Stealth[length=7mm,width=4mm]}] (7,1) .. controls (8,4) .. (7,7);
			 \draw (6.9,4.5) node[anchor=west]{\LARGE $R_x$};


%
%
%
%
%
%
%
%
%
%
			
			\end{tikzpicture}
		}}
	\end{center}
	\caption{There is only one transformation in $\FTwoT$ that maps point $a$ onto point $b$---a translation by $\langle6,2\rangle$. However, there are two transformations in $\FTwoTR$ that map point $a$ onto point $b$---a translation by $\langle6,2\rangle$, or a translation by $\langle6,-4\rangle$, which takes $a$ to $c$, followed by a reflection in the $x$-axis. The number of transformations in $\FTwoSTR$ that map point $a$ onto point $b$ is infinite, as there are two transformations for every possible $x$-axis scale factor.}
	\label{inter-point-transformations-example-1}
\end{figure}
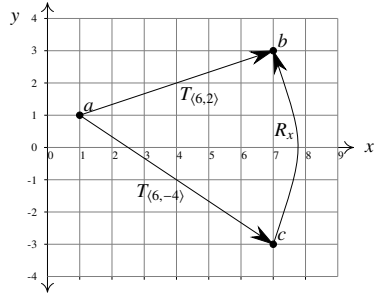

However, given $p,q\in\reals^2$,  the number of transformations in $\FTwoSTR$ that map $p$ onto $q$ is infinite; whereas, if we define that $\basis_1 = \langle p,q\rangle$ and $\basis_2 = \langle p',q'\rangle$, then there are at most two transformations that map $\basis_1$ onto $\basis_2$, so $\beta(\FTwoSTR)=2$ and $\tau(\FTwoSTR)=2$.

We now present a formal definition of the concept of a transformation class.

\begin{definition}[Transformation Class]
A set of functions, $\tc\subset \lbrace\transformation: \left(\transformation:\reals^k\rightarrow\reals^k\right)\rbrace$, is a {\em transformation class}, if and only if the following constraints are satisfied:
\begin{enumerate}
\item it is possible to define a finite integer, $\lambda(\tc)>0$, such that each function, $\transformation\in\tc$, can be uniquely identified by a {\em transformation parameter vector}, $\sigma_{\tc}(\transformation)\in\reals^{\lambda(\tc)}$; 
\item it is possible to define a {\em transformation class function}, $\phi_{\tc}:\reals^{\lambda(\tc)}\times\reals^k\rightarrow\reals^k$, such that $\phi_{\tc}(\sigma_{\tc}(\transformation),\point)=\transformation(\point)$ for all $\point\in\reals^k$ and all $\transformation\in\tc$;
\item it is possible to define finite integers, $\beta(\tc)>0$ and $\tau(\tc)>0$, along with a {\em basis-pair transformation set function}, 
$$\BTSF_{\tc}:\left(
			\left(
				\reals^k
			\right)^{
				\beta(\tc)
			}
		       \right)^2
		       \rightarrow
		       2^{
		       	\left(
				\reals^{\lambda(\tc)}
			\right)
		       }\,,$$ 
such that, given $\langle\basis_1,\basis_2\rangle\in\left(\left(\reals^k\right)^{\beta(\tc)}\right)^2$,
$$\BTSF_{\tc}(\basis_1,\basis_2)=\lbrace\sigma:\basis_2[i]=\phi_{\tc}(\sigma,\basis_1[i])\text{ for all } 0\le i  < \beta(\tc)\rbrace\,$$
and $\beta(\tc)$ is the minimum value for which there exists a finite integer, $\tau(\tc)$, such that
$$
|\BTSF_{\tc}(\basis_1,\basis_2)|\le\tau(\tc) \text{ for all } (\basis_1,\basis_2)
$$
and 
$$
|\BTSF_{\tc}(\basis_1,\basis_2)|=\tau(\tc) \text{ for at least one pair, } (\basis_1,\basis_2)\,.
$$
\end{enumerate}
\label{transformation-class-definition}
\end{definition}
In this definition, $\lambda(F)$ is called the {\em parameter vector length\/} of $F$. \label{basis-definition} $\basis_1$ and $\basis_2$ are ordered sets of distinct points of length and cardinality $\beta(\tc)$ called the {\em object basis\/} and the {\em image basis}, respectively. $\beta(\tc)$ is therefore called the {\em basis size\/} of $\tc$. Given a pair of bases, $\langle\basis_1,\basis_2\rangle$, for a transformation class, $\tc$, $\tau(\tc)$ is the maximum number of transformations within $\tc$ that map $\basis_1$ onto $\basis_2$. $k$ is the {\em dimensionality\/} of the transformation class, $\tc$, and $\tc$ can thus be referred to as a {\em $k$-dimensional transformation class}. For convenience, we define that, 
\begin{enumerate}
\item if $S$ is a set of points, then $\phi_F(\sigma, S)=\bigcup_{p\in S}\lbrace\phi_F(\sigma,p)\rbrace$; and
\item if $\mathbf A$ is an ordered set of points, then $\phi_F(\sigma, \mathbf A)=\bigconcat_{i=0}^{\len(\mathbf A)-1}\langle\phi_F(\sigma,\mathbf A[i])\rangle$.
\end{enumerate}
We also define the {\em universe of transformation parameter vectors\/} for a transformation class, $\tc$, denoted by $\allsigmas(\tc)$, as follows:
\begin{equation}
\allsigmas(\tc)=\left\lbrace\sigma:\left(\exists f\in\tc : \sigma_\tc(f)=\sigma\right)\right\rbrace\,.
\label{allsigmas-definition}
\end{equation}

\subsection{Examples of transformation classes}

We will now define a number of transformation classes that are of particular relevance to pattern matching and pattern discovery in music.

\subsubsection{The transformation class $\FTwoT$}
Let $\FTwoT$ be the set of 2-dimensional translations. A translation in a geometric, pitch--time representation such as those identified in \cite{MeredithArxivChapter12026}, corresponds to the musical transformation of a time-shift combined with a transposition. Each transformation, $f\in \FTwoT$, is a function, $f:\mathbb{R}^2\rightarrow\mathbb{R}^2$, that translates a point, $p\in\mathbb{R}^2$, by a  vector, $v\in\mathbb{R}^2$. That is, $f(p)=p+v$ where $p+v$ is the point that results from translating the point, $p$, by the vector, $v$. The vector $v$ thus uniquely determines $f$ within the transformation class, $\FTwoT$. 
We can therefore define that $\lambda(\FTwoT) = 2$ and $\sigma_{\FTwoT}(f)=v$ where $f(p)=p+v$. This satisfies the first constraint in Definition~\ref{transformation-class-definition}.

We can define the transformation class function for $\FTwoT$ as follows:
\begin{equation}
\phi_{\FTwoT}(\sigma_{\FTwoT}(f),p)=p+\sigma_{\FTwoT}(f)\,,
\label{ftwot-phi}
\end{equation}
where $\sigma_{\FTwoT}(f)\in\reals^2$ for all $f\in\FTwoT$.
This satisfies the second constraint in Definition~\ref{transformation-class-definition}.

We can define that $\beta(\FTwoT)=1$ and that $\tau(\FTwoT)=1$. The basis-pair transformation set function for $\FTwoT$ can then be defined as follows:
$$
\BTSF_{\FTwoT}(\basis_1,\basis_2)=\left\lbrace\basis_2[0]-\basis_1[0]\right\rbrace\,.
$$
That is, $\basis_1$ and $\basis_2$ each contain just one 2-dimensional point, and $\BTSF_{\FTwoT}(\basis_1,\basis_2)$ simply contains the vector from $\basis_1[0]$ to $\basis_2[0]$. This satisfies the third constraint in Definition~\ref{transformation-class-definition} and $\FTwoT$ is therefore a transformation class.


\subsubsection{The transformation class $\FTwoTR$}
Let $\FTwoTR$ be the set containing those transformations that consist of a 2-dimensional translation, optionally followed by a reflection in the $x$-axis. If we use one of the geometric, pitch--time representations introduced in \cite{MeredithArxivChapter12026}, then $\FTwoTR$ contains the geometric representations that correspond to combinations of the musical transformations of transposition and inversion. Each transformation, $f\in\FTwoTR$, first translates a point, $p$,  by a vector, $v$ to give a new point, $p'=p+v$. It then may or may not reflect $p'$ in the $x$-axis to get the final resulting point $q$. Reflecting the point, $p'$, in the $x$-axis is equivalent to multiplying its $y$-coordinate by $-1$. If $f$ does not reflect $p'$, then we can simply multiply the $y$-coordinate by $1$.

\sloppypar We can define that $\lambda(\FTwoTR)=3$, so that each transformation, $f\in\FTwoTR$, is uniquely identified by a transformation parameter vector, $\sigma_{\FTwoTR}(f)\in\reals^3$, where $\langle\sigma_{\FTwoTR}(f)[0],\sigma_{\FTwoTR}(f)[1]\rangle$ is equal to the vector, $v$, by which $f$ translates its argument point, $p$, to obtain $p'$. We can then specify that $\sigma_{\FTwoTR}(f)[2]\in\lbrace1,-1\rbrace$ is the factor by which the $y$-coordinate of $p'$ is multiplied to give the final resulting point, $q$, with $\sigma_{\FTwoTR}(f)[2]=-1$ if $p'$ is reflected in the $x$-axis and $\sigma_{\FTwoTR}(f)[2]=1$ if it is not. This satisfies the first constraint in Definition~\ref{transformation-class-definition}.

We can define the transformation class function for $\FTwoTR$ as follows:
\begin{equation}
\phi_{\FTwoTR}(\sigma,p)=\langle p[0]+\sigma[0],\sigma[2](p[1]+\sigma[1])\rangle
\label{phi_F2TR}
\end{equation}
where $p\in\reals^2$ and there exists an $f\in\FTwoTR$ such that $\sigma=\sigma_{\FTwoTR}(f)$. That is, $\phi_{\FTwoTR}(\sigma,p)$ is only defined when $\langle\sigma[0], \sigma[1]\rangle\in \reals^2$ and $\sigma[2]\in\lbrace1,-1\rbrace$. This satisfies the second constraint in Definition~\ref{transformation-class-definition}.

For $\FTwoTR$, we can set $\beta(\FTwoTR)=1$ and $\tau(\FTwoTR)=2$ and then define the basis-pair transformation set function as follows:
\begin{equation}
\BTSF_{\FTwoTR}(\langle p\rangle,\langle q\rangle)=\lbrace\langle q[0]-p[0],q[1]-p[1],1\rangle,\langle q[0]-p[0],-q[1]-p[1],-1\rangle\rbrace\,.
\label{BTSF-for-F2TR}
\end{equation}
It follows that $\FTwoTR$ is a transformation class.

As an example, let us compute the transformations in $\FTwoTR$ that map point $a$ onto point $b$ in Fig.~\ref{inter-point-transformations-example-1}. In this example, $a=\langle1,1\rangle$ and $b= \langle7,3\rangle$. Applying Eq.~\eqref{BTSF-for-F2TR}, we therefore obtain
\begin{eqnarray}
\footnotesize
\BTSF_{\FTwoTR}(\langle a\rangle,\langle b\rangle)&=&\lbrace\langle b[0]-a[0],b[1]-a[1],1\rangle,\langle b[0]-a[0],-b[1]-a[1],-1\rangle\rbrace\nonumber\\
&=&\lbrace\langle 7-1,3-1,1\rangle,\langle 7-1,-3-1,-1\rangle\rbrace\nonumber\\
&=&\lbrace\langle 6,2,1\rangle,\langle 6,-4,-1\rangle\rbrace\nonumber
\end{eqnarray}
This tells us that the transformations in $\FTwoTR$ that map point $a$ onto point $b$ in Fig.~\ref{inter-point-transformations-example-1} are those whose transformation parameter vectors are $\langle 6,2,1\rangle$ and $\langle 6,-4,-1\rangle$. If we set $\sigma=\langle 6,2,1\rangle$, then with Eq.~\ref{phi_F2TR} we can confirm that
\begin{eqnarray}
\phi_{\FTwoTR}(\sigma,a)&=&\langle a[0]+\sigma[0],\sigma[2](a[1]+\sigma[1])\rangle\nonumber\\
&=&\langle 1+6,1\times(1+2)\rangle\nonumber\\
&=&\langle 7,3\rangle=b\nonumber
\end{eqnarray}
We can then set $\sigma=\langle 6,-4,-1\rangle$ and confirm that
\begin{eqnarray}
\phi_{\FTwoTR}(\sigma,a)&=&\langle a[0]+\sigma[0],\sigma[2](a[1]+\sigma[1])\rangle\nonumber\\
&=&\langle 1+6,-1\times(1-4)\rangle\nonumber\\
&=&\langle 7,3\rangle=b\nonumber
\end{eqnarray}

\subsubsection{The transformation class $\FTwoST$}
\label{ftwost-section}
Let $\FTwoST$ be the set containing those 2-dimensional transformations that consist of a scaling parallel to the $x$-axis by a real-valued scale factor, followed by a translation. 
If we use one of the geometric, pitch--time representations introduced in \cite{MeredithArxivChapter12026}, then $\FTwoST$ includes those geometric transformations that correspond to combinations of the musical transformations of transposition, diminution, augmentation and retrograde. Each transformation, $f\in\FTwoST$, therefore takes a point, $p=\langle x_p, y_p\rangle$, and multiplies the $x$-coordinate of $p$ by a scale factor, $s$, to get a new point, $p'=\langle sx_p,y_p\rangle$. $p'$ is then translated by a 2-dimensional vector, $v= \langle x_v, y_v\rangle$, to give a new point, $q=p'+v= \langle sx_p+x_v,y_p+y_v\rangle$.

We can therefore define that $\lambda(\FTwoST)= 3$ and that each transformation, $f\in\FTwoST$, can be uniquely identified by a 3-vector, $\langle s, v_x, v_y\rangle\in\reals^3$, where $s$, $v_x$ and $v_y$ can take any real values. In other words, if $\sigma_{\FTwoST}(f)=\langle s, v_x, v_y\rangle$, then $f$ is the transformation that scales parallel to the $x$-axis by a scale factor of $s$, then translates by the vector, $\langle v_x, v_y\rangle$. This satisfies the first constraint in Definition~\ref{transformation-class-definition}. As an example, the points, $p_1$ and $p_2$, in Fig.~\ref{f2st-example-a} are mapped onto points, $q_1$ and $q_2$, respectively, by the transformation in $\FTwoST$ whose parameter vector is $\langle2,2,1\rangle$.

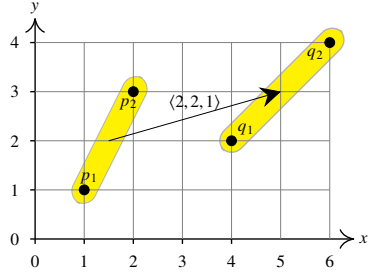
\begin{figure}[t]
	\begin{center}
		{\resizebox{5cm}{!}{
			\begin{tikzpicture}
			\filldraw [rounded corners=8pt, fill=yellow, draw opacity=0.3, cm={cos(-63.4) ,-sin(-63.4) ,sin(-63.4) ,cos(-63.4) ,(1.1,.63)}] (0,0) rectangle (2.85, .5);
			
\filldraw [thick, rounded corners=8pt, fill=yellow, draw opacity=0.3, cm={cos(-45) ,-sin(-45) ,sin(-45) ,cos(-45) ,(4,1.65)}] (0,0) rectangle (3.4, .5);
			
			\draw[thick,-{>[length=3mm,width=4mm]}] (-.1,0) -- (6.5,0);
			\draw[thick,-{>[length=3mm,width=4mm]}] (0,-.1) -- (0,4.5);		
			
			\draw (6.7,0) node{$x$};
			\draw (0,4.7) node{$y$};
			
			\foreach \x in {1,...,4} {
				\draw[very thin,gray] (\x,0) -- (\x,4);
				\draw[thick] (\x,-.1) -- (\x,0);
				\draw[very thin,gray] (0,\x) -- (6,\x);
				\draw[thick] (-.1,\x) -- (0,\x);
			}
			\draw[very thin, gray](5,0) -- (5,4);
			\draw[very thin, gray](6,0) -- (6,4);
			
			\foreach \x in {0,1,...,4} {
				\draw (-.2,\x) node[anchor=east]{\x};
			}
			\foreach \x in {0,1,...,6} {
				\draw (\x,-.2) node[anchor=north]{\x};
			}

			\filldraw  (1,1) circle[radius=.1]
				(2,3) circle[radius=.1]
				(4,2) circle[radius=.1]
				(6,4) circle[radius=.1];

			\draw (1.1,1.25) node[]{$p_1$};
			\draw (1.9,2.75) node[]{$p_2$};
			\draw (4,2) node[anchor=south west]{$q_1$};
			\draw (6,4) node[anchor=north east]{$q_2$};
			
			 \draw[-{Stealth[length=4mm,width=4mm]}] (1.5,2) -- (5,3) node[midway, above] {$\langle2,2,1\rangle$};

%
			\end{tikzpicture}
		}}
	\end{center}
	\caption{The only transformation in $\FTwoST$ that maps the ordered pair of points, $\langle p_1,p_2\rangle$, onto the ordered pair of points, $\langle q_1,q_2\rangle$, is the one with the parameter vector, $\langle2,2,1\rangle$. }
	\label{f2st-example-a}
\end{figure}

We can define the transformation class function for $\FTwoST$ as follows,
\begin{equation}
\phi_{\FTwoST}(\sigma,p)=\langle p[0]\sigma[0]+\sigma[1],p[1]+\sigma[2]\rangle\,,
\label{eqn-f2st}
\end{equation}
where $p\in\reals^2$ and there must exist a transformation, $f\in\FTwoST$, such that $\sigma=\sigma_{\FTwoST}(f)$. This satisfies the second constraint in Definition~\ref{transformation-class-definition}.

For $\FTwoST$, we can define that $\beta(\FTwoST)=2$ and that $\tau(\FTwoST)=1$. The basis-pair transformation set function, $\BTSF_{\FTwoST}$, therefore takes any two ordered pairs of distinct points, $\objectBasis$ and $\imageBasis$, and computes the transformation parameter vectors of the transformations in $\FTwoST$ that map $\objectBasis$ onto $\imageBasis$. In the case of $\FTwoST$, there is only ever at most one such transformation for any given basis pair. Note that, given any two points, $p,q\in\mathbb{R}^2$, there is an infinite number of transformations, $f\in\FTwoST$, that map $p$ onto $q$, since, for every possible value of  the $x$-axis scale factor, $\sigma_{\FTwoST}(f)[0]$, it is possible to find a vector, $\left\langle\sigma_{\FTwoST}(f)[1],\sigma_{\FTwoST}(f)[2]\right\rangle$, that results in $p$ mapping onto $q$. We therefore need to set $\beta(\FTwoST)$ to $2$, which results in  $\tau(\FTwoST)$ being $1$.

\begin{figure}[t]
	\begin{center}
		{\resizebox{5cm}{!}{
			\begin{tikzpicture}
			\filldraw [rounded corners=8pt, fill=yellow, draw opacity=0.3, cm={cos(-63.4) ,-sin(-63.4) ,sin(-63.4) ,cos(-63.4) ,(1.1,.63)}] (0,0) rectangle (2.85, .5);

\filldraw [rounded corners=8pt, fill=yellow, draw opacity=0.3, cm={cos(-56.3) ,-sin(-56.3) ,sin(-56.3) ,cos(-56.3) ,(4.05,.65)}] (0,0) rectangle (4.2, .5);
			
			\draw[thick,-{>[length=3mm,width=4mm]}] (-.1,0) -- (6.5,0);
			\draw[thick,-{>[length=3mm,width=4mm]}] (0,-.1) -- (0,4.5);		
			
			\draw (6.7,0) node{$x$};
			\draw (0,4.7) node{$y$};
			
			\foreach \x in {1,...,4} {
				\draw[very thin,gray] (\x,0) -- (\x,4);
				\draw[thick] (\x,-.1) -- (\x,0);
				\draw[very thin,gray] (0,\x) -- (6,\x);
				\draw[thick] (-.1,\x) -- (0,\x);
			}
			\draw[very thin, gray](5,0) -- (5,4);
			\draw[very thin, gray](6,0) -- (6,4);
			
			\foreach \x in {0,1,...,4} {
				\draw (-.2,\x) node[anchor=east]{\x};
			}
			\foreach \x in {0,1,...,6} {
				\draw (\x,-.2) node[anchor=north]{\x};
			}

			\filldraw  (1,1) circle[radius=.1]
				(2,3) circle[radius=.1]
				(4,1) circle[radius=.1]
				(6,4) circle[radius=.1];

			\draw (1.1,1.25) node[]{$p_1$};
			\draw (1.9,2.75) node[]{$p_2$};
			\draw (4.2,1.2) node[]{$q_1$};
			\draw (5.8,3.8) node[]{$q_2$};
			

%
			\end{tikzpicture}
		}}
	\end{center}
	\caption{There is no transformation in $\FTwoST$ that maps the ordered pair of points, $\langle p_1,p_2\rangle$, onto the ordered pair of points, $\langle q_1,q_2\rangle$ because $q_2[1]-q_1[1]\neq p_2[1]-p_1[1]$. }
	\label{f2st-example-b}
\end{figure}
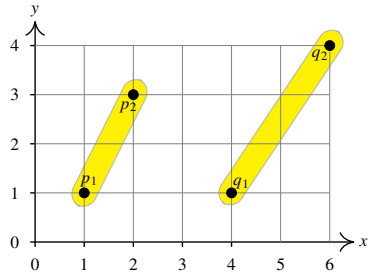

 Let us suppose that $\objectBasis=\langle p_1, p_2 \rangle$ and that  $\imageBasis=\langle q_1, q_2 \rangle$. If $\sigma\in\BTSF_{\FTwoST}(\objectBasis,\imageBasis)$, then we know from \eqref{eqn-f2st} that 
\begin{eqnarray}
q_1[0]&=&p_1[0]\sigma[0]+\sigma[1]\,,\label{one-F2ST}\\
q_1[1]&=&p_1[1]+\sigma[2]\,,\label{two-F2ST}\\
q_2[0]&=&p_2[0]\sigma[0]+\sigma[1]\textrm{, and }\label{three-F2ST}\\
q_2[1]&=&p_2[1]+\sigma[2]\,.\label{four-F2ST}
\end{eqnarray}
We can immediately make the following inferences:
\begin{eqnarray}
\eqref{two-F2ST}&\Rightarrow& \sigma[2]=q_1[1]-p_1[1]\textrm{, and}\label{nine-F2ST}\\
\eqref{four-F2ST}&\Rightarrow& \sigma[2]=q_2[1] - p_2[1]\,.\label{ten-F2ST}
\end{eqnarray}
Equations \eqref{nine-F2ST} and \eqref{ten-F2ST} together imply that, in order for there to be a transformation in $\FTwoST$ that maps $\objectBasis$ onto $\imageBasis$, it must be true that 
\begin{equation}
q_1[1]-p_1[1]=q_2[1] - p_2[1]\,.
\label{y-difference-equal-F2ST}
\end{equation}
Thus, for example, there is no transformation in $\FTwoST$ that maps $\langle p_1,p_2\rangle$ onto $\langle q_1,q_2\rangle$ in Fig.~\ref{f2st-example-b} because the condition in \eqref{y-difference-equal-F2ST} does not hold in this case.

Equations \eqref{one-F2ST} and \eqref{three-F2ST} together imply that
\begin{eqnarray}
q_2[0]&=&p_2[0]\sigma[0]+q_1[0]-p_1[0]\sigma[0]\nonumber\\
&=&\sigma[0](p_2[0]-p_1[0])+q_1[0]\nonumber\\
\Rightarrow \sigma[0]&=&\frac{q_2[0]-q_1[0]}{p_2[0]-p_1[0]}\,.\label{eleven-F2ST}
\end{eqnarray}
Equation \eqref{eleven-F2ST} implies that $\sigma[0]= 0$ when $q_2[0]=q_1[0]$, unless $p_2[0]=p_1[0]$, in which case the right-hand side of \eqref{eleven-F2ST} evaluates to $\frac00$, which is undefined. However, we define that, in this case, $\sigma[0]$ should be set to 1, as this allows us to intuitively discover that, if $q_2[0]=q_1[0]$ and $p_2[0]=p_1[0]$ and condition \eqref{y-difference-equal-F2ST} holds, then $\objectBasis$ can be mapped onto $\imageBasis$ by a simple translation (i.e., no scaling required). Clearly, if $p_2[0]=p_1[0]$ and $q_2[0]\neq q_1[0]$, then $\sigma[0]$ is undefined and there is no transformation in $\FTwoST$ that maps $\objectBasis$ onto $\imageBasis$. 

\begin{figure}[t]
	\begin{center}
		{\resizebox{5cm}{!}{
			\begin{tikzpicture}
			\filldraw [rounded corners=8pt, fill=yellow, draw opacity=0.3, cm={cos(-63.4) ,-sin(-63.4) ,sin(-63.4) ,cos(-63.4) ,(1.1,.63)}] (0,0) rectangle (2.85, .5);
			
\filldraw [rounded corners=8pt, fill=yellow, draw opacity=0.3] (4.75,1.75) rectangle (5.25,4.25);
\filldraw [rounded corners=8pt, fill=yellow!80!white, draw opacity=0.3] (-.25,.75) rectangle (.25,3.25);
			
			\draw[thick,-{>[length=3mm,width=4mm]}] (-.1,0) -- (6.5,0);
			\draw[thick,-{>[length=3mm,width=4mm]}] (0,-.1) -- (0,4.5);		
			
			\draw (6.7,0) node{$x$};
			\draw (0,4.7) node{$y$};
			
			\foreach \x in {1,...,4} {
				\draw[very thin,gray] (\x,0) -- (\x,4);
				\draw[thick] (\x,-.1) -- (\x,0);
				\draw[very thin,gray] (0,\x) -- (6,\x);
				\draw[thick] (-.1,\x) -- (0,\x);
			}
			\draw[very thin, gray](5,0) -- (5,4);
			\draw[very thin, gray](6,0) -- (6,4);
			
			\foreach \x in {0,1,...,4} {
				\draw (-.2,\x) node[anchor=east]{\x};
			}
			\foreach \x in {0,1,...,6} {
				\draw (\x,-.2) node[anchor=north]{\x};
			}

			\filldraw  (1,1) circle[radius=.1]
				(2,3) circle[radius=.1]
				(5,2) circle[radius=.1]
				(5,4) circle[radius=.1];
			\filldraw[gray] (0,1) circle[radius=.1]
				(0,3) circle[radius=.1];

			\draw (1.1,1.25) node[]{$p_1$};
			\draw (1.9,2.75) node[]{$p_2$};
			\draw[gray] (0.25,1.25) node[]{$p_1'$};
			\draw[gray] (0.25,2.75) node[]{$p_2'$};
			\draw (5,2.25) node[]{$q_1$};
			\draw (5,3.75) node[]{$q_2$};
			
			 \draw[-{Stealth[length=4mm,width=4mm]}] (1.5,2) -- (5,3) node[midway, above] {$\langle0,5,1\rangle$};

%
			\end{tikzpicture}
		}}
	\end{center}
	\caption{$\langle p_1,p_2\rangle$ can be mapped onto $\langle q_1,q_2\rangle$ by a transformation in $\FTwoST$ by first scaling parallel to the $x$-axis with a scale factor of zero, to give $\langle p_1',p_2'\rangle$, and then translating $\langle p_1',p_2'\rangle$ by $\langle5,1\rangle$. This transformation has the parameter vector, $\langle0,5,1\rangle$.}
	\label{f2st-example-c}
\end{figure}
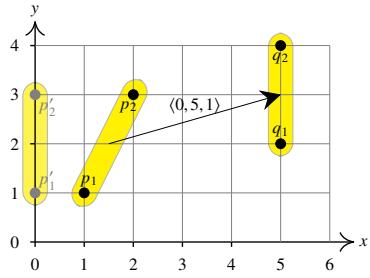

This is illustrated by the examples in Figs.~\ref{f2st-example-c} and \ref{f2st-example-d}. In Fig.~\ref{f2st-example-c}, $\langle p_1,p_2\rangle$ is mapped onto $\langle q_1,q_2\rangle$ by the transformation in $\FTwoST$ whose parameter vector is $\langle0,5,1\rangle$---that is, an $x$-axis scaling with a scale factor of zero, which maps $\langle p_1,p_2\rangle$ to $\langle p_1',p_2'\rangle$, followed by a translation of $\langle5,1\rangle$. In Fig.~\ref{f2st-example-d}, $\sigma[0]=\frac00$, so we define that $\sigma[0]=1$, implying that $\langle p_1,p_2\rangle$ is mapped onto $\langle q_1,q_2\rangle$ in this case by the transformation whose parameter vector is $\langle1,4,1\rangle$.

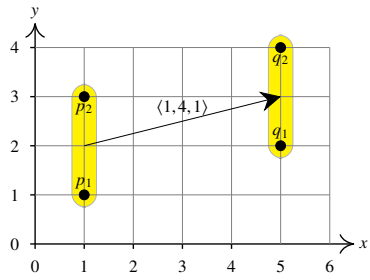
\begin{figure}[t]
	\begin{center}
		{\resizebox{5cm}{!}{
			\begin{tikzpicture}
			
\filldraw [rounded corners=8pt, fill=yellow, draw opacity=0.3] (4.75,1.75) rectangle (5.25,4.25);
\filldraw [rounded corners=8pt, fill=yellow, draw opacity=0.3] (.75,.75) rectangle (1.25,3.25);
			
			\draw[thick,-{>[length=3mm,width=4mm]}] (-.1,0) -- (6.5,0);
			\draw[thick,-{>[length=3mm,width=4mm]}] (0,-.1) -- (0,4.5);		
			
			\draw (6.7,0) node{$x$};
			\draw (0,4.7) node{$y$};
			
			\foreach \x in {1,...,4} {
				\draw[very thin,gray] (\x,0) -- (\x,4);
				\draw[thick] (\x,-.1) -- (\x,0);
				\draw[very thin,gray] (0,\x) -- (6,\x);
				\draw[thick] (-.1,\x) -- (0,\x);
			}
			\draw[very thin, gray](5,0) -- (5,4);
			\draw[very thin, gray](6,0) -- (6,4);
			
			\foreach \x in {0,1,...,4} {
				\draw (-.2,\x) node[anchor=east]{\x};
			}
			\foreach \x in {0,1,...,6} {
				\draw (\x,-.2) node[anchor=north]{\x};
			}

			\filldraw  (1,1) circle[radius=.1]
				(1,3) circle[radius=.1]
				(5,2) circle[radius=.1]
				(5,4) circle[radius=.1];

			\draw (1,1.25) node[]{$p_1$};
			\draw (1,2.75) node[]{$p_2$};
			\draw (5,2.25) node[]{$q_1$};
			\draw (5,3.75) node[]{$q_2$};
			
			 \draw[-{Stealth[length=4mm,width=4mm]}] (1,2) -- (5,3) node[midway, above] {$\langle1,4,1\rangle$};

%
			\end{tikzpicture}
		}}
	\end{center}
	\caption{In this example, the value of $\sigma[0]$, computed using \eqref{eleven-F2ST}, is $\frac00$, so it is set to 1. $\langle p_1,p_2\rangle$ can therefore be mapped onto $\langle q_1,q_2\rangle$ by the transformation in $\FTwoST$ whose parameter vector is $\langle1,4,1\rangle$.}
	\label{f2st-example-d}
\end{figure}

From \eqref{three-F2ST}, it follows that 
\begin{equation}
\sigma[1]=q_2[0]-p_2[0]\sigma[0]\,,
\label{six-F2ST}
\end{equation}
and \eqref{one-F2ST} implies that 
\begin{equation}
\sigma[1]=q_1[0]-p_1[0]\sigma[0]\,.
\label{sixb-F2ST}
\end{equation}
From \eqref{eleven-F2ST} and \eqref{six-F2ST}, we can infer that
\begin{eqnarray}
q_2[0]-p_2[0]\sigma[0]&=&q_2[0]-p_2[0]\left(\frac{q_2[0]-q_1[0]}{p_2[0]-p_1[0]}\right)\nonumber\\
&=&\frac{q_2[0]p_2[0]-q_2[0]p_1[0]-p_2[0]q_2[0]+p_2[0]q_1[0]}{p_2[0]-p_1[0]}\nonumber\\
&=&\frac{p_2[0]q_1[0]-q_2[0]p_1[0]}{p_2[0]-p_1[0]}\,.
\label{sigma1-rhs-are-equal-a-F2ST}
\end{eqnarray}
Similarly, from \eqref{eleven-F2ST} and \eqref{sixb-F2ST}, we can infer that
\begin{eqnarray}
q_1[0]-p_1[0]\sigma[0]&=&q_1[0]-p_1[0]\left(\frac{q_2[0]-q_1[0]}{p_2[0]-p_1[0]}\right)\nonumber\\
&=&\frac{p_2[0]q_1[0]-p_1[0]q_1[0]-p_1[0]q_2[0]+p_1[0]q_1[0]}{p_2[0]-p_1[0]}\nonumber\\
&=&\frac{p_2[0]q_1[0]-p_1[0]q_2[0]}{p_2[0]-p_1[0]}\,.
\label{sigma1-rhs-are-equal-b-F2ST}
\end{eqnarray}
Equations \eqref{sigma1-rhs-are-equal-a-F2ST} and \eqref{sigma1-rhs-are-equal-b-F2ST} together imply that the right-hand sides of \eqref{six-F2ST} and \eqref{sixb-F2ST} are always equal. It follows that, when $\sigma[0]$ is undefined (i.e., when $p_2[0]=p_1[0]$ and $q_2[0]\neq q_1[0]$), $\sigma[1]$ is also undefined. 

The basis-pair transformation set function for $\FTwoST$ can therefore be defined as follows:
\begin{equation}
\BTSF_{\FTwoST}\left(\langle p_1,p_2\rangle,\langle q_1,q_2\rangle\right)=\left\lbrace\left\langle\sigma[0],q_2[0]-p_2[0]\sigma[0],q_1[1]-p_1[1]\right\rangle\right\rbrace\,,
\end{equation}
where
$$q_1[1]-p_1[1]=q_2[1] - p_2[1]$$
and where
\begin{equation}
\sigma[0]=\begin{cases}
1,&\text{ if $p_2[0]=p_1[0]$ and $q_2[0]= q_1[0]$; and}\\
\frac{q_2[0]-q_1[0]}{p_2[0]-p_1[0]},&\text{otherwise, provided that }p_2[0]\neq p_1[0]\,.
\end{cases}
\end{equation}
$\FTwoST$ is therefore a transformation class. The basis-pair transformation set function, $\BTSF_{\FTwoST}$, can be computed using the algorithm shown in Fig.~\ref{FTwoST-basis-transformation-set-function}.

\begin{figure}[t]
\begin{center}
\resizebox*{.7\columnwidth}{!}{
\begin{minipage}{\linewidth}
\begin{tabbing}
\input{ColumnDefinitions}
$\fun{$\BTSF_{\FTwoST}$}(\objectBasis,\imageBasis)$\\
1\>$\dmif \imageBasis = \objectBasis$\\
2\>\>$\dmreturn\lbrace\langle1,0,0\rangle\rbrace\dmcomment{Identity transformation}$\\
3\>$p_1\becomes\objectBasis[0],p_2\becomes\objectBasis[1],q_1\becomes\imageBasis[0],q_2\becomes\imageBasis[1]$\\
4\>$\dmif q_2[0]\neq q_1[0]\land p_2[0]=p_1[0]\dmcomment{$\sigma[0]$ is undefined.}$\\
5\>\>$\dmreturn \lbrace\rbrace$\\
6\>$\dmif p_2[0]=p_1[0] \land q_2[0]= q_1[0]\dmcomment{$\sigma[0]=0/0$}$\\
7\>\>$\sigma_0\becomes1$\\
8\>$\dmelse$\\
9\>\>$\sigma_0\becomes(q_2[0]-q_1[0])/(p_2[0]-p_1[0])\dmcomment{From \eqref{eleven-F2ST}}$\\
10\>$\sigma_1\becomes q_2[0]-p_2[0]\sigma_0\dmcomment{From \eqref{six-F2ST}}$\\
11\>$\dmif q_1[1]-p_1[1]=q_2[1] - p_2[1]\dmcomment{From \eqref{y-difference-equal-F2ST}}$\\
12\>\>$\sigma_2\becomes q_1[1]-p_1[1]\dmcomment{From \eqref{nine-F2ST}}$\\
13\>\>$\dmreturn\lbrace\langle\sigma_0,\sigma_1,\sigma_2\rangle\rbrace$\\
14\>$\dmreturn\lbrace\rbrace$
\end{tabbing}
\end{minipage}
}
\end{center}
\caption{Algorithm for computing the basis-pair transformation set function for $\FTwoST$.}
\label{FTwoST-basis-transformation-set-function}
\end{figure}

\subsubsection{The transformation class $\FTwoSTR$}
\label{example-FTwoSTR}
Let $\FTwoSTR$ be the set of 2-dimensional transformations that consist of a scaling parallel to the $x$-axis by a real-valued scale factor, followed by a translation, optionally followed by a reflection in the $x$-axis. If we represent a piece of music as a set of points in pitch--time space, where the $y$-axis is a number representing pitch and the $x$-axis represents time, then this set of transformations corresponds to the class of traditional contrapuntal transformations in music, where a pattern can be transformed by any combination of transposition (corresponding to translating in the $y$-direction), time-shifting (translation in the $x$-direction), augmentation (scaling parallel to the $x$-axis by a scale factor whose absolute value is greater than 1), diminution (scaling parallel to the $x$-axis by a non-zero scale factor whose absolute value is less than 1), inversion (reflection in the $x$-axis) and retrograde (scaling parallel to the $x$-axis by a negative scale factor). 

Each transformation, $f\in\FTwoSTR$, therefore takes a point, $p=\langle x_p, y_p\rangle$, and multiplies the $x$-coordinate of $p$ by a scale factor, $s$, to get a new point, $p'=\langle sx_p,y_p\rangle$. $p'$ is then translated by a 2-dimensional vector, $v= \langle x_v, y_v\rangle$, to give a new point, $p''=p'+v= \langle sx_p+x_v,y_p+y_v\rangle$. Finally, $p''$ may be reflected in the $x$-axis to give a point, $q=\langle sx_p+x_v,b(y_p+y_v)\rangle$, where $b=-1$ if $p''$ is reflected and $b=1$ if it is not.

We can therefore define that $\lambda(\FTwoSTR)= 4$ and that each transformation $f\in\FTwoSTR$, can be uniquely identified by a 4-vector, $\langle s, v_x, v_y, b\rangle\in\reals^4$, where $s$, $v_x$ and $v_y$ can take any real values and $b\in\lbrace-1,1\rbrace$. In other words, if $\sigma_{\FTwoSTR}(f)=\langle s, v_x, v_y, b\rangle$, then $f$ is the transformation that scales parallel to the $x$-axis by a scale factor of $s$, then translates by the vector, $\langle v_x, v_y\rangle$ and then multiplies the $y$-coordinate of the resulting point by $b$ (which is either $-1$ or $1$). This satisfies the first constraint in Definition~\ref{transformation-class-definition}.

We can define the transformation class function for $\FTwoSTR$ as follows,
\begin{equation}
\phi_{\FTwoSTR}(\sigma,p)=\langle p[0]\sigma[0]+\sigma[1],\sigma[3](p[1]+\sigma[2])\rangle\,,
\label{eqn-f2str}
\end{equation}
where $p\in\reals^2$ and there must exist a transformation, $f\in\FTwoSTR$, such that $\sigma=\sigma_{\FTwoSTR}(f)$. This satisfies the second constraint in Definition~\ref{transformation-class-definition}.

For $\FTwoSTR$, we can define that $\beta(\FTwoSTR)=2$ and that $\tau(\FTwoSTR)=2$. The basis-pair transformation set function, $\BTSF_{\FTwoSTR}$, thus takes any two ordered pairs of distinct points, $\objectBasis$ and $\imageBasis$, and computes the transformation parameter vectors of the transformations in $\FTwoSTR$ that map $\objectBasis$ onto $\imageBasis$. Let us suppose that $\objectBasis=\langle p_1, p_2 \rangle$ and that  $\imageBasis=\langle q_1, q_2 \rangle$. If $\sigma\in\BTSF_{\FTwoSTR}(\objectBasis,\imageBasis)$, then we know from \eqref{eqn-f2str} that 
\begin{eqnarray}
q_1[0]&=&p_1[0]\sigma[0]+\sigma[1]\,,\label{one}\\
q_1[1]&=&\sigma[3](p_1[1]+\sigma[2])\,,\label{two}\\
q_2[0]&=&p_2[0]\sigma[0]+\sigma[1]\textrm{, and }\label{three}\\
q_2[1]&=&\sigma[3](p_2[1]+\sigma[2])\,.\label{four}
\end{eqnarray}
We know that $\sigma[3]$ is equal to either $1$ or $-1$. Let us first consider the case where $\sigma[3]=1$. In this case, 
\begin{eqnarray}
\eqref{two}&\Rightarrow& \sigma[2]=q_1[1]-p_1[1]\textrm{, and}\label{nine}\\
\eqref{four}&\Rightarrow& \sigma[2]=q_2[1] - p_2[1]\,.\label{ten}
\end{eqnarray}
Equations \eqref{nine} and \eqref{ten} together imply that, in order for there to be a transformation in $\FTwoSTR$ that maps $\objectBasis$ onto $\imageBasis$ when $\sigma[3]=1$, it must be true that 
\begin{equation}
q_1[1]-p_1[1]=q_2[1] - p_2[1]\,.
\label{y-difference-equal}
\end{equation}
This implies, for example, that, when $\sigma[3]=1$, there is no transformation in $\FTwoSTR$ that maps $\langle p_1,p_2\rangle$ onto $\langle q_1,q_2\rangle$ in Fig.~\ref{f2st-example-b}, because condition \eqref{y-difference-equal} does not hold in this case.

Equations \eqref{one} and \eqref{three} together imply that
\begin{eqnarray}
q_2[0]&=&p_2[0]\sigma[0]+q_1[0]-p_1[0]\sigma[0]\nonumber\\
&=&\sigma[0](p_2[0]-p_1[0])+q_1[0]\nonumber\\
\Rightarrow \sigma[0]&=&\frac{q_2[0]-q_1[0]}{p_2[0]-p_1[0]}\,.\label{eleven}
\end{eqnarray}
Equation \eqref{eleven} implies that $\sigma[0]= 0$ when $q_2[0]=q_1[0]$, unless $p_2[0]=p_1[0]$, in which case the right-hand side of \eqref{eleven} evaluates to $\frac00$, which is undefined. However, as with $\FTwoST$, we define that, in this case, $\sigma[0]$ should be set to 1, as this allows us to intuitively discover that, if $q_2[0]=q_1[0]$ and $p_2[0]=p_1[0]$ and condition \eqref{y-difference-equal} holds, then $\objectBasis$ can be mapped onto $\imageBasis$ by a simple translation (i.e., neither scaling nor reflection required). Clearly, if $p_2[0]=p_1[0]$ and $q_2[0]\neq q_1[0]$, then $\sigma[0]$ is undefined and there is consequently no transformation in $\FTwoSTR$ that maps $\objectBasis$ onto $\imageBasis$.

From \eqref{three}, it follows that 
\begin{equation}
\sigma[1]=q_2[0]-p_2[0]\sigma[0]\,,
\label{six}
\end{equation}
and \eqref{one} implies that 
\begin{equation}
\sigma[1]=q_1[0]-p_1[0]\sigma[0]\,.
\label{sixb}
\end{equation}
From \eqref{eleven} and \eqref{six}, we can infer that
\begin{eqnarray}
q_2[0]-p_2[0]\sigma[0]&=&q_2[0]-p_2[0]\left(\frac{q_2[0]-q_1[0]}{p_2[0]-p_1[0]}\right)\nonumber\\
&=&\frac{q_2[0]p_2[0]-q_2[0]p_1[0]-p_2[0]q_2[0]+p_2[0]q_1[0]}{p_2[0]-p_1[0]}\nonumber\\
&=&\frac{p_2[0]q_1[0]-q_2[0]p_1[0]}{p_2[0]-p_1[0]}\,.
\label{sigma1-rhs-are-equal-a}
\end{eqnarray}
Similarly, from \eqref{eleven} and \eqref{sixb}, we can infer that
\begin{eqnarray}
q_1[0]-p_1[0]\sigma[0]&=&q_1[0]-p_1[0]\left(\frac{q_2[0]-q_1[0]}{p_2[0]-p_1[0]}\right)\nonumber\\
&=&\frac{q_1[0]p_2[0]-q_1[0]p_1[0]-p_1[0]q_2[0]+p_1[0]q_1[0]}{p_2[0]-p_1[0]}\nonumber\\
&=&\frac{p_2[0]q_1[0]-p_1[0]q_2[0]}{p_2[0]-p_1[0]}\,.
\label{sigma1-rhs-are-equal-b}
\end{eqnarray}
Equations \eqref{sigma1-rhs-are-equal-a} and \eqref{sigma1-rhs-are-equal-b} together imply that the right-hand sides of \eqref{six} and \eqref{sixb} are always equal. It follows that, when $\sigma[0]$ is undefined (i.e., when $p_2[0]=p_1[0]$ and $q_2[0]\neq q_1[0]$), $\sigma[1]$ is also undefined. 

To summarise the result so far, if $\sigma[3]=1$, then
\begin{equation}
\sigma=\left\langle\sigma[0],q_2[0]-p_2[0]\sigma[0],q_1[1]-p_1[1],1\right\rangle\,,
\end{equation}
where
$$q_1[1]-p_1[1]=q_2[1] - p_2[1]\,,$$
and where
$$
\sigma[0]=\begin{cases}
1,&\text{ if $p_2[0]=p_1[0]$ and $q_2[0]= q_1[0]$; and}\\
\frac{q_2[0]-q_1[0]}{p_2[0]-p_1[0]},&\text{otherwise, provided that }p_2[0]\neq p_1[0]\,.
\end{cases}
$$

Let us now consider the case where $\sigma[3]=-1$. In this case, $\sigma[0]$ and $\sigma[1]$ are still given by \eqref{eleven} and \eqref{six}, respectively, since they do not depend on $\sigma[3]$. From \eqref{two} and \eqref{four} we conclude that 
\begin{eqnarray}
\eqref{two}&\Rightarrow& \sigma[2]=-p_1[1]-q_1[1]\textrm{, and}\label{seven}\\
\eqref{four}&\Rightarrow& \sigma[2]=-p_2[1]-q_2[1]\,.\label{eight}
\end{eqnarray}
Therefore, in order for there to be a transformation that maps $\objectBasis$ onto $\imageBasis$ when $\sigma[3]=-1$, it must be true that 
\begin{equation}
p_1[1]+q_1[1]=p_2[1] + q_2[1]\,.
\label{y-difference-equal-sigma-3-equals-1}
\end{equation}

\begin{figure}[t]
\begin{center}
\resizebox*{.7\columnwidth}{!}{
\begin{minipage}{\linewidth}
\begin{tabbing}
\input{ColumnDefinitions}
$\fun{$\BTSF_{\FTwoSTR}$}(\objectBasis,\imageBasis)$\\
1\>$\dmif \imageBasis = \objectBasis$\\
2\>\>$\dmreturn\lbrace\langle1,0,0,1\rangle\rbrace\dmcomment{Identity transformation}$\\
3\>$p_1\becomes\objectBasis[0],p_2\becomes\objectBasis[1],q_1\becomes\imageBasis[0],q_2\becomes\imageBasis[1]$\\
4\>$\dmif q_2[0]\neq q_1[0]\land p_2[0]=p_1[0]\dmcomment{$\sigma[0]$ is undefined.}$\\
5\>\>$\dmreturn \lbrace\rbrace$\\
6\>$\dmif p_2[0]=p_1[0] \land q_2[0]= q_1[0]\dmcomment{$\sigma[0]=0/0$}$\\
7\>\>$\sigma_0\becomes1$\\
8\>$\dmelse$\\
9\>\>$\sigma_0\becomes(q_2[0]-q_1[0])/(p_2[0]-p_1[0])\dmcomment{From \eqref{eleven}}$\\
10\>$S\becomes\lbrace\rbrace$\\
11\>$\sigma_1\becomes q_2[0]-p_2[0]\sigma_0\dmcomment{From \eqref{six}}$\\
12\>$\dmif q_1[1]-p_1[1]=q_2[1] - p_2[1]\dmcomment{From \eqref{y-difference-equal}}$\\
13\>\>$\sigma_3\becomes1$\\
14\>\>$\sigma_2\becomes q_1[1]-p_1[1]\dmcomment{From \eqref{nine}}$\\
15\>\>$S\becomes \lbrace\langle\sigma_0,\sigma_1,\sigma_2,\sigma_3\rangle\rbrace$\\
16\>$\dmif p_1[1]+q_1[1]=p_2[1] + q_2[1]\dmcomment{From \eqref{y-difference-equal-sigma-3-equals-1}}$\\
17\>\>$\sigma_3\becomes-1$\\
18\>\>$\sigma_2\becomes-p_1[1]-q_1[1]\dmcomment{From \eqref{seven}}$\\
19\>\>$S\becomes S\cup\lbrace\langle\sigma_0,\sigma_1,\sigma_2,\sigma_3\rangle\rbrace$\\
20\>$\dmreturn S$
\end{tabbing}
\end{minipage}
}
\end{center}
\caption{Algorithm for computing the basis-pair transformation set function for $\FTwoSTR$.}
\label{FTwoSTR-basis-transformation-set-function}
\end{figure}

To summarise, we can define the basis-pair transformation set function, $\BTSF_{\FTwoSTR}$, as follows:
\begin{multline}
\BTSF_{\FTwoSTR}\left(\langle p_1,p_2\rangle,\langle q_1,q_2\rangle\right)
=\\
\left\lbrace\left\langle\sigma[0],q_2[0]-p_2[0]\sigma[0],q_1[1]-p_1[1],1\right\rangle\right.,\\
\left.\left\langle\sigma[0],q_2[0]-p_2[0]\sigma[0],-p_1[1]-q_1[1],-1\right\rangle\right\rbrace\,,
\label{equation-BTSFFTwoSTR}
\end{multline}
where 
$$q_1[1]-p_1[1]=q_2[1] - p_2[1] \text{ when } \sigma[3]=1\,,$$
and 
$$p_1[1]+q_1[1]=p_2[1] + q_2[1]\text{ when }\sigma[3]=-1\,,$$
and where 
\begin{equation}
\sigma[0]=\begin{cases}
1,&\text{ if $p_2[0]=p_1[0]$ and $q_2[0]= q_1[0]$; and}\\
\frac{q_2[0]-q_1[0]}{p_2[0]-p_1[0]},&\text{otherwise, provided that }p_2[0]\neq p_1[0]\,.
\end{cases}
\end{equation}
Note that, for a given pair of bases, $\langle p_1,p_2\rangle$ and $\langle q_1, q_2\rangle$, \eqref{y-difference-equal} and \eqref{y-difference-equal-sigma-3-equals-1} are both true only when $p_1[1]=p_2[1]$. This implies that $|\BTSF_{\FTwoSTR}(\langle p_1,p_2\rangle,\langle q_1, q_2\rangle)|=\tau(\FTwoSTR)=2$ when $p_1[1]=p_2[1]$ and the other conditions above are satisfied. Figure~\ref{f2str-example} shows an example of a basis pair for which the basis-pair transformation set function returns a set containing two transformations.

$\FTwoSTR$, as defined here, is therefore a transformation class. The basis-pair transformation set function, $\BTSF_{\FTwoSTR}$, can be computed using the algorithm shown in Fig.~\ref{FTwoSTR-basis-transformation-set-function}.

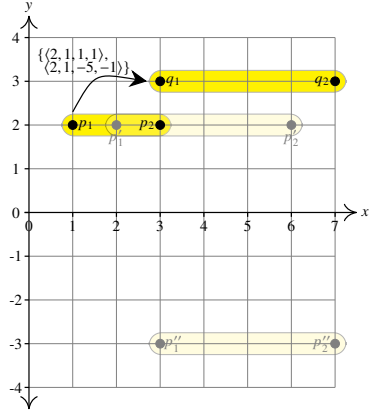
\begin{figure}[t]
	\begin{center}
		{\resizebox{5cm}{!}{
			\begin{tikzpicture}
\filldraw [rounded corners=8pt, fill=yellow, draw opacity=0.3] (.75,5.75) rectangle (3.25,6.25);
\filldraw [rounded corners=8pt, fill=yellow, draw opacity=0.3] (2.75,6.75) rectangle (7.25,7.25);
\filldraw [rounded corners=8pt, fill=yellow!50!white, fill opacity=.3, draw opacity=0.3] (2.75,.75) rectangle (7.25,1.25);
\filldraw [rounded corners=8pt, fill=yellow!50!white, fill opacity=.3, draw opacity=0.3] (1.75,5.75) rectangle (6.25,6.25);
%
			\foreach \x in {1,...,7} {
				\draw[very thin,gray] (\x,0) -- (\x,8);
				\draw[thick] (\x,3.9) -- (\x,4);
			}
%
			\foreach \y in {0,...,8} {
				\draw[very thin,gray] (0,\y) -- (7,\y);
				\draw[thick] (-.1,\y) -- (0,\y);
			}
%
			\foreach \y in {-4,-3,...,4} {
				\draw (-.3,\y+4) node[]{\y};
			}
			\foreach \x in {0,1,...,7} {
				\draw (\x,3.7) node[]{\x};
			}
%
			\draw[thick,-{>[length=3mm,width=4mm]}] (-.1,4) -- (7.5,4);
			\draw[thick,-{>[length=3mm,width=4mm]}] (0,4) -- (0,8.5);		
			\draw[thick,-{>[length=3mm,width=4mm]}] (0,4) -- (0,-.5);		
			\draw (7.7,4) node{$x$};
			\draw (0,8.7) node{$y$};
			\filldraw  (1,6) circle[radius=.1]
				(3,6) circle[radius=.1]
				(3,7) circle[radius=.1]
				(7,7) circle[radius=.1];
			\filldraw[gray] (3,1) circle[radius=.1]
				(7,1) circle[radius=.1]
				(2,6) circle[radius=.1]
				(6,6) circle[radius=.1];

			\draw (1.3,6) node[]{$p_1$};
			\draw (2.7,6) node[]{$p_2$};
			\draw[gray] (3.3,1) node[]{$p_1''$};
			\draw[gray] (6.7,1) node[]{$p_2''$};
			\draw[gray] (2,5.7) node[]{$p_1'$};
			\draw[gray] (6,5.7) node[]{$p_2'$};
			\draw (3.3,7) node[]{$q_1$};
			\draw (6.7,7) node[]{$q_2$};
			
			 \draw[-{Stealth[length=4mm,width=4mm]}, thick] (1,6.3) .. controls (1.5,7.2) .. (2.7,7);
			 \draw (1,7.6) node{$\lbrace\langle2,1,1,1\rangle,$};
			 \draw (1.15,7.3) node{$\quad\langle2,1,-5,-1\rangle\rbrace$};
			\end{tikzpicture}
}}
\end{center}
\caption{There are two transformations in $\FTwoSTR$ that map $\langle p_1,p_2\rangle$ onto $\langle q_1,q_2\rangle$---those whose parameter vectors are $\langle2,1,1,1\rangle$ and $\langle2,1,-5,-1\rangle$. The transformation with parameter vector $\langle2,1,1,1\rangle$ first $x$-scales $\langle p_1,p_2\rangle$ by a scale factor of 2 to give $\langle p_1',p_2'\rangle$; it then translates $\langle p_1',p_2'\rangle$ by $\langle1,1\rangle$ to give $\langle q_1,q_2\rangle$. The transformation with parameter vector $\langle2,1,-5,-1\rangle$ again first $x$-scales $\langle p_1,p_2\rangle$ by a scale factor of 2 to give $\langle p_1',p_2'\rangle$, but then translates $\langle p_1',p_2'\rangle$ by $\langle1,-5\rangle$ to give $\langle p_1'',p_2''\rangle$, which is then reflected in the $x$-axis to give $\langle q_1,q_2\rangle$.}
	\label{f2str-example}
\end{figure}

\subsubsection{The set of transformation classes $\FTwoSTRMod{\mu}$}
\label{example-FTwoSTRmodMu}
For some specified value of $\mu$, $\FTwoSTRMod{\mu}$ is the set of 2-dimensional transformations, such that each transformation, $f\in\FTwoSTRMod{\mu}$, takes a point, $p=\langle x, y\rangle$, multiplies $x$ by a real-valued scale factor, $s$, to get a new point, $p'=\langle sx,y\rangle$, then translates $p'$ by the vector, $v= \langle v_x, v_y\rangle$, with the $y$-component of $p'$ being added modulo $\mu$ to the $y$-component of $v$, to give a new point, $p''= \langle sx+v_x,(y+v_y)\bmod\mu\rangle$.
Finally, $p''$ may be reflected in the $x$-axis, modulo $\mu$, to give a point, $q=\langle sx+v_x,(b(y+v_y))\bmod\mu\rangle$, where $b=-1$ if $p''$ is reflected and $b=1$ if it is not.
Note that we define that $a\bmod b= a - b\left\lfloor\frac a b\right\rfloor$, where $a,b\in\reals$.

We can therefore define that $\lambda(\FTwoSTRMod{\mu})= 4$ and that each transformation $f\in\FTwoSTRMod{\mu}$, can be uniquely identified by a 4-vector, $\langle s, v_x, v_y, b\rangle\in\reals^2\times\reals_{\mu}\times\left\lbrace-1,1\right\rbrace$, where we define that $\reals_{\mu}=\left\lbrace x: x\in\reals\land0\le x<\mu\right\rbrace$. If $\sigma_{\FTwoSTRMod{\mu}}(f)=\langle s, v_x, v_y, b\rangle$ and $p=\langle x,y\rangle$, then $f(p)=\langle sx+v_x,(b(y+v_y))\bmod\mu\rangle$. This satisfies the first constraint in Definition~\ref{transformation-class-definition}.

We can define the transformation class function for $\FTwoSTRMod{\mu}$ as follows:
\begin{equation}
\phi_{\FTwoSTRMod{\mu}}(\sigma,p)=\langle p[0]\sigma[0]+\sigma[1],(\sigma[3](p[1]+\sigma[2]))\bmod\mu\rangle\,,
\label{eqn-f2str-mod}
\end{equation}
where $p\in\reals\times\reals_{\mu}$ and there must exist a transformation, $f\in\FTwoSTRMod{\mu}$, such that $\sigma=\sigma_{\FTwoSTRMod{\mu}}(f)$. This satisfies the second constraint in Definition~\ref{transformation-class-definition}.

For $\FTwoSTRMod{\mu}$, we can define that $\beta(\FTwoSTRMod{\mu})=2$ and that $\tau(\FTwoSTRMod{\mu})=2$. The basis-pair transformation set function, $\BTSF_{\FTwoSTRMod{\mu}}$, takes any two ordered pairs of distinct points, $\objectBasis$ and $\imageBasis$, and computes the transformation parameter vectors of the transformations in $\FTwoSTRMod{\mu}$ that map $\objectBasis$ onto $\imageBasis$. Let us suppose that $\objectBasis=\langle p_1, p_2 \rangle$ and that  $\imageBasis=\langle q_1, q_2 \rangle$. If $\sigma\in\BTSF_{\FTwoSTRMod{\mu}}(\objectBasis,\imageBasis)$, then we know from \eqref{eqn-f2str-mod} that 
\begin{eqnarray}
q_1[0]&=&p_1[0]\sigma[0]+\sigma[1]\,,\label{f2strmod-eq1}\\
q_2[0]&=&p_2[0]\sigma[0]+\sigma[1]\,,\label{f2strmod-eq2}\\
q_1[1]&=&(\sigma[3](p_1[1]+\sigma[2]))\bmod\mu\textrm{, and }\label{f2strmod-eq3}\\
q_2[1]&=&(\sigma[3](p_2[1]+\sigma[2]))\bmod\mu\,.\label{f2strmod-eq4}
\end{eqnarray}
We know that $\sigma[3]$ is equal to either $1$ or $-1$. Let us first consider the case where $\sigma[3]=1$. In this case, 
\begin{eqnarray}
\eqref{f2strmod-eq3}&\Rightarrow& \sigma[2]=(q_1[1] - p_1[1])\bmod\mu\textrm{, and }\label{f2strmod-eq5}\\
\eqref{f2strmod-eq4}&\Rightarrow& \sigma[2]=(q_2[1] - p_2[1])\bmod\mu\,.\label{f2strmod-eq6}
\end{eqnarray}
Equations \eqref{f2strmod-eq5} and \eqref{f2strmod-eq6} together imply that, in order for there to be a transformation in 
$\FTwoSTRMod{\mu}$ that maps $\objectBasis$ onto $\imageBasis$ when $\sigma[3]=1$, it must be true that 
\begin{equation}
(q_1[1]-p_1[1])\bmod\mu=(q_2[1] - p_2[1])\bmod\mu\,.
\label{f2strmod-y-difference-equal}
\end{equation}
Figure~\ref{f2strmod-y-difference-equal-example} shows an example of two ordered pairs of points that do not satisfy \eqref{f2strmod-y-difference-equal} when $\mu=7$. There is therefore no transformation in $\FTwoSTRMod{\mu}$ that maps $\langle p_1,p_2\rangle$ onto $\langle q_1, q_2\rangle$ in this figure.
\begin{figure}[t]
	\begin{center}
		{\resizebox{5cm}{!}{
			\begin{tikzpicture}
			

\filldraw [rounded corners=8pt, fill=yellow, draw opacity=0.3, cm={cos(-63.4) ,-sin(-63.4) ,sin(-63.4) ,cos(-63.4) ,(4.1,.63)}] (0,0) rectangle (2.85, .5);
\filldraw [rounded corners=8pt, fill=yellow, draw opacity=0.3, cm={cos(-45) ,-sin(-45) ,sin(-45) ,cos(-45) ,(1,5.63)}] (0,0) rectangle (1, .5);
\filldraw [rounded corners=8pt, fill=yellow, draw opacity=0.3, cm={cos(-45) ,-sin(-45) ,sin(-45) ,cos(-45) ,(1.65,-.7)}] (0,0) rectangle (1, .5);

\fill[white](1.25,-.75) rectangle (2.5,-.25);
\fill[white](.75,6.3) rectangle (1.75,7);
			
			\draw[thick,-{>[length=3mm,width=4mm]}] (-.1,0) -- (6.5,0);
			\draw[thick,-{>[length=3mm,width=4mm]}] (0,-.1) -- (0,6.5);		
			
			\draw (6.7,0) node{$x$};
			\draw (0,6.7) node{$y$};
			
			\foreach \x in {1,...,6} {
				\draw[very thin,gray] (\x,0) -- (\x,6);
				\draw[thick] (\x,-.1) -- (\x,0);
				\draw[very thin,gray] (0,\x) -- (6,\x);
				\draw[thick] (-.1,\x) -- (0,\x);
			}
			
			\foreach \x in {0,1,...,6} {
				\draw (-.2,\x) node[anchor=east]{\x};
			}
			\foreach \x in {0,1,...,6} {
				\draw (\x,-.2) node[anchor=north]{\x};
			}

			\filldraw  (1,6) circle[radius=.1]
				(2,0) circle[radius=.1]
				(4,1) circle[radius=.1]
				(5,3) circle[radius=.1];

			\draw (1.4,5.6) node[]{$p_1$};
			\draw (1.6,0.4) node[]{$p_2$};
			\draw (3.6,1.4) node[]{$q_1$};
			\draw (5.4,2.6) node[]{$q_2$};
			
			\end{tikzpicture}
		}}
	\end{center}
	\caption{There is no transformation in $\FTwoSTRMod7$ that maps $\langle p_1,p_2\rangle$ onto $\langle q_1, q_2\rangle$ because condition~\eqref{f2strmod-y-difference-equal} is not satisfied.}
	\label{f2strmod-y-difference-equal-example}
\end{figure}
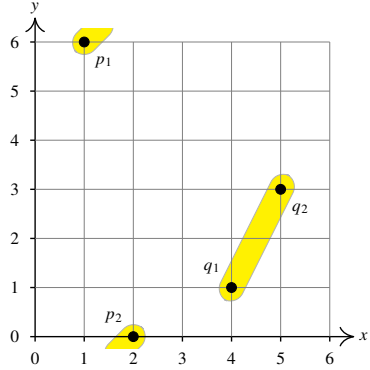

Equations~\eqref{f2strmod-eq1} and \eqref{f2strmod-eq2} together imply that
\begin{eqnarray}
q_1[0]&=&p_1[0]\sigma[0]+q_2[0]-p_2[0]\sigma[0]\nonumber\\
&=&\sigma[0](p_1[0]-p_2[0])+q_2[0]\nonumber\\
\Rightarrow\sigma[0]&=&\frac{q_1[0]-q_2[0]}{p_1[0]-p_2[0]}\,.\label{f2strmod-eq8}
\end{eqnarray}
Equation \eqref{f2strmod-eq8} implies that $\sigma[0]= 0$ when $q_1[0]=q_2[0]$, unless $p_1[0]=p_2[0]$, in which case the right-hand side of \eqref{f2strmod-eq8} evaluates to $\frac00$, which is undefined. However, as with $\FTwoST$ and $\FTwoSTR$, we define that, in this case, $\sigma[0]$ should be set to 1, as this allows us to intuitively discover that, if $q_2[0]=q_1[0]$ and $p_2[0]=p_1[0]$ and condition \eqref{f2strmod-y-difference-equal} holds, then $\objectBasis$ can be mapped onto $\imageBasis$ by a simple translation (i.e., no scaling or reflection required). Figure~\ref{f2strmod7-zero-over-zero-ex} shows an example of where $\sigma[0]$ is set to 1 because the right-hand side of \eqref{f2strmod-eq8} evaluates to $\frac00$. Clearly, if $p_2[0]=p_1[0]$ when $q_2[0]\neq q_1[0]$, then $\sigma[0]$ is undefined.
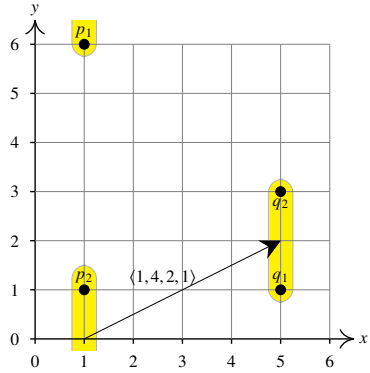
\begin{figure}[t]
	\begin{center}
		{\resizebox{5cm}{!}{
			\begin{tikzpicture}
			
\filldraw [rounded corners=8pt, fill=yellow, draw opacity=0.3] (4.75,0.75) rectangle (5.25,3.25);
\filldraw [rounded corners=8pt, fill=yellow, draw opacity=0.3] (.75,-.5) rectangle (1.25,1.5);
\filldraw [rounded corners=8pt, fill=yellow, draw opacity=0.3] (.75,5.75) rectangle (1.25,6.75);


\fill[white](.5,6.5) rectangle (1.5,7.25);
\fill[white](.5,-.75) rectangle (1.5,-.25);
			
			\draw[thick,-{>[length=3mm,width=4mm]}] (-.1,0) -- (6.5,0);
			\draw[thick,-{>[length=3mm,width=4mm]}] (0,-.1) -- (0,6.5);		
			
			\draw (6.7,0) node{$x$};
			\draw (0,6.7) node{$y$};
			
			\foreach \x in {1,...,6} {
				\draw[very thin,gray] (\x,0) -- (\x,6);
				\draw[thick] (\x,-.1) -- (\x,0);
				\draw[very thin,gray] (0,\x) -- (6,\x);
				\draw[thick] (-.1,\x) -- (0,\x);
			}
			
			\foreach \x in {0,1,...,6} {
				\draw (-.2,\x) node[anchor=east]{\x};
			}
			\foreach \x in {0,1,...,6} {
				\draw (\x,-.2) node[anchor=north]{\x};
			}

			\filldraw  (1,6) circle[radius=.1]
				(1,1) circle[radius=.1]
				(5,1) circle[radius=.1]
				(5,3) circle[radius=.1];

			\draw (1,6.25) node[]{$p_1$};
			\draw (1,1.25) node[]{$p_2$};
			\draw (5,1.25) node[]{$q_1$};
			\draw (5,2.75) node[]{$q_2$};
		\draw[-{Stealth[length=4mm,width=4mm]}] (1,0) -- (5,2);
		\draw (2.6,1.25) node{$\langle1,4,2,1\rangle$};
			\end{tikzpicture}
		}}
	\end{center}
	\caption{In this case, $\frac{q_1[0]-q_2[0]}{p_1[0]-p_2[0]}$ evaluates to $\frac00$, so $\sigma[0]$ is set to 1. The only transformation in $\FTwoSTRMod7$ that maps $\langle p_1,p_2\rangle$ onto $\langle q_1,q_2\rangle$ is that with the transformation parameter vector, $\langle1,4,2,1\rangle$.}
	\label{f2strmod7-zero-over-zero-ex}
\end{figure}

Equation \eqref{f2strmod-eq1} implies that 
\begin{equation}
\sigma[1]=q_1[0]-p_1[0]\sigma[0]\,,
\label{f2strmod-eq9}
\end{equation}
and from \eqref{f2strmod-eq2}, it follows that 
\begin{equation}
\sigma[1]=q_2[0]-p_2[0]\sigma[0]\,.
\label{f2strmod-eq10}
\end{equation}
It was shown above (see \eqref{sigma1-rhs-are-equal-a} and \eqref{sigma1-rhs-are-equal-b} on page~\pageref{sigma1-rhs-are-equal-a}) that the right-hand side of \eqref{f2strmod-eq9} is always equal to the right-hand side of \eqref{f2strmod-eq10}. Thus, when $\sigma[0]$ is undefined (i.e., when $p_2[0]=p_1[0]$ and $q_2[0]\neq q_1[0]$), $\sigma[1]$ is also undefined. 

To summarise the results so far, if $\sigma[3]=1$, then
\begin{equation}
\sigma=\left\langle\sigma[0],q_2[0]-p_2[0]\sigma[0],(q_1[1]-p_1[1])\bmod \mu,1\right\rangle\,,\end{equation}
where
$$(q_1[1]-p_1[1])\bmod\mu=(q_2[1] - p_2[1])\bmod\mu\,,$$
and where
$$
\sigma[0]=\begin{cases}
1,&\text{ if $p_2[0]=p_1[0]$ and $q_2[0]= q_1[0]$; and}\\
\frac{q_2[0]-q_1[0]}{p_2[0]-p_1[0]},&\text{otherwise, provided that }p_2[0]\neq p_1[0]\,.
\end{cases}
$$


Let us now consider the case where $\sigma[3]=-1$. In this case, $\sigma[0]$ and $\sigma[1]$ are still given by \eqref{f2strmod-eq8} and \eqref{f2strmod-eq9}, respectively, since they do not depend on $\sigma[3]$. 

From \eqref{f2strmod-eq3} and \eqref{f2strmod-eq4} we conclude that 
\begin{eqnarray}
\eqref{f2strmod-eq3}\Rightarrow \sigma[2]=(-p_1[1]-q_1[1])\bmod\mu\label{f2strmod-eq12}\\
\eqref{f2strmod-eq4}\Rightarrow \sigma[2]=(-p_2[1]-q_2[1])\bmod\mu\label{f2strmod-eq13}
\end{eqnarray}
and that, therefore, in order for there to be a transformation that maps $\objectBasis$ onto $\imageBasis$ when $\sigma[3]=-1$, it must be true (from Lemma~\ref{mod-lemma-1}) that 
\begin{equation}
(p_1[1]+q_1[1])\bmod\mu=(p_2[1] + q_2[1])\bmod\mu\,.
\label{f2strmod-y-difference-equal-sigma-3-equals-1}
\end{equation}

\begin{figure}[t]
\begin{center}
\resizebox*{.7\columnwidth}{!}{
\begin{minipage}{\linewidth}
\begin{tabbing}
\input{ColumnDefinitions}
$\fun{$\BTSF_{\FTwoSTRMod{\mu}}$}(\objectBasis,\imageBasis)$\\
1\>$\dmif \imageBasis = \objectBasis$\\
2\>\>$\dmreturn\lbrace\langle1,0,0,1\rangle\rbrace\dmcomment{Identity transformation}$\\
3\>$p_1\becomes\objectBasis[0],p_2\becomes\objectBasis[1],q_1\becomes\imageBasis[0],q_2\becomes\imageBasis[1]$\\
4\>$\dmif q_2[0]\neq q_1[0]\land p_2[0]=p_1[0]\dmcomment{$\sigma[0]$ is undefined.}$\\
5\>\>$\dmreturn \lbrace\rbrace$\\
6\>$\dmif q_2[0]= q_1[0]\land p_2[0]=p_1[0] \dmcomment{$\sigma[0]=0/0$}$\\
7\>\>$\sigma_0\becomes1$\\
8\>$\dmelse$\\
9\>\>$\sigma_0\becomes(q_2[0]-q_1[0])/(p_2[0]-p_1[0])\dmcomment{From \eqref{f2strmod-eq8}}$\\
10\>$S\becomes\lbrace\rbrace$\\
11\>$\sigma_1\becomes q_2[0]-p_2[0]\sigma_0\dmcomment{From \eqref{f2strmod-eq10}}$\\
12\>$\dmif (q_1[1]-p_1[1])\bmod\mu=(q_2[1] - p_2[1])\bmod\mu\dmcomment{From \eqref{f2strmod-y-difference-equal}}$\\
13\>\>$\sigma_3\becomes1$\\
14\>\>$\sigma_2\becomes (q_1[1]-p_1[1])\bmod\mu\dmcomment{From \eqref{f2strmod-eq5}}$\\
15\>\>$S\becomes \lbrace\langle\sigma_0,\sigma_1,\sigma_2,\sigma_3\rangle\rbrace$\\
16\>$\dmif (p_1[1]+q_1[1])\bmod\mu=(p_2[1] + q_2[1])\bmod\mu\dmcomment{From \eqref{f2strmod-y-difference-equal-sigma-3-equals-1}}$\\
17\>\>$\sigma_3\becomes-1$\\
18\>\>$\sigma_2\becomes(-p_1[1]-q_1[1])\bmod\mu\dmcomment{From \eqref{f2strmod-eq12}}$\\
19\>\>$S\becomes S\cup\lbrace\langle\sigma_0,\sigma_1,\sigma_2,\sigma_3\rangle\rbrace$\\
20\>$\dmreturn S$
\end{tabbing}
\end{minipage}
}
\end{center}
\caption{Algorithm for computing the basis-pair transformation set function for $\FTwoSTRMod{\mu}$.}
\label{FTwoSTRModMu-basis-transformation-set-function}
\end{figure}

To summarise, we can define the basis-pair transformation set function, $\BTSF_{\FTwoSTRMod{\mu}}$, as follows:
\begin{multline}
\BTSF_{\FTwoSTRMod{\mu}}\left(\langle p_1,p_2\rangle,\langle q_1,q_2\rangle\right)=\\
\left\lbrace\left\langle\sigma[0],q_2[0]-p_2[0]\sigma[0],(q_1[1]-p_1[1])\bmod\mu,1\right\rangle\right.,\\
\left.\left\langle\sigma[0],q_2[0]-p_2[0]\sigma[0],(-p_1[1]-q_1[1])\bmod\mu,-1\right\rangle\right\rbrace
\label{equation-BTSFFTwoSTRModMu}
\end{multline}
where
\begin{equation}
(q_1[1]-p_1[1])\bmod\mu=(q_2[1] - p_2[1])\bmod\mu\text{ when }\sigma[3]=1\,,
\label{tfstrmodmu-condition-1}
\end{equation}
and
\begin{equation}
(p_1[1]+q_1[1])\bmod\mu=(p_2[1] + q_2[1])\bmod\mu\text{ when }\sigma[3]=-1\,,
\label{tfstrmodmu-condition-2}
\end{equation}
and
\begin{equation}
\sigma[0]=\begin{cases}
1,&\text{ if $p_2[0]=p_1[0]$ and $q_2[0]= q_1[0]$; and}\\
\frac{q_2[0]-q_1[0]}{p_2[0]-p_1[0]},&\text{otherwise, provided that }p_2[0]\neq p_1[0]\,.
\end{cases}
\end{equation}

Conditions \eqref{tfstrmodmu-condition-1} and \eqref{tfstrmodmu-condition-2} are true if and only if $p_1[1]=p_2[1]$---recall that, in the context of $\FTwoSTRMod{\mu}$ we assume that all points are in $\mathbb R\times\mathbb R_{\mu}$. Therefore $|\BTSF_{\FTwoSTRMod{\mu}}(\langle p_1,p_2\rangle,\langle q_1,q_2\rangle)|=\tau(\FTwoSTRMod{\mu})=2$ if and only if $p_1[1]=p_2[1]$ and the other conditions above are satisfied.

$\FTwoSTRMod{\mu}$ is therefore a transformation class. The basis-pair transformation set function, $\BTSF_{\FTwoSTRMod{\mu}}$, can be computed using the algorithm shown in Fig.~\ref{FTwoSTRModMu-basis-transformation-set-function}.


\section{Matches and occurrences}
\label{matches-and-occurrences-section}

In this section, we begin by defining the concepts of a complete match and a complete occurrence. We then introduce the notions of maximal match and maximal pattern. We then prove some results relating to the important case where maximal matches and patterns are for transformations that are bijections. Finally, we note the relationship between maximal patterns and partial symmetries.

\subsection{Complete matches and occurrences}
We define the terms, {\em match\/} and {\em occurrence\/} as follows.
\begin{definition}[Match]
Given two, non-empty, finite point sets, $P,Q\subset\mathbb{R}^k$, and a $k$-dimensional transformation, $f$, we define that $Q$ is a {\em match\/} or {\em complete exact match\/} of $P$ {\em under the transformation}, $f$, if and only if $Q=f(P)$. 
Furthermore, if $\tc$ is a $k$-dimensional transformation class, then we can say that  $Q$ is a {\em match\/} or {\em complete exact match\/} of $P$ {\em with respect to} $\tc$, if and only if there exists a transformation, $f\in\tc$, such that $Q=f(P)$.
\label{match-definition}
\end{definition}

\begin{definition}[Occurrence]
Given two, non-empty, finite point sets, $P,Q\subset\mathbb{R}^k$, and a $k$-dimensional transformation, $f$,
we say that $P$ {\em occurs in\/} or {\em occurs completely and exactly in\/} $Q$ {\em under the transformation}, $f$, if and only if $f(P)\subseteq Q$. In this case, $f(P)$ is an {\em occurrence\/} or a {\em complete exact occurrence\/} of $P$ in $Q$ {\em under the transformation}, $f$. 
Furthermore, if $F$ is a $k$-dimensional transformation class, then
we say that $P$ {\em occurs in\/} or  {\em occurs completely and exactly in\/} $Q$ {\em with respect to\/} $F$, if and only if there exists a transformation, $g\in F$, such that $g(P)\subseteq Q$. In this case, $g(P)$ is an {\em occurrence\/} or a {\em complete exact occurrence\/} of $P$ in $Q$ {\em with respect to\/} $F$.
\label{occurrence-definition}
\end{definition}

Note that, in Definition~\ref{occurrence-definition}, $\pattern$ does not necessarily need to be a subset of $Q$ in order for $\pattern$ to occur in $Q$. 
For example, the point set in Fig.~\ref{occurrence-example}(a) {\em occurs\/} twice in the point set in Fig.~\ref{occurrence-example}(b) under translations, but it is not a {\em subset\/} of the point set in Fig.~\ref{occurrence-example}(b). 

Note that we use the term ``match'' when we compare two whole point sets (typically called ``patterns'') with each other, whereas we use the term ``occurrence'' when we are interested in whether one point set (which we would typically call a ``pattern'') appears within (i.e., as part of) some other point set (typically called a ``dataset'') under some transformation. For example, if we refer to the point set in Fig.~\ref{occurrence-example}(a) as $P$ and the point set in Fig.~\ref{occurrence-example}(b) as $Q$, then we might say that the red pattern in $Q$ is a {\em match\/} of $P$ under the transformation, $T_{\langle1,1\rangle}$, and that it is an {\em occurrence\/} of $P$ in $Q$, with respect to the transformation class, $\FTwoT$.

\begin{figure}[t]
	\begin{center}
		\subfloat[]{\resizebox{5cm}{!}{
			\begin{tikzpicture}

			\draw[thick] (-.1,0) -- (8,0);
			\draw[thick] (0,-.1) -- (0,8);		
			
			\draw (4,-1) node{$x$};
			\draw (-1,4) node{$y$};
			
			\foreach \x in {1,...,8} {
				\draw[very thin] (\x,0) -- (\x,8);
				\draw[thick] (\x,-.1) -- (\x,0);
				\draw[very thin] (0,\x) -- (8,\x);
				\draw[thick] (-.1,\x) -- (0,\x);
			}
			\foreach \x in {0,1,...,8} {
				\draw (\x,-.2) node[anchor=north]{\x};
				\draw (-.2,\x) node[anchor=east]{\x};
			}

			\filldraw  [green!50!black] (1,1) circle[radius=.1]
					        (1,3) circle[radius=.1]
					        (2,4) circle[radius=.1]
					        (4,3) circle[radius=.1];

			\filldraw (4,1) circle[radius=.1];


%
%
%
%
%
%
%
%
%
%
			
			\end{tikzpicture}
		}}
		\subfloat[]{\resizebox{5cm}{!}{
			\begin{tikzpicture}
			
			



			\draw[thick] (-.1,0) -- (8,0);
			\draw[thick] (0,-.1) -- (0,8);		
			
			\draw (4,-1) node{$x$};
			\draw (-1,4) node{$y$};
			
			\foreach \x in {1,...,8} {
				\draw[very thin] (\x,0) -- (\x,8);
				\draw[thick] (\x,-.1) -- (\x,0);
				\draw[very thin] (0,\x) -- (8,\x);
				\draw[thick] (-.1,\x) -- (0,\x);
			}
			\foreach \x in {0,1,...,8} {
				\draw (\x,-.2) node[anchor=north]{\x};
				\draw (-.4,\x) node[anchor=east]{\x};
			}

			\filldraw[red]  (2,2) circle[radius=.1]
					  (2,4) circle[radius=.1]
					  (3,5) circle[radius=.1]
					  (5,4) circle[radius=.1]
					  (5,2) circle[radius=.1];

			\filldraw[blue]  (5,5) circle[radius=.1]
					  (5,7) circle[radius=.1]
					  (6,8) circle[radius=.1]
					  (8,7) circle[radius=.1]
					  (8,5) circle[radius=.1];
					  
			\filldraw[green!50!black] (0,6) circle[radius=.1]
		    			      (0,8) circle[radius=.1]
					      (1,5) circle[radius=.1]
					      (3,6) circle[radius=.1];
%
%
%
%
%
%
%
%
%
%
			
			\end{tikzpicture}
		}}
	\end{center}
	\caption{Two point sets.}
	\label{occurrence-example}
\end{figure}
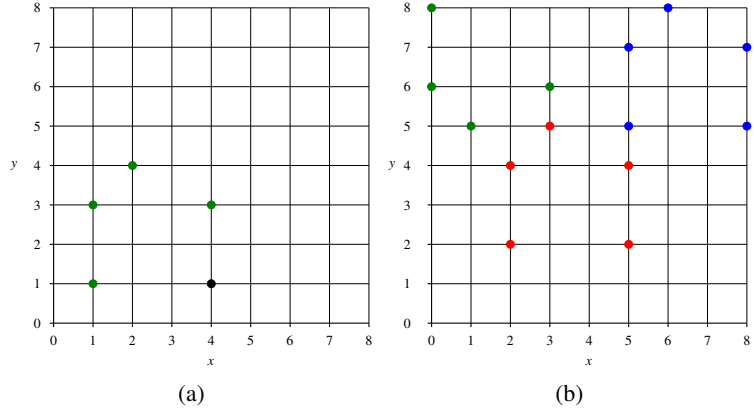

\subsection{Maximal matches and maximal patterns}
\label{maximal-matches-and-maximal-patterns-section}
In Definitions~\ref{match-definition} and \ref{occurrence-definition}, we consider only complete, exact matches and occurrences of a pattern. We may also be interested in cases where only part of a pattern occurs in a dataset under some transformation. We therefore define the concepts of a {\em maximal match\/} and a {\em maximal pattern} as follows.
 \begin{definition}[Maximal match for a transformation] 
Given two point sets, $P,Q\subset\reals^k$, and a $k$-dimensional transformation, 
$f$, we define the {\em maximal match\/} of 
$P$ in $Q$ for $f$ to be $f(P)\cap Q$, and we denote it by $\maxmatchfun(P,Q,f)$. In the case where the two point sets, $P$ and $Q$, are the same, so that $P=Q=D$, then we define that $\maxmatchfun(D,D,f) = \maxmatchfun(D,f) = f(D)\cap D$. In this case, we can say that $\maxmatchfun(D,f)$ is the {\em maximal match in $D$ for the transformation, $f$}. 
\label{maximal-match-definition}
\end{definition}

\begin{definition}[Maximal pattern for a transformation] 
Given two point sets, $P,Q\subset\reals^k$, and a $k$-dimensional transformation, 
$f$, we define the {\em maximal pattern\/} in $P$ for $Q$ and $f$ to be the set of points, $\lbrace p\in P: f(p)\in Q\rbrace$, and we denote it by $\maxpatfun(P,Q,f)$. In the case where the two point sets, $P$ and $Q$, are the same, so that $P=Q=D$, we define that $\maxpatfun(D,D,f)=\maxpatfun(D,f)=\lbrace p\in D: f(p)\in D\rbrace$. In this case, we say that $\maxpatfun(D,f)$ is the {\em maximal pattern in $D$ for the transformation $f$}. We also say that the transformation, $f$, {\em occurs in\/} $D$ if and only if $\maxpatfun(D,f)\neq\emptyset$. \label{maximal-pattern-definition}
\end{definition}

Definition~\ref{maximal-match-definition} therefore defines the maximal match of a pattern, $P$, in another pattern, $Q$, for a transformation, $f$, to be the set of points in $Q$ that are images of points in $P$ under the transformation, $f$. It also defines that the maximal match in a dataset, $D$, for a transformation, $f$, is the image of the set of points in $D$ that are mapped onto points in $D$ by the transformation, $f$.

Definition~\ref{maximal-pattern-definition} defines that the maximal pattern in a pattern, $P$, for another pattern, $Q$, and a transformation, $f$, is the set of points in $P$ that are mapped by $f$ onto points in $Q$. It also defines that the maximal pattern in a dataset, $D$, for a transformation, $f$, is the set of points in $D$ that are mapped by $f$ onto points in $D$.

For example, if we set $P$ to equal the point set in Fig.~\ref{occurrence-example}(a) and we set $Q$ to be the point set in Fig.~\ref{occurrence-example}(b), then $\maxmatchfun(P,Q,T_{\langle-1-9\rangle}R_x)$ is equal to the set of green points ({\tikz{\filldraw[green!50!black](0,0.05)circle[radius=.05];}}) in Fig.~\ref{occurrence-example}(b), and $\maxpatfun(P,Q,T_{\langle-1-9\rangle}R_x)$ is equal to the set of green points in Fig.~\ref{occurrence-example}(a). As examples of cases where the two point sets are the same, $\maxmatchfun(Q,T_{\langle3,3\rangle})$ is equal to the set of blue points in Fig.~\ref{occurrence-example}(b), whereas $\maxpatfun(Q,T_{\langle3,3\rangle})$ is equal to the set of red points in Fig.~\ref{occurrence-example}(b). As another example, $\maxmatchfun(Q,T_{\langle-2,-10\rangle}R_x)$ is equal to the set containing the green points in Fig.~\ref{occurrence-example}(b) together with the red point, $\langle3,5\rangle$; whereas $\maxpatfun(Q,T_{\langle-2,-10\rangle}R_x) = \lbrace\langle2,2\rangle,\langle2,4\rangle,\langle3,5\rangle,\langle5,4\rangle,\langle5,5\rangle\rbrace$.


We now present a lemma that characterises the relationship between a maximal pattern and a maximal match.

\begin{lemma} 
If $P,Q,D\subset\reals^k$, and $f$ is a $k$-dimensional transformation, then
\begin{equation}
f(\maxpatfun(P,Q,f)) = \maxmatchfun(P,Q,f)
\label{f-MP-equals-MM-eq}
\end{equation}
and
\begin{equation}
f(\maxpatfun(D,f)) = \maxmatchfun(D,f)\,.
\label{f-MP-D-equals-MM-D-eq}
\end{equation}
\label{f-MP-equals-MM-lemma}
\end{lemma}
\begin{proof}
$\maxpatfun(P,Q,f)$ is defined in Definition~\ref{maximal-pattern-definition} to be equal to $\lbrace p\in P: f(p)\in Q\rbrace$. That is, $\maxpatfun(P,Q,f)$ contains all the points in $P$ that are mapped onto points in $Q$ by $f$. Therefore, if we transform $\maxpatfun(P,Q,f)$ by $f$, we obtain the set of all points in $Q$ that are images of points in $P$ under $f$---that is, the set $f(P)\cap Q$. By Definition~\ref{maximal-match-definition}, $f(P)\cap Q=\maxmatchfun(P,Q,f)$, from which \eqref{f-MP-equals-MM-eq} follows directly. If we set $P=Q=D$, then \eqref{f-MP-D-equals-MM-D-eq} follows directly from \eqref{f-MP-equals-MM-eq} and Definitions~\ref{maximal-match-definition} and \ref{maximal-pattern-definition}.\qed
\end{proof}
Note that Lemma~\ref{f-MP-equals-MM-lemma} holds for any $k$-dimensional transformation, $f$, even one that is not a bijection (i.e., does not have a unique inverse). For example, if we set $P$ to be the point set in Fig.~\ref{f-MP-equals-MM-lemma-example}(a), $Q$ to be the point set in Fig.~\ref{f-MP-equals-MM-lemma-example}(b), and $f$ to be the transformation, $S_0T_{\langle3,0\rangle}$, then $\maxpatfun(P,Q,f)$ is equal to the set of points in Fig.~\ref{f-MP-equals-MM-lemma-example}(a) indicated by open circles ({\tikz{\draw[thick](0,0.05)circle[radius=.05];}}). When this set of points is transformed by $S_0,T_{\langle3,0\rangle}$, we obtain the set of points marked with crosses in Fig.~\ref{f-MP-equals-MM-lemma-example}(b), which is equal to $\maxmatchfun(P,Q,f)$.

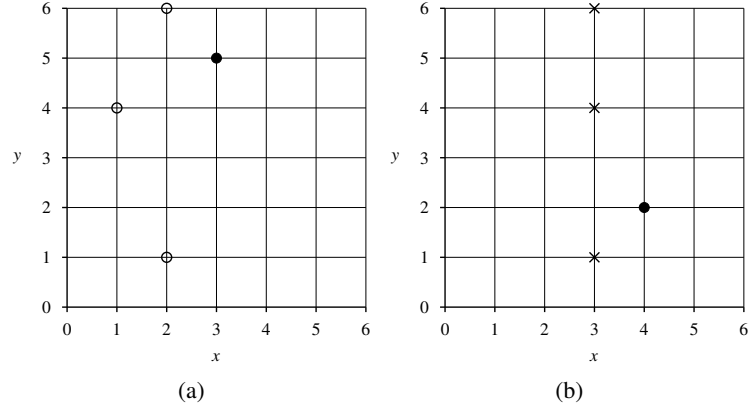
\begin{figure}[t]
	\begin{center}
		\subfloat[]{\resizebox{5cm}{!}{
			\begin{tikzpicture}

			\draw[thick] (-.1,0) -- (6,0);
			\draw[thick] (0,-.1) -- (0,6);		
			
			\draw (3,-1) node{$x$};
			\draw (-1,3) node{$y$};
			
			\foreach \x in {1,...,6} {
				\draw[very thin] (\x,0) -- (\x,6);
				\draw[thick] (\x,-.1) -- (\x,0);
				\draw[very thin] (0,\x) -- (6,\x);
				\draw[thick] (-.1,\x) -- (0,\x);
			}
			\foreach \x in {0,1,...,6} {
				\draw (\x,-.2) node[anchor=north]{\x};
				\draw (-.2,\x) node[anchor=east]{\x};
			}

			\draw  [thick] (2,1) circle[radius=.1]
					  (1,4) circle[radius=.1]
					  (2,6) circle[radius=.1];
			\filldraw[black] (3,5) circle[radius=.1];
			\end{tikzpicture}
		}}
		\subfloat[]{\resizebox{5cm}{!}{
			\begin{tikzpicture}
			
			\draw[thick] (-.1,0) -- (6,0);
			\draw[thick] (0,-.1) -- (0,6);		
			
			\draw (3,-1) node{$x$};
			\draw (-1,3) node{$y$};
			
			\foreach \x in {1,...,6} {
				\draw[very thin] (\x,0) -- (\x,6);
				\draw[thick] (\x,-.1) -- (\x,0);
				\draw[very thin] (0,\x) -- (6,\x);
				\draw[thick] (-.1,\x) -- (0,\x);
			}
			\foreach \x in {0,1,...,6} {
				\draw (\x,-.2) node[anchor=north]{\x};
				\draw (-.2,\x) node[anchor=east]{\x};
			}

			\filldraw[black] (4,2) circle[radius=.1];
			\drawcross31{black}
			\drawcross34{black}
			\drawcross36{black}
			
			\end{tikzpicture}
		}}
	\end{center}
	\caption{Illustration that Lemma~\ref{f-MP-equals-MM-lemma} holds even for transformations that are not bijections. The set of points marked with crosses in (b) form the maximal match of the pattern in (a) under the transformation $S_0T_{\langle3,0\rangle}$ which has no unique inverse and is therefore not a bijection. The set of points marked with open circles in (a) are mapped onto the set of crosses in (b) by $S_0T_{\langle3,0\rangle}$ and form the maximal pattern in (a) for (b) and this non-bijective transformation.}
	\label{f-MP-equals-MM-lemma-example}
\end{figure}

Definitions~\ref{maximal-match-definition} and \ref{maximal-pattern-definition} define the concepts of a maximal match and a maximal pattern, respectively, for a specific transformation. We now define the concepts of a maximal match and a maximal pattern with respect to a transformation class.

\begin{definition}[Maximal match with respect to a transformation class]
Let $P, Q\subset\reals^k$ and let $\tc$ be a $k$-dimensional transformation class. We define that {\em $S$ is a maximal match of $P$ in $Q$ with respect to $\tc$}, if and only if there exists a transformation, $f\in\tc$, such that $S=f(P)\cap Q$. 
If $P=Q=D$, then we can say that {\em $S$ is a maximal match in $D$ with respect to $\tc$}, if and only if there exists a transformation, $f\in\tc$, such that $S=f(D)\cap D$.  If $|S|\ge\basisSize(\tc)$, then we say that $S$ is a {\em non-trivial maximal match with respect to} $\tc$; otherwise, $S$ is a {\em trivial maximal match with respect to} $\tc$.
\label{max-match-for-a-tran-class-def}
\end{definition}

\begin{definition}[Maximal pattern with respect to a transformation class]
Let $P, Q\subset\reals^k$ and let $\tc$ be a $k$-dimensional transformation class.
We define that {\em $S$ is a maximal pattern in $P$ for $Q$ with respect to $\tc$}, if and only if there exists a transformation, $f\in\tc$, such that $f(S)=f(P)\cap Q$. 
If $P=Q=D$, then we can say that {\em $S$ is a maximal pattern in $D$ with respect to $\tc$}, if and only if there exists a transformation, $f\in\tc$, such that $f(S)=f(D)\cap D$. If $|S|\ge\basisSize(\tc)$, then we say that $S$ is a {\em non-trivial maximal pattern with respect to} $\tc$; otherwise, $S$ is a {\em trivial maximal pattern with respect to} $\tc$.
\label{max-pattern-for-a-tran-class-def}
\end{definition}

Note that we define that a maximal match or pattern is ``trivial'' if and only if its size is less than the basis size for the transformation class. This is because, by definition (see Def.~\ref{transformation-class-definition}), for any transformation class, there may be an infinite number of transformations that map any ordered set of $m$ points onto another ordered set of $m$ points, where $m$ is less than the basis size of the transformation class.

\subsection{Maximal matches and maximal patterns for bijections}

If a transformation, $f$, is a bijection then we define its {\em inverse}, $f^{-1}$, as
\begin{equation}
f^{-1}(q)=p\Leftrightarrow f(p)=q
\label{inverse-eq}
\end{equation} 
where $p,q\in\reals^k$. It follows directly that
\begin{equation}
f(f^{-1}(p))=f^{-1}(f(p))=p\,.
\label{f-inv-f-p-equals-p-eq}
\end{equation}
We now present a number of results pertaining to the relationship between maximal matches and maximal patterns, that hold for transformations that are bijections. 

\begin{lemma} 
If $P,Q,D\subset\reals^k$, and $f$ is a bijective $k$-dimensional transformation, then
\begin{equation}
\maxpatfun(P,Q,f) = P \cap f^{-1}(Q)
\label{MP-equals-inv-f-Q-cap-P-eq}
\end{equation}
and
\begin{equation}
\maxpatfun(D,f) =  D\cap f^{-1}(D)\,.
\label{MP-D-equals-inv-f-D-cap-D-eq}
\end{equation}
\label{MP-equals-inv-f-Q-cap-P-lemma}
\end{lemma}
\begin{proof}
From Lemma~\ref{f-MP-equals-MM-lemma}, we know that $f(\maxpatfun(P,Q,f)) = \maxmatchfun(P,Q,f)$ and from Definition~\ref{maximal-match-definition}, we know that $\maxmatchfun(P,Q,f)=f(P)\cap Q$. Therefore.
\begin{equation}
f(\maxpatfun(P,Q,f))=f(P)\cap Q\,.
\label{f-MP-equals-f-P-cap-Q-eq}
\end{equation}
Since $f$ is a bijection, it follows that 
\begin{eqnarray}
\maxpatfun(P,Q,f)&=&f^{-1}(f(\maxpatfun(P,Q,f)))\quad\textrm{(from \eqref{f-inv-f-p-equals-p-eq})}\nonumber\\
			&=&f^{-1}(f(P)\cap Q)\quad\textrm{ (from \eqref{f-MP-equals-f-P-cap-Q-eq})}\nonumber\\
			&=&P\cap f^{-1}(Q)\nonumber.
\end{eqnarray}
This proves \eqref{MP-equals-inv-f-Q-cap-P-eq}.
It follows directly that, if $P=Q=D$, then $\maxpatfun(D,f)=D\cap f^{-1}(D)$, which proves \eqref{MP-D-equals-inv-f-D-cap-D-eq}.\qed
\end{proof}

\begin{theorem} 
If $P,Q,D\subset\reals^k$, and $f$ is a bijective $k$-dimensional transformation, then
\begin{equation}
\maxmatchfun(P,Q,f) = f(\maxpatfun(P,Q,f))=f(\maxmatchfun(Q,P,f^{-1}))=\maxpatfun(Q,P,f^{-1})
\label{MM-MP-Theorem-PQ-eq}
\end{equation}
and
\begin{equation}
\maxmatchfun(D,f) = f(\maxpatfun(D,f))=f(\maxmatchfun(D,f^{-1}))=\maxpatfun(D,f^{-1})\,.
\label{MM-MP-Theorem-D-eq}
\end{equation}
\label{MM-MP-theorem}
\end{theorem}
\begin{proof}
From Lemma~\ref{f-MP-equals-MM-lemma}, we know that 
$$\maxmatchfun(P,Q,f) = f(\maxpatfun(P,Q,f))\,.$$ 
From Lemma~\ref{MP-equals-inv-f-Q-cap-P-lemma}, we know that 
$$\maxpatfun(P,Q,f) = P \cap f^{-1}(Q)\,,$$
which implies, from Definition~\ref{maximal-match-definition}, that 
$$\maxpatfun(P,Q,f) = \maxmatchfun(Q,P,f^{-1})\,.$$
It follows that $$f(\maxpatfun(P,Q,f)) = f(\maxmatchfun(Q,P,f^{-1}))\,.$$ 
From Lemma~\ref{f-MP-equals-MM-lemma} we know that $$f^{-1}(\maxpatfun(Q,P,f^{-1}))=\maxmatchfun(Q,P,f^{-1})\,.$$
Therefore 
$$f(f^{-1}(\maxpatfun(Q,P,f^{-1})))=f(\maxmatchfun(Q,P,f^{-1}))\,,$$ 
which implies, by \eqref{f-inv-f-p-equals-p-eq}, that 
$$\maxpatfun(Q,P,f^{-1})=f(\maxmatchfun(Q,P,f^{-1}))\,.$$ 
Therefore,
$$
\maxmatchfun(P,Q,f) = f(\maxpatfun(P,Q,f))=f(\maxmatchfun(Q,P,f^{-1}))=\maxpatfun(Q,P,f^{-1})\,.
$$
If $P=Q=D$, it follows directly that
$$
\maxmatchfun(D,f) = f(\maxpatfun(D,f))=f(\maxmatchfun(D,f^{-1}))=\maxpatfun(D,f^{-1})\,.
$$\qed
\end{proof}
For example, suppose $P$ is the set of points in Fig.~\ref{theorem21-example}(a) and $Q$ is the set of points in Fig.~\ref{theorem21-example}(b). Let $f$ be a translation by the vector, $\langle0,1\rangle$, so that $f^{-1}$ is a translation by the vector, $\langle0,-1\rangle$. The maximal match of $P$ in $Q$ for $f$, $\maxmatchfun(P,Q,f)$, is the set of points indicated by open circles in Fig.~\ref{theorem21-example}(b), which is also equal to the maximal pattern in $Q$ for $P$ and $f^{-1}$. Also, the set of points indicated by crosses in Fig.~\ref{theorem21-example}(a) is the maximal pattern in $P$ for $Q$ and $f$, $\maxpatfun(P,Q,f)$, which is equal to the maximal match of $Q$ in $P$ for $f^{-1}$. Note that, if we transform the set of points indicated by crosses in Fig.~\ref{theorem21-example}(a) by $f$, we obtain the set of points indicated by open circles in Fig.~\ref{theorem21-example}(b).

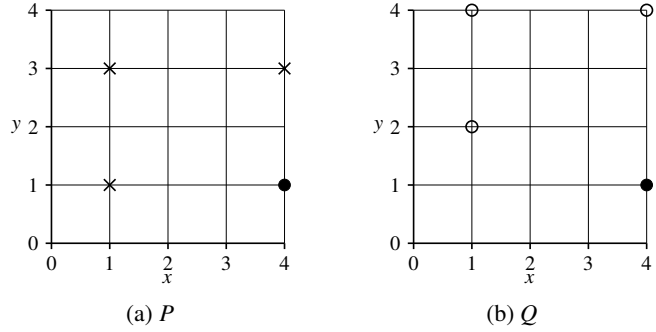
\begin{figure}[t]
	\begin{center}
		\subfloat[$P$]{\resizebox{4cm}{!}{
			\begin{tikzpicture}

			\draw[thick] (-.1,0) -- (4,0);
			\draw[thick] (0,-.1) -- (0,4);		
			
			\draw (2,-.6) node{$x$};
			\draw (-.6,2) node{$y$};
			
			\foreach \x in {1,...,4} {
				\draw[very thin] (\x,0) -- (\x,4);
				\draw[thick] (\x,-.1) -- (\x,0);
				\draw[very thin] (0,\x) -- (4,\x);
				\draw[thick] (-.1,\x) -- (0,\x);
			}
			\foreach \x in {0,1,...,4} {
				\draw (\x,-.1) node[anchor=north]{\x};
				\draw (-.1,\x) node[anchor=east]{\x};
			}

			\drawcross11{black};
			\drawcross13{black};
			\drawcross43{black};
			\filldraw[black] (4,1) circle[radius=.1];
			\end{tikzpicture}
		}}$\qquad$
		\subfloat[$Q$]{\resizebox{4cm}{!}{
			\begin{tikzpicture}
			
			\draw[thick] (-.1,0) -- (4,0);
			\draw[thick] (0,-.1) -- (0,4);		
			
			\draw (2,-.6) node{$x$};
			\draw (-.6,2) node{$y$};
			
			\foreach \x in {1,...,4} {
				\draw[very thin] (\x,0) -- (\x,4);
				\draw[thick] (\x,-.1) -- (\x,0);
				\draw[very thin] (0,\x) -- (4,\x);
				\draw[thick] (-.1,\x) -- (0,\x);
			}
			\foreach \x in {0,1,...,4} {
				\draw (\x,-.1) node[anchor=north]{\x};
				\draw (-.1,\x) node[anchor=east]{\x};
			}

			\filldraw[black] (4,1) circle[radius=.1];
			\draw  [thick] (1,2) circle[radius=.1]
					  (1,4) circle[radius=.1]
					  (4,4) circle[radius=.1];			
			\end{tikzpicture}
		}}
	\end{center}
	\caption{Two points sets illustrating Theorem~\ref{MM-MP-theorem}. See main text for details.}
	\label{theorem21-example}
\end{figure}

\subsection{Maximal patterns and partial symmetries}
Maximal patterns correspond to partial symmetries. A geometric object is said to have a particular type of symmetry if there exists a transformation of the corresponding type (e.g., a reflection or a rotation) that maps the object onto itself. More generally, we can say that an object or set of $k$-dimensional points, $D$, is symmetrical (or invariant) with respect to some class of transformations, $F$, if there exists a transformation, $f\in F$, such that $f(D)=D$. In this case, we can say that $f$ defines a {\em total symmetry\/} on $D$. More generally, a transformation, $g\in F$, defines a {\em partial symmetry\/} on $D$ if and only if $g(D)\cap D\neq\emptyset$ and $g(D)\neq D$. From Definition~\ref{maximal-pattern-definition}, it follows that a transformation $f$ defines a {\em partial symmetry\/} on a point set, $D$, if and only if $\maxpatfun(D,f)\neq\emptyset$ and $\maxpatfun(D,f)\neq D$. In other words, $f$ defines a partial symmetry on $D$ if and only if $f$ {\em occurs\/} in $D$ but $f(D)$ is not equal to $D$. If $\maxpatfun(D,f)=D$, then $f$ defines a total symmetry on $D$.

%

\section{Occurrence sets}
\label{main-occurrence-sets-section}
Figure~\ref{RAVEL-HAYDN-GROUND-TRUTH-CPO-GT} shows 31 ``occurrences'' of the HAYDN theme in Ravel's {\em Menuet sur le nom d'Haydn}, where the term, ``occurrence'' here is to be interpreted in its colloquial sense rather than in the precise sense of Definition~\ref{occurrence-definition}. Each of the ``occurrences'' in Fig.~\ref{RAVEL-HAYDN-GROUND-TRUTH-CPO-GT} is to be understood to be a ``transformation'' (again, in the colloquial sense) of the theme provided by Jules Écorcheville, editor of the {\em Revue musicale mensuelle de la Soci{\'e}t{\'e} Internationale de Musique\/} ({\em RSIM\/}), in which the piece was first published in 1910 \citep{HommageAJosephHaydn1910}. Écorcheville's theme is shown in Fig.~\ref{haydn-hommage-explanatory-note}. However, as discussed in \cite{MeredithArxivChapter12026}, it is not always easy to decide whether or not a given set of notes should be considered an ``occurrence'' of some given musical pattern.

\begin{figure}[t]
\centering
\resizebox{.7\linewidth}{!}{\includegraphics{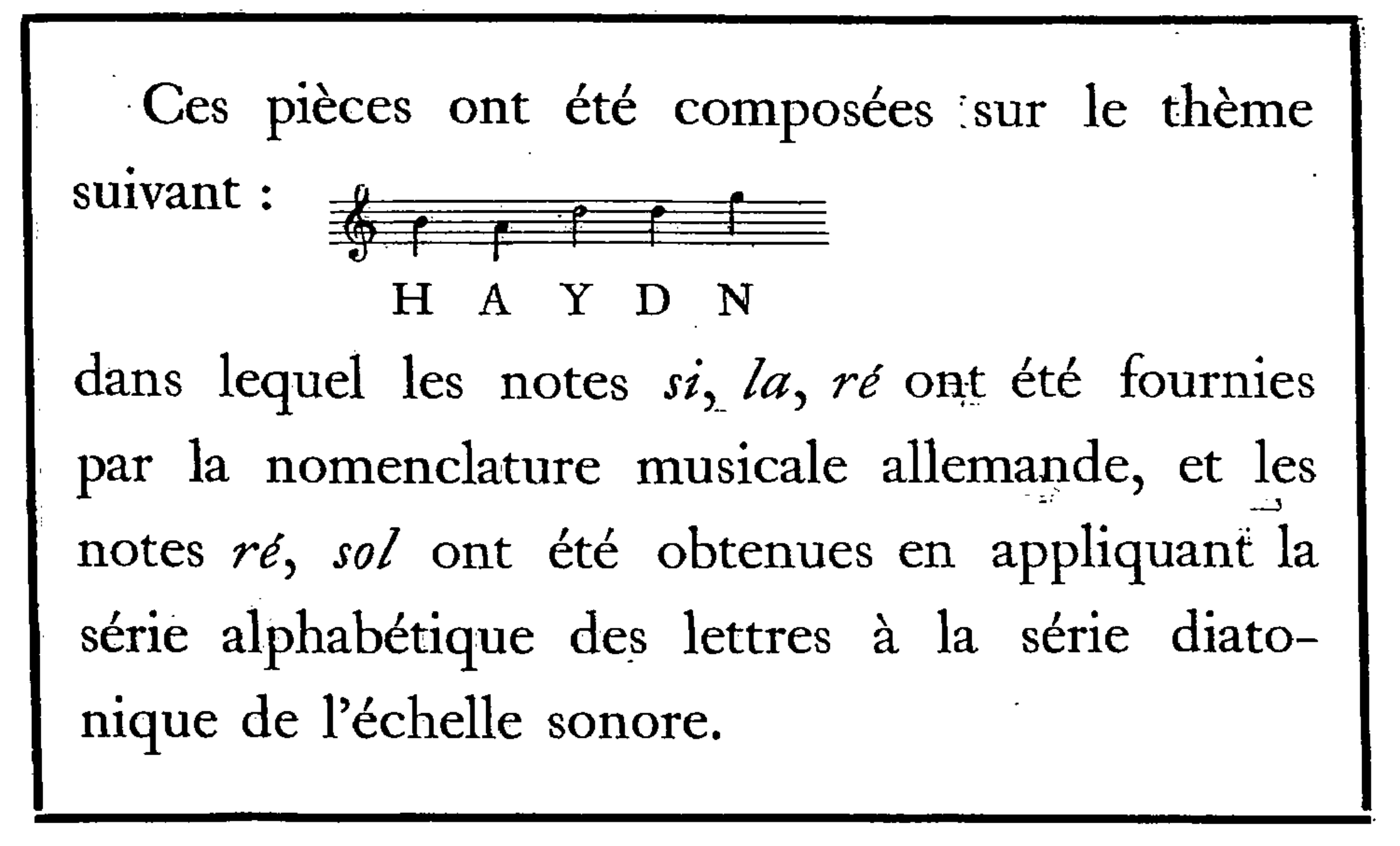}}
\caption{Explanatory note on page 2 of the 15 January 1910 issue of the {\em Revue musicale mensuelle de la Soci{\'e}t{\'e} Internationale de Musique}, preceding the collection of six pieces constituting the ``Hommage à Joseph Haydn''. Translated into English, the note reads as follows: ``These pieces have been composed on the following theme ... in which the notes B, A, D have been supplied by German musical nomenclature, and the notes D and G have been obtained by applying the sequence of letters in the alphabet to the diatonic series of the musical scale.''}
\label{haydn-hommage-explanatory-note}
\end{figure}

In this section, we continue our explication of the notions of a match and an occurrence by formalising the concept of an {\em occurrence set\/} and generalizing this concept to a {\em heterogeneous occurrence set}.

\subsection{Occurrence sets}
\label{occurrence-sets-section}
We define the concept of an {\em occurrence set\/} as follows.

\begin{definition}[Occurrence set]
Given a pattern, $P\subset\reals^k$, a dataset, $D\subset\reals^k$ and a $k$-dimensional transformation class, $F$, we define the {\em occurrence set\/} of $P$ in $D$ with respect to $F$ to be
\begin{equation}
\occsetfun(P,D,F) = \left\lbrace f(P): f(P)\subseteq D \land f\in F\right\rbrace\,.
\label{occurrence-set-eq}
\end{equation}
\label{occurrence-set-def}
\end{definition}
That is, the occurrence set of a pattern, $P$, in a dataset, $D$, with respect to a transformation class, $F$, contains all and only the occurrences (in the strict sense of Definition~\ref{occurrence-definition}) of $P$ in $D$ with respect to $F$.
Note that, in general, there may be more than one transformation in $\tc$ that maps $P$ onto a given occurrence of $P$ in $D$. The cardinality of the occurrence set of $P$ in $D$ with respect to $F$ will therefore always be less than or equal to the cardinality of the set of transformations in $F$ that map $P$ onto occurrences of $P$ in $D$. Note also that $P$ is not necessarily a subset of $D$.

Figure~\ref{occurrence-set-example} shows a set of twenty, 2-dimensional points with four pentagonal shapes indicated in different colours. The red, green and blue shapes are occurrences of the black shape with respect to the transformation class $\FTwoSTR$. The set of four shapes collectively form the occurrence set of the black shape in this point set with respect to $\FTwoSTR$---i.e., there are no other occurrences of the black shape in this dataset with respect to $\FTwoSTR$ apart from the four indicated.


An occurrence set, $\occsetfun(P,D,F)$, can be precisely specified by providing an explicit listing of the points in 
$P$, along with a set, $\mathit\Psi$, of transformation parameter vectors, specifying, for each occurrence of $P$, a transformation in $F$ that maps $P$ onto that occurrence. We call such an encoding an {\em ordered-pair encoding\/} of an occurrence set, and we define it as follows.
\begin{definition}[Ordered-pair encoding of an occurrence set]
Given a pattern, $P\subset\reals^k$, a dataset, $D\subset\reals^k$, and a $k$-dimensional transformation class, $F$, then the {\em ordered-pair encoding\/} of the occurrence set, $\occsetfun(P,D,F)$, is defined to be the unique ordered pair, $\langle P,\mathit\Psi\rangle$, that satisfies the following conditions:
\begin{enumerate}
\item $\mathit\Psi\subseteq\lbrace\sigma:\phi_F(\sigma,P)\subseteq D\rbrace$,
\item $|\mathit\Psi|=|\occsetfun(P,D,F)|$,
\item $\occsetfun(P,D,F)=\bigcup_{\sigma\in\mathit\Psi}\left\lbrace\phi_F(\sigma,P) \right\rbrace$, and
\item for each $\sigma\in\mathit\Psi$, there must not exist any transformation, $f\in F$, such that $\sigma_F(f)<\sigma$ and $f(P)=\phi_F(\sigma,P)$.
\end{enumerate}
\label{occ-set-ord-pair-def}
\end{definition}
Note that there can be more than one ordered pair, $\langle P,\mathit\Psi\rangle$, that satisfies the first three of the four conditions in Def.~\ref{occ-set-ord-pair-def}, since there can, in general, be more than one transformation in $F$ that maps a given pattern onto a particular occurrence of that pattern. The fourth condition ensures that each transformation parameter vector in an ordered-pair encoding is the {\em lexicographically least\/} vector whose corresponding transformation maps $P$ onto that particular occurrence.
For example, the ordered-pair encoding of the occurrence set in Fig.~\ref{occurrence-set-example} would be
\begin{equation}
\begin{split}
\langle&\lbrace\langle1,1\rangle,\langle1,2\rangle,\langle2,3\rangle,\langle4,1\rangle,\langle4,2\rangle\rbrace,\\ &\lbrace\langle-2, 13, 0, 1\rangle, \langle-1, 5, -7, -1\rangle, \langle1, 0, 0, 1\rangle, \langle2, 3, -8, -1\rangle\rbrace\rangle\,.
\end{split}
\label{P-Psi-encoding}
\end{equation}
Typically, an ordered-pair encoding gives us a shorter description of the occurrence set than one would obtain by using an extensional encoding where each point in each occurrence is explicitly recorded. For example, such an extensional encoding of the occurrence set in Fig.~\ref{occurrence-set-example} would look like this:
\begin{equation}
\begin{split}
\lbrace&\lbrace\langle1,1\rangle,\langle1,2\rangle,\langle2,3\rangle,\langle4,1\rangle,\langle4,2\rangle\rbrace,\\
&\lbrace\langle1,5\rangle,\langle1,6\rangle,\langle3,4\rangle,\langle4,5\rangle,\langle4,6\rangle\rbrace,\\
&\lbrace\langle5,1\rangle,\langle5,2\rangle,\langle9,3\rangle,\langle11,1\rangle,\langle11,2\rangle\rbrace,\\
&\lbrace\langle5,6\rangle,\langle5,7\rangle,\langle7,5\rangle,\langle11,6\rangle,\langle11,7\rangle\rbrace
\rbrace\,.
\end{split}
\label{extensional-encoding}
\end{equation}
More specifically, such an extensional description, in which each point in each occurrence is explicitly listed, requires $k|\occsetfun(P,D,F)|\times|P|$ real values to be recorded, where $k$ is the dimensionality of the points in $P$ and $D$---that is, $P,D\subset\reals^k$. For example, the encoding in \eqref{extensional-encoding} requires 40 independent numerical values to be recorded. On the other hand, the ordered-pair encoding, $\langle P,\mathit\Psi\rangle$, only requires $k|P|+\lambda(F)|\occsetfun(P,D,F)|$ values to be recorded, where $\lambda(F)$ is the parameter vector length  of $F$. Thus, in \eqref{P-Psi-encoding} there are only 26 independent values. Typically, for the transformation classes considered in this paper, the ordered-pair encoding of an occurrence set is much shorter than its extensional description.

\begin{figure}[t]
	\begin{center}
		\resizebox{10cm}{!}{
			\begin{tikzpicture}

			\draw[thick] (-.1,0) -- (12,0);
			\draw[thick] (0,-.1) -- (0,8);		
			
			\draw (12,-1) node{$x$};
			\draw (-1,8) node{$y$};
			
			\foreach \x in {0,1,...,12} {
				\draw[very thin] (\x,0) -- (\x,8);
				\draw[thick] (\x,-.1) -- (\x,0);
				\draw (\x,-.2) node[anchor=north]{\x};
			}
			\foreach \x in {0,1,...,8} {
				\draw (-.2,\x) node[anchor=east]{\x};
				\draw[very thin] (0,\x) -- (12,\x);
				\draw[thick] (-.1,\x) -- (0,\x);
			}

\draw[thick] (1,1)--(1,2)--(2,3)--(4,2)--(4,1)--(1,1);

\draw[thick,red] (1,6)
--(1,5)
--(3,4)
--(4,5)
--(4,6)
-- cycle;

\draw[thick,blue] 
(5,1)
--(5,2)
--(9,3)
--(11,2)
--(11,1)
-- cycle;

\draw[thick,green!50!black] 
(5,6)
--(5,7)
--(11,7)
--(11,6)
--(7,5)
-- cycle;

\drawcross{1}{1}{black}
\drawcross{1}{2}{black}
\drawcross{1}{5}{black}
\drawcross{1}{6}{black}
\drawcross{2}{3}{black}
\drawcross{3}{4}{black}
\drawcross{4}{1}{black}
\drawcross{4}{2}{black}
\drawcross{4}{5}{black}
\drawcross{4}{6}{black}
\drawcross{5}{1}{black}
\drawcross{5}{2}{black}
\drawcross{5}{6}{black}
\drawcross{5}{7}{black}
\drawcross{7}{5}{black}
\drawcross{9}{3}{black}
\drawcross{11}{1}{black}
\drawcross{11}{2}{black}
\drawcross{11}{6}{black}
\drawcross{11}{7}{black}

			\end{tikzpicture}
		}
	\end{center}
	\caption{An example of an occurrence set. See main text for details.}
	\label{occurrence-set-example}
\end{figure}
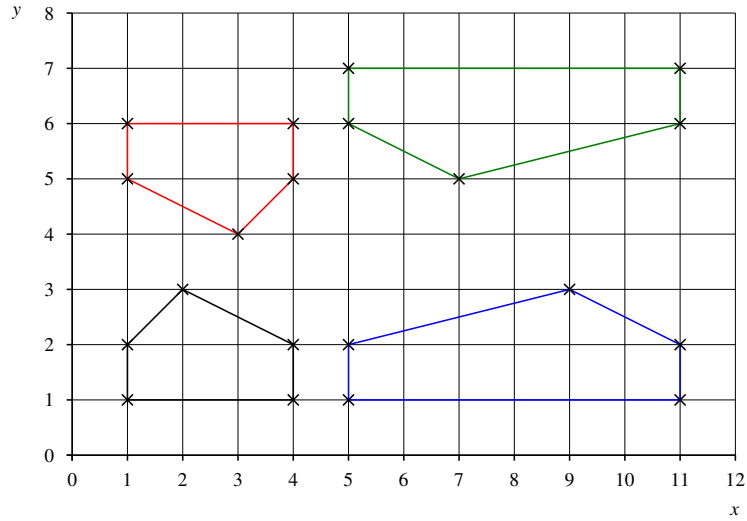

We define the {\em covered set\/} of an occurrence set as follows.
\begin{definition}[Covered set of an occurrence set]
If $\mathcal S$ is an occurrence set, then we define the {\em covered set\/} of $\mathcal S$ to be the point set,
\begin{equation}
\covsetfun(\mathcal S) = \bigcup_{Q\in\mathcal S}Q\,.
\label{covered-set-eq}
\end{equation}
We also define the {\em coverage\/} of $\mathcal S$ to be the cardinality of its covered set, $\left|\covsetfun(\mathcal S)\right|$.
\label{covered-set-definition}
\end{definition}
The covered set of an occurrence set, $\occsetfun(P,D,F)$, can therefore be recovered from its ordered-pair encoding, $\langle P,\mathit\Psi\rangle$, by using the following equation:
\begin{equation}
\covsetfun(\langle P,\mathit\Psi\rangle)=\bigcup_{\sigma\in\mathit\Psi}\phi_F(\sigma,P)\,.
\end{equation}
The covered set of the occurrence set in Fig.~\ref{occurrence-set-example} is therefore the entire set of 20 points shown in the figure, since every point occurs within an occurrence of the black pattern. The coverage of this occurrence set is therefore 20.

As pointed out above, the number of independent numerical values that need to be recorded when using an ordered-pair encoding of an occurrence set is typically much less than the number of values required to provide an extensional description of the occurrence set in which each point in each occurrence is explicitly listed. In fact, the ordered-pair encoding also often turns out to be much shorter than an extensional description of the {\em covered set\/} of an occurrence set---that is, a simple listing of the points in the covered set. This provides the basis of a compression-based approach to music analysis (introduced in \cite{MeredithArxiv2022,MeredithHCII2023}) in which we try to find the shortest description of a point set representing music (such as the one in Fig.~\ref{RAVEL-HAYDN-GROUND-TRUTH-CPO-GT}) that takes the form of a set of occurrence sets that collectively cover the point set, where each occurrence set is represented by its ordered-pair encoding. In order to do this, we therefore need a way of characterising how much compression can be obtained by representing an occurrence set by its ordered-pair encoding as opposed to extensionally listing the points that it covers. We therefore define the {\em compression factor\/} of an occurrence set as follows.

\begin{definition}[Compression factor of an occurrence set]
Given a pattern, $P\subset\reals^k$, a dataset, $D\subset\reals^k$ and a $k$-dimensional transformation class, $F$, we define the {\em compression factor\/} of the occurrence set, $\occsetfun(P,D,F)$, as follows:
\begin{equation}
\compfacfun(\occsetfun(P,D,F)) = \frac{k|\covsetfun(\occsetfun(P,D,F))|}{k|P|+\lambda(F)|\occsetfun(P,D,F)|}\,.
\label{compression-factor-eq}
\end{equation}
\label{compression-factor-definition}
\end{definition}
For example, the compression factor of the occurrence set in Fig.~\ref{occurrence-set-example} is $\frac{40}{26}\approx1.54$.

Given our interest in achieving short encodings of occurrence sets, it may seem inefficient to require that the set of transformation parameter vectors, $\mathit\Psi$, in an ordered-pair encoding must contain a transformation parameter vector for {\em every\/} occurrence of $P$ in the occurrence set. Why do we not instead choose $P$ to be one of the occurrences in the set, and then only provide parameter vectors for transformations that map this occurrence onto the other occurrences? This would reduce the number of real values in the encoding by $\lambda(F)$.

There are at least two reasons for this design choice. 
The first is that, in a pattern-matching setting, where we are interested in the occurrences of a query pattern, $Q$, in a dataset, $D$, it may be that $Q$ itself is not a subset of $D$. In this case, if we required that the reference pattern in the ordered-pair encoding must be an occurrence in the occurrence set, then we would have to use one of the occurrences of $Q$ in $D$ rather than $Q$ itself, which would result in $\mathit\Psi$ failing to provide us with the transformations required to map the query pattern onto its occurrences---which is typically what we are most interested in. The second reason will be given in the next section, where we generalize the notion of an occurrence set to {\em heterogeneous occurrence sets\/} that are defined with respect to {\em sets\/} of transformation classes.


\subsection{Heterogeneous occurrence sets}
\label{heterogeneous-occurrence-sets-section}
So far in this paper, we have assumed that we are only interested in matches and occurrences that are with respect to a {\em single\/} transformation class. But what if, in a particular context, it is more convenient to define the class of transformations that we are interested in as being the {\em union\/} of two or more transformation classes?
In this case, we might want to find matches, occurrences and occurrence sets with respect to a {\em set\/} of transformation classes, $\mathcal F=\lbrace F_1,F_2,\ldots,F_{|\mathcal F|}\rbrace$.

We therefore define the concept of a {\em heterogeneous occurrence set\/} as follows.
\begin{definition}[Heterogeneous occurrence set]
Given a pattern, $P\subset \reals^k$, a dataset, $D\subset \reals^k$, and a set of $k$-dimensional transformation classes, $\mathcal F=\lbrace F_1, F_2,\ldots,F_{|\mathcal F|}\rbrace$, we define the {\em heterogeneous occurrence set\/} of $P$ in $D$ with respect to $\mathcal F$ to be
\begin{equation}
\occsetfun(P,D,\mathcal F)=\left\lbrace f(P): f(P)\subseteq D \land\left(\exists F\in\mathcal F:f\in F\right)\right\rbrace\,.\label{het-occ-set-eq}\end{equation}\label{het-occ-set-def}\end{definition}\vskip-.2in
We allow ourselves to also use the term {\em homogeneous occurrence set\/} for an ordinary occurrence set that is defined with respect to a single transformation class (i.e., as defined in Def.~\ref{occurrence-set-def}), because, in this case, the transformations that transform the pattern onto its occurrences are ``homogeneous'' in the sense that they all come from the same transformation class. However, we only use the term ``homogeneous occurrence set'' when we want to contrast it with a heterogeneous occurrence set.

As explained in Section~\ref{occurrence-sets-section} above, an occurrence set with respect to a single transformation class, $\occsetfun(P,D,F)$, can be encoded in a compact manner as an ordered pair, $\langle P,\mathit{\Psi}\rangle$, where $\mathit\Psi$ is a set of transformation parameter vectors for a set of transformations in $F$ that map $P$ onto its occurrences in $D$. Generalizing this scheme to heterogeneous occurrence sets requires us to be able to handle a {\em set\/} of transformation classes, $\mathcal F=\lbrace F_1,F_2,\ldots,F_{|\mathcal F|}\rbrace$. We can do this by replacing $\mathit{\Psi}$ in the ordered-pair encoding of a homogeneous occurrence set with a {\em sequence\/} of sets of transformation parameter vectors, $\langle\mathit\Psi_1, \mathit\Psi_2,\ldots,\mathit\Psi_{|\mathcal F|}\rangle$, one set for each of the transformation classes in $\mathcal F$.
%
%
The first step in doing this is to construct a partially ordered set, $\mathbf \Phi$, obtained by sorting $\mathcal F$ so that the parameter vector length of a transformation class in $\mathbf \Phi$ is greater than or equal to the parameter vector lengths of all the classes occurring earlier in $\mathbf\Phi$. Note that, in general, there may be more than one ordering of the transformation classes in $\mathcal F$ that satisfies this partial order. The ordering used in any particular case must therefore be specified, and we denote  the zero-based index of $\tc$ in $\mathbf\Phi$ by $\gamma(\tc,\mathbf\Phi)$. That is, $\mathbf\Phi[\gamma(\tc,\mathbf\Phi)]=\tc$.

Suppose we have a set of transformation classes, $\mathcal F=\lbrace F_1, F_2, F_3\rbrace$, and suppose that $\mathbf\Phi=\langle F_1,F_2,F_3\rangle$, so that we know that $\lambda(F_3)\ge\lambda(F_2)\ge\lambda(F_1)$. Suppose further that we have a set of four occurrences, $\lbrace Q_1,Q_2, Q_3, Q_4\rbrace$, of a pattern, $P$, in a dataset, $D$, where $Q_i=f_i(P)$ and where $f_1,f_2\in F_1$, $f_3\in F_2$ and $f_4\in F_3$.
We can encode the set of occurrences $\left\lbrace Q_1,Q_2,Q_3,Q_4\right\rbrace$ as the ordered triple, 
\begin{equation}
\langle P, \mathbf\Phi, \langle\lbrace\sigma_{F_1}(f_1), \sigma_{F_1}(f_2)\rbrace, \lbrace \sigma_{F_2}(f_3)\rbrace, \lbrace \sigma_{F_3}(f_4)\rbrace \rangle\rangle\,,
\label{het-occ-set-ord-pair-ex}
\end{equation}
where the $n$th element of the sequence that forms the third element in this triple is the set of transformation parameter vectors for the occurrences of $P$ onto which $P$ is mapped by transformations in the $n$th element of $\mathbf\Phi$.
Note that, in general, there may exist one or more pairs of patterns, $Q_x,Q_y$, in $\occsetfun(P,D,\mathcal F)$ for which there exists {\em no\/} transformation in {\em any\/} of the transformation classes in $\mathcal F$ that map $Q_x$ onto $Q_y$. This is one reason why we choose to use the query pattern, $P$, itself rather than an occurrence of $P$ in $D$ as the first element in the ordered-pair encoding of an occurrence set, regardless of whether $P$ is a subset of $D$. 

\sloppypar As illustrated in \eqref{het-occ-set-ord-pair-ex}, a heterogeneous occurrence set, $\occsetfun(P,D,\mathcal F)$, can be encoded in a compact manner as an ordered triple, $\langle P,\mathbf\Phi,\mathbf Z\rangle$, where $\mathbf Z=\langle\mathit\Psi_1,\mathit\Psi_2,\ldots,\mathit\Psi_{\len(\mathbf\Phi)}\rangle$, and where $\mathbf Z[i]$ contains the transformation parameter vectors for transformations in $\mathbf\Phi[i]$ that map $P$ onto some of its occurrences. 
Recall that $\mathbf\Phi$ contains all and only the transformation classes in $\mathcal F$, ordered so that $\lambda(F_1)<\lambda(F_2)\Rightarrow\gamma(F_1,\mathbf\Phi)<\gamma(F_2,\mathbf\Phi)$ for all $F_1,F_2\in\mathcal F$. 
We call such an encoding simply an {\em ordered-triple encoding of a heterogeneous occurrence set}.
For each occurrence, $Q$, of $P$ in $\occsetfun(P,D,\mathcal F)$ there exists exactly one transformation parameter vector, $\sigma_F(f)$ in $\mathbf Z[\gamma(F,\mathbf\Phi)]$, such that $Q=f(P)$.
A heterogeneous occurrence set can therefore be recovered from its ordered-triple encoding, $\langle P,\mathbf\Phi,\mathbf Z\rangle$, using the formula,
$$
\bigcup_{i=0}^{\len(\mathbf\Phi)-1}\left(\bigcup_{\sigma\in\mathbf Z[i]}\lbrace\phi_{\mathbf\Phi[i]}(\sigma,P)\rbrace\right)\,.
$$
Since specifying a sequence, $\mathbf\Phi$, of the transformation classes in $\mathcal F$ is essential in order to do anything useful with a heterogeneous occurrence set, we sometimes write $\occsetfun(P,D,\mathbf\Phi)$ instead of $\occsetfun(P,D,\mathcal F)$.
On an analogy with the concept of the covered set of an occurrence set (see Def.~\ref{covered-set-definition}), we define the covered set of a heterogeneous occurrence set as follows.
\begin{definition}[Covered set of a heterogeneous occurrence set]
If $\mathcal S$ is a heterogeneous occurrence set, then we define the {\em covered set\/} of $\mathcal S$ to be the point set,
\begin{equation}
\covsetfun(\mathcal S) = \bigcup_{Q\in\mathcal S}Q\,.
\label{het-occ-set-cov-set-eq}
\end{equation}
We also define the {\em coverage\/} of $\mathcal S$ to be the cardinality of its covered set, $\left|\covsetfun(\mathcal S)\right|$.
\label{het-occ-set-cov-set-def}
\end{definition}
Note that the definition of the covered set of a heterogeneous occurrence set is essentially identical to that of an occurrence set.

The covered set of a heterogeneous occurrence set, $\occsetfun(P,D,\mathbf\Phi)$, can be recovered from its ordered-triple encoding, $\langle P,\mathbf\Phi, \mathbf Z\rangle$, using the following equation:
\begin{equation}
\covsetfun(\langle P,\mathbf\Phi,\mathbf Z\rangle) = \bigcup_{i=0}^{\len(\mathbf\Phi)-1}\left(\bigcup_{\sigma\in\mathbf Z[i]}\phi_{\mathbf\Phi[i]}(\sigma,P)\right)\,.
\end{equation}

On an analogy with the compression factor of an occurrence set, as defined in Def.~\ref{compression-factor-definition}, we define the compression factor of a heterogeneous occurrence set as follows.

\begin{definition}[Compression factor of a heterogeneous occurrence set]
Given a pattern, $P\subset\reals^k$, a dataset, $D\subset\reals^k$ and a set of $k$-dimensional transformation classes, $\mathcal F$, we define an ordered set of transformation classes, $\mathbf\Phi$, satisfying the following constraints:
\begin{enumerate}
\item $F\in\mathcal F\Leftrightarrow F\in\mathbf\Phi$,
\item $\len(\mathbf\Phi)=|\mathcal F|$, and
\item $\forall F_1, F_2\in\mathcal F, \lambda(F_1)<\lambda(F_2)\Rightarrow\gamma(F_1,\mathbf\Phi)<\gamma(F_2,\mathbf\Phi)$, where $\gamma(F,\mathbf\Phi)$ is the zero-based index of $F$ in $\mathbf\Phi$.
\end{enumerate}
We then derive an ordered-triple encoding of $\occsetfun(P,D,\mathbf\Phi)$ that we denote by $\langle P,\mathbf\Phi,\mathbf Z\rangle$.
We then define the {\em compression factor\/} of the heterogeneous occurrence set, $\occsetfun(P,D,\mathbf\Phi)$, as follows:
\begin{equation}
\compfacfun(\occsetfun(P,D,\mathbf\Phi)) = \frac{k|\covsetfun(\occsetfun(P,D,\mathbf\Phi))|}{k|P|+\len(\mathbf\Phi)+\sum_{i = 0}^{\len(\mathbf \Phi)-1}\len(\mathbf Z[i])\lambda(\mathbf\Phi[i])}\,.
\label{het-occ-set-cf-eq}
\end{equation}
\label{het-occ-set-cf-def}
\end{definition}

It is sometimes (but not always) possible to obtain a more compact encoding of an occurrence set using a heterogeneous occurrence set than it is by using an ordinary, homogeneous occurrence set. This can happen if there are a sufficient number of occurrences of the query pattern where two transformations from different transformation classes with different parameter vector lengths  map the query pattern onto {\em the same\/} occurrence. In such cases, we can choose to encode the occurrence using the least complex transformation---i.e., the one with the shortest parameter vector length.

In the following, we first consider a case where encoding a set of occurrences as a heterogeneous occurrence set does {\em not\/} lead to more compression. Then we show that if the occurrence set is augmented with further occurrences that are related to the query pattern by simple transformations, then encoding as a heterogeneous occurrence set {\em can\/} lead to better compression. 

Suppose we encode the occurrence set in Fig.~\ref{occurrence-set-example} as an ordered-triple encoding of a heterogeneous occurrence set, $\langle P,\mathbf\Phi, \mathbf Z\rangle$, as follows. We first define that $\mathcal F=\lbrace \FTwoT,\FTwoST,\FTwoSTR\rbrace$ and that $\mathbf\Phi=\langle\FTwoT,\FTwoST,\FTwoSTR\rangle$. In this case, $P$ is the black pattern, $\lbrace\langle1,1\rangle,\langle1,2\rangle,\langle2,3\rangle,\langle4,1\rangle,\langle4,2\rangle\rbrace$. The black pattern, $P$, is mapped onto itself by a simple translation by $\langle0,0\rangle$. This occurrence can therefore be represented by the transformation parameter vector, $\langle0,0\rangle$, in $\mathbf Z[0]$. The blue pattern in Fig.~\ref{occurrence-set-example} can be obtained from the black pattern by first scaling by a factor of $-2$ and then translating by $\langle13,0\rangle$. There is no need for any inversion in this case, so we can represent the blue pattern as being the result of a transformation in $\FTwoST$, whose parameter vector length is 3. We can therefore represent the blue pattern by the transformation parameter vector, $\langle-2,13,0\rangle$, in $\mathbf Z[1]$. To obtain the red pattern from the black pattern, we need to scale by a factor of $-1$, then translate by $\langle5,-7\rangle$ and then reflect in the $x$-axis. This transformation is in $\FTwoSTR= \mathbf\Phi[2]$. We can therefore represent this occurrence by the transformation parameter vector, $\langle-1,5,-7,-1\rangle$, in $\mathbf Z[2]$. Finally, $P$ can be mapped onto the green pattern by the transformation in $\FTwoSTR=\mathbf\Phi[2]$ that has the transformation parameter vector, $\langle2,3,-8,-1\rangle$. This gives the following ordered-triple encoding for this heterogeneous occurrence set:
\begin{equation}
\begin{split}
\langle&\lbrace\langle1,1\rangle,\langle1,2\rangle,\langle2,3\rangle,\langle4,1\rangle,\langle4,2\rangle\rbrace,\\ 
&\langle\FTwoT,\FTwoST,\FTwoSTR\rangle,\\
&\langle\lbrace\langle0,0\rangle\rbrace,\lbrace\langle-2,13,0\rangle\rbrace,\lbrace\langle-1,5,-7,-1\rangle,\langle2,3,-8,-1\rangle\rbrace\rangle\rangle\,.
\end{split}
\label{P-Z-encoding}
\end{equation}
Using \eqref{het-occ-set-cf-eq}, we can compute the compression factor of this heterogeneous occurrence set to be $\frac{40}{10+3+13}=\frac{40}{26}\approx1.538$. In this case, therefore, using the ordered-triple encoding of the heterogeneous occurrence set does not result in any more compression than one can achieve by using the ordered-pair encoding of the homogeneous occurrence set given in \eqref{P-Psi-encoding}. However, if there were more occurrences of the pattern and these occurrences were related to the pattern by transformations in transformation classes with parameter vector lengths less than that of $\FTwoSTR$, then we may achieve better compression using a heterogeneous occurrence set than by using $\FTwoSTR$ alone. 

\begin{figure}[t]
	\begin{center}
		\resizebox{10cm}{!}{
			\begin{tikzpicture}

			\draw[thick] (-.1,0) -- (12,0);
			\draw[thick] (0,-.1) -- (0,11);		
			
			\draw (12,-1) node{$x$};
			\draw (-1,11) node{$y$};
			
			\foreach \x in {0,1,...,12} {
				\draw[very thin] (\x,0) -- (\x,11);
				\draw[thick] (\x,-.1) -- (\x,0);
				\draw (\x,-.2) node[anchor=north]{\x};
			}
			\foreach \x in {0,1,...,11} {
				\draw (-.2,\x) node[anchor=east]{\x};
				\draw[very thin] (0,\x) -- (12,\x);
				\draw[thick] (-.1,\x) -- (0,\x);
			}

\draw[thick,cyan] (1,8)--(1,9)--(2,10)--(4,9)--(4,8)--(1,8);
\draw[thick,magenta] (5,9)--(5,10)--(6,11)--(8,10)--(8,9)--(5,9);

\draw[thick,red] (1,6)
--(1,5)
--(3,4)
--(4,5)
--(4,6)
-- cycle;

\draw[thick] 
(1,1)
--(1,2)
--(2,3)
--(4,2)
--(4,1)
-- cycle;

\draw[thick,blue] 
(5,1)
--(5,2)
--(9,3)
--(11,2)
--(11,1)
-- cycle;

\draw[thick,green!50!black] 
(5,6)
--(5,7)
--(11,7)
--(11,6)
--(7,5)
-- cycle;

\drawcross{1}{1}{black}
\drawcross{1}{2}{black}
\drawcross{1}{5}{black}
\drawcross{1}{6}{black}
\drawcross{2}{3}{black}
\drawcross{3}{4}{black}
\drawcross{4}{1}{black}
\drawcross{4}{2}{black}
\drawcross{4}{5}{black}
\drawcross{4}{6}{black}
\drawcross{5}{1}{black}
\drawcross{5}{2}{black}
\drawcross{5}{6}{black}
\drawcross{5}{7}{black}
\drawcross{7}{5}{black}
\drawcross{9}{3}{black}
\drawcross{11}{1}{black}
\drawcross{11}{2}{black}
\drawcross{11}{6}{black}
\drawcross{11}{7}{black}

\drawcross{1}{8}{black}
\drawcross{1}{9}{black}
\drawcross{2}{10}{black}
\drawcross{4}{8}{black}
\drawcross{4}{9}{black}

\drawcross{5}{9}{black}
\drawcross{5}{10}{black}
\drawcross{6}{11}{black}
\drawcross{8}{9}{black}
\drawcross{8}{10}{black}

			\end{tikzpicture}
		}
	\end{center}
	\caption{Another example of an occurrence set.}
	\label{occurrence-set-example-2}
\end{figure}
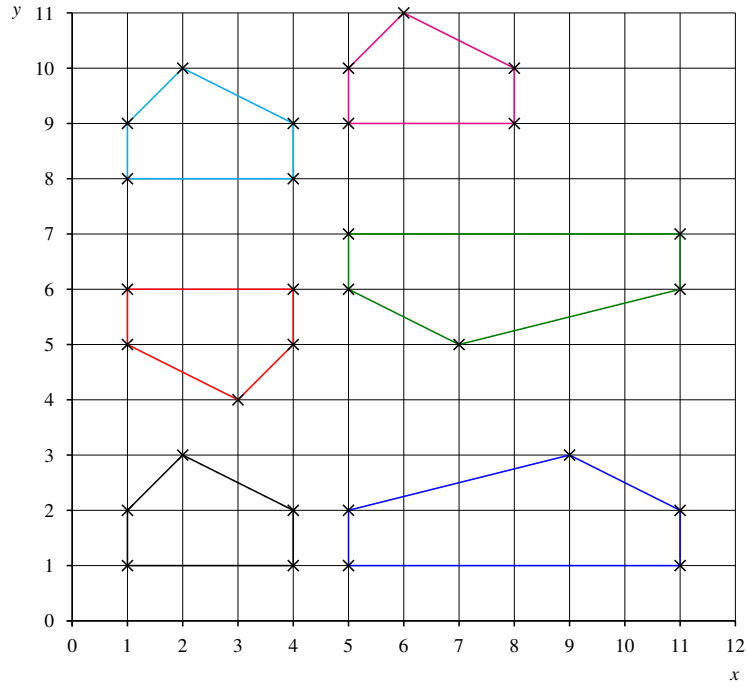

For example, in Fig.~\ref{occurrence-set-example-2} we have added two extra occurrences of the black pattern that are related to the black pattern by simple translations (i.e., transformations in $\FTwoT$). If we represent the six occurrences in Fig.~\ref{occurrence-set-example-2} as an ordinary occurrence set with respect to $\FTwoSTR$, then we get the ordered pair encoding,
\begin{equation}
\begin{split}
\langle&\lbrace\langle1,1\rangle,\langle1,2\rangle,\langle2,3\rangle,\langle4,1\rangle,\langle4,2\rangle\rbrace,\\ &\lbrace\langle-2, 13, 0, 1\rangle, \langle-1, 5, -7, -1\rangle, \langle1, 0, 0, 1\rangle, \langle1, 0, 7, 1\rangle,\langle1, 4, 8, 1\rangle,\langle2, 3, -8, -1\rangle\rbrace\rangle\,.
\end{split}
\label{P-Psi-encoding-of-occurrence-set-example-2}
\end{equation}
This encoding gives a compression factor of $\frac{60}{34}\approx 1.76$. However, if we set $\mathbf\Phi$ to be $\langle\FTwoT,\FTwoST,\FTwoSTR\rangle$, as above, then we can encode the six occurrences in Fig.~\ref{occurrence-set-example-2} as a heterogeneous occurrence set as follows:
{\small
\begin{equation}
\begin{split}
\langle&\lbrace\langle1,1\rangle,\langle1,2\rangle,\langle2,3\rangle,\langle4,1\rangle,\langle4,2\rangle\rbrace,\\ 
&\langle\FTwoT,\FTwoST,\FTwoSTR\rangle,\\
&\langle
	\lbrace
		\langle0, 0\rangle,
		\langle0, 7\rangle,
		\langle4,8\rangle
	\rbrace,
	\lbrace
		\langle-2, 13, 0\rangle
	\rbrace,
	\lbrace
		\langle-1, 5, -7, -1\rangle,
		\langle2, 3, -8, -1\rangle
	\rbrace
\rangle\,.
\end{split}
\label{P-Z-encoding-of-occurrence-set-example-2}
\end{equation}}
This ordered-triple encoding gives a compression factor of $\frac{60}{10+3+17}=\frac{60}{30}=2$. In this case, we can therefore achieve a higher compression factor by representing the set of occurrences as a heterogeneous occurrence set.

\section{Transformational equivalence}
\label{transformational-equivalence-section}

In 
\cite{MeredithArxivChapter12026}, inter-pattern transformation graphs were used to discover sets of occurrences of the HAYDN theme in Ravel's {\em Menuet\/} that are {\em transformationally equivalent}. For example, 
occurrences 3, 4, 8 and 9 form a set of transformationally equivalent patterns with respect to $\FTwoSTR$ in the CPO representation of the {\em Menuet}, shown in Fig.~\ref{RAVEL-HAYDN-GROUND-TRUTH-CPO-GT}. 

In this section, we begin by defining a number of concepts relating to transformational equivalence. We then consider the relationship between the group-theoretic properties of a transformation class and transformational equivalence. Finally, we identify which of the transformation classes defined in Section~\ref{section-transformation-classes} form groups under transformation composition and determine for which of these classes the occurrence sets form sets of transformationally equivalent patterns.\footnote{For definitions of standard mathematical terms relating to relations and equivalence classes, the reader is referred to \citet[pp.~1163--1166]{CormenEtAl2009}.}

\subsection{Definitions of concepts relating to transformational equivalence}

We begin by defining the {\em match relation\/} of a transformation class or set of transformation classes.
\begin{definition}[Match relation of a transformation class]
Let $F$ be a $k$-dimensional transformation class and let  $P,Q\subset\reals^k$. We define the binary relation, $\mathcal M_{F}$, which we call the {\em match relation of $F$}, as follows:
\begin{equation}
(P,Q)\in\mathcal M_{F}\Leftrightarrow\left(\exists f\in F: Q=f(P)\right)\,.
\label{match-rel-tc-eq}
\end{equation}
If $(P,Q)\in\mathcal M_F$, then we may write $P\rightarrow_F Q$ and $P\mathbin{{\mathcal M}_{F}}Q$ and we can say that $Q$ {\em is a match of\/} $P$ {\em with respect to\/} $F$. 
\label{match-rel-tc-def}
\end{definition}
Note that, if $F$ is a $k$-dimensional transformation class, then $\mathcal M_F\subset\left(2^{\reals^k}\times2^{\reals^k}\right)$, where $2^{\reals^k}$ is the universe of all possible sets of $k$-dimensional points (i.e., $2^{\reals^k}$ is the powerset of $\reals^k$).

\begin{definition}[Match relation of a set of transformation classes]
Let $\mathcal F$ be a set of $k$-dimensional transformation classes, and let  $P,Q\subset\reals^k$. We define the binary relation, $\mathcal M_{\mathcal F}$, which we call the {\em match relation of\/} $\mathcal F$, as follows:
\begin{equation}
(P,Q)\in\mathcal M_{\mathcal F}\Leftrightarrow\left(\exists F\in \mathcal F: P\rightarrow_F Q\right)\,.
\label{match-rel-tc-set-eq}
\end{equation}
If $(P,Q)\in\mathcal M_{\mathcal F}$, then we may write $P\rightarrow_{\mathcal F} Q$ or $P\mathbin{{\mathcal M}_{\mathcal F}}\thinspace Q$ and we can say that $Q$ {\em is a match of\/} $P$ {\em with respect to\/} $\mathcal F$. 
\label{match-rel-tc-set-def}
\end{definition}

\begin{definition}[Transformational equivalence]
Suppose we have a match relation, $\mathcal M_X$, where $X$ is either a $k$-dimensional transformation class or a set of $k$-dimensional transformation classes. We say that two patterns, $P,Q\subset\reals^k$, are {\em transformationally equivalent with respect to} $X$, denoted by $P\equiv_X Q$, if and only if $P\rightarrow_XQ$ and $Q\rightarrow_XP$.
\label{tran-eq-def}
\end{definition}

\begin{definition}[Set of transformationally equivalent patterns]
Suppose we have a match relation, $\mathcal M_X$, where $X$ is either a $k$-dimensional transformation class or a set of $k$-dimensional transformation classes. We say that a set of patterns, $\mathcal P=\lbrace P_1,P_2,\ldots P_{\ell}\rbrace$, where $P_i\subset\reals^k$ for all $P_i\in\mathcal P$, is a {\em set of patterns that are transformationally equivalent with respect to\/} $X$, if and only if $P_i\equiv_XP_j$ for all $P_i,P_j\in\mathcal P$.
\label{set-of-tran-eq-patterns-def}
\end{definition}

\begin{definition}[Transformational equivalence relation]
Suppose we have a match relation, $\mathcal M_X$, where $X$ is either a $k$-dimensional transformation class or a set of $k$-dimensional transformation classes.
If $\mathcal M_X$ is reflexive, symmetric and transitive, then it is, by definition, an equivalence relation in the general mathematical sense and we say that $\mathcal M_X$ is a {\em transformational equivalence relation}.
\label{tran-eq-rel-def}
\end{definition}

\begin{definition}[Transformational equivalence class]
Suppose that $\mathcal M_X$ is a transformational equivalence relation, where $X$ is either a $k$-dimensional transformation class or a set of $k$-dimensional transformation classes.
If $P\subset\reals^k$, then we define the {\em transformational equivalence class\/} of $P$ with respect to $X$, denoted by $\left[P\right]_X$, to be the set,
\begin{equation}
\left[P\right]_X=\left\lbrace Q\subset\reals^k: P\equiv_X Q\right\rbrace\,.
\label{tran-eq-class-eq}
\end{equation}
Note that, because $\mathcal M_X$ is an equivalence relation in the general mathematical sense, it follows that $\left[P\right]_X$ is also an equivalence class in the general mathematical sense. 
\label{tran-eq-class-def}
\end{definition}

It is well known \citep[see, e.g., ][p.~1164]{CormenEtAl2009} that an equivalence relation is the same as a partition. We can therefore simply state the following lemma.

\begin{lemma}
A transformational equivalence relation, $\mathcal M_X$, where $X$ is either a $k$-dimensional transformation class or a set of $k$-dimensional transformation classes, partitions the set $2^{\reals^k}$ exhaustively and exclusively into transformational equivalence classes.
\label{tran-eq-rel-part-lemma}\end{lemma}
\subsection{Transformation groups}

Let us now consider the case where a $k$-dimensional transformation class, $F$, forms a {\em group\/} under the operation of transformation composition, $\circ$, where we define that, for any pair of transformations, $f,g$, $(f\circ g)(p)=f(g(p))$ for any point, $p\in\reals^k$. That is, we now consider the case where the following constraints on a transformation class, $F$, are satisfied:
\begin{enumerate}
	\item {\bfseries Closure\/}: for all $f,g\in F$, $f\circ g\in F$.
	\item {\bfseries Associativity\/}: for all $f,g,h\in F$, $(f\circ g)\circ h=f\circ(g\circ h)$.
	\item {\bfseries Presence of a unique, universal identity element\/}: there exists a unique identity element, $e\in F$, such that, for all $f\in F$, $e\circ f=f\circ e=f$ holds.
	\item {\bfseries Invertibility\/}: for every transformation, $f\in F$, there exists a unique {\em inverse\/} transformation, $f^{-1}\in F$, such that $f\circ f^{-1}=f^{-1}\circ f=e$.
\end{enumerate}
We are now in a position to define the concept of a {\em transformation group\/}:
\begin{definition}[Transformation group]
A set of transformations is defined to be a {\em transformation group}, if and only if it forms a group under transformation composition.
\end{definition}
We now consider the following well-known theorem that relates transformation groups to equivalence relations.
\begin{theorem}
If $S$ is a set of functions over a set, $A$ (i.e., $S\subseteq \lbrace f: (f: A\rightarrow A\rbrace)\rbrace$), and $S$ forms a group under the operation of function composition, denoted by $\circ$, then the binary relation, $R$, defined as $aRb\Leftrightarrow(\exists f\in S: b=f(a))$ where $a,b\in A$, is an equivalence relation over the set $A$.
\label{thm-tran-grp-implies-eq-rel}
\end{theorem}
\begin{proof}
In order to be an equivalence relation, $R$ must be reflexive, symmetric and transitive.
\begin{enumerate}
\item {\bfseries Reflexivity}: If $S$ forms a group under the function composition operation, $\circ$, then, by definition, $S$ contains a unique identity element, $e\in S$, such that $e(a)=a$ for all $a\in A$. Therefore $aRa$ for all $a\in A$ and therefore $R$ is reflexive.

\item {\bfseries Symmetry}: If $a,b\in A$ and $aRb$, then there must exist a function, $f\in S$, such that $b=f(a)$. If $S$ forms a group under function composition, then, for every function, $f\in S$, the inverse of $f$, $f^{-1}$, is also an element in $S$. Thus, since $a=f^{-1}(b)$, $bRa$ and therefore $R$ is symmetric.

\item {\bfseries Transitivity}: Suppose that $g,h\in S$, that $a,b,c\in A$, that $b=g(a)$ and that $c=h(b)$. Therefore $aRb$ and $bRc$. But $c=h(g(a))=(h\circ g)(a)$ and, if $S$ forms a group under function composition, then $h\circ g\in S$ and therefore $aRc$ and $R$ is transitive.
\end{enumerate}
$R$ is therefore an equivalence relation over $A$.\qed
\end{proof}

The following corollary follows directly from Theorem~\ref{thm-tran-grp-implies-eq-rel}.
\begin{corollary}
If $F$ is a $k$-dimensional transformation class, and $F$ forms a group under the operation of transformation composition, then $\mathcal M_F$ is a transformational equivalence relation.
\label{match-rel-of-tran-grp-is-eq-rel}
\end{corollary}
\begin{proof}
The result follows directly from  Theorem~\ref{thm-tran-grp-implies-eq-rel}, together with the definitions of a transformation class (Def.~\ref{transformation-class-definition}) and a transformational equivalence relation (Def.~\ref{tran-eq-rel-def}). \qed
\end{proof}

However, the converse of Theorem~\ref{thm-tran-grp-implies-eq-rel} is not generally true. This can be shown by finding a single counter-example. For example, suppose that $A=\mathbb{Z}$, that $S=\lbrace f_i : f_i(n)=i\text{ for all }n\in\mathbb{Z},\text{ where } i\in\mathbb{Z}\rbrace$ and that $R$ is the binary relation, $aRb\Leftrightarrow(\exists f\in S: b=f(a))$ where $a,b\in A$. We first prove that $R$ is an equivalence relation by showing that it is reflexive, symmetric and transitive:
\begin{enumerate}
\item {\bfseries Reflexivity}: For any $a\in A$, $f_a(a)=a$ and $f_a\in S$. Therefore $aRa$ for all $a\in A$ and $R$ is reflexive.
\item {\bfseries Symmetry}: If $a,b\in A$, then $f_b(a)=b$ and therefore $aRb$. Also, $f_a(b)=a$ and therefore $bRa$. So $R$ is symmetric.
\item {\bfseries Transitivity}: If $a,b,c\in A$, then $f_b(a)=b$, $f_c(b)=c$ and $f_c(a)=c$ and $f_a,f_b,f_c\in S$. Therefore, if $aRb$ and $aRc$, then $aRc$ for all $a,b,c\in A$ and $R$ is transitive.
\end{enumerate}
$R$ is therefore reflexive, symmetric and transitive and $R$ is therefore an equivalence relation over $A$. However, $S$ does {\em not\/} form a group under function composition, as it does not contain a unique universal identity element and the functions in $S$ do not have inverses:
\begin{enumerate}
\item {\bfseries Presence of a unique, universal identity element\/}: For any given value of $n\in A$, there is only one function, $f_n\in S$ such that $f_n(n)=n$. However, $f_n(x)\neq x$ for all $x\neq n$. There therefore does not exist any unique, universal identity element, $e$ in $S$ such that $e(k)=k$ for all $k\in A$. $S$ therefore does {\em not\/} satisfy this condition and is therefore {\em not\/} a group.
\item {\bfseries Invertibility\/}: Each of the functions, $f_i$, in $S$ is a many-to-one function and therefore does not have a unique inverse function. $S$ therefore does not satisfy the invertibility condition and is therefore {\em not\/} a group.
\end{enumerate}

Therefore, if $R$ is an equivalence relation, and $aRb\Leftrightarrow(\exists f\in S: b=f(a))$ where $a,b\in A$, and each element in $S$ is a function, $f:A\rightarrow A$, then we cannot guarantee that $S$ will form a group under function composition.

\subsection{Group-theoretic and transformational equivalence properties of selected transformation classes}

In this section, we identify which of the transformation classes considered in Section~\ref{section-transformation-classes} form transformation groups. We also identify for which of these classes occurrence sets always form sets of transformationally equivalent patterns.

\subsubsection{$\FTwoT$}

We begin by considering the set of all 2-dimensional translations, $\FTwoT$.

\begin{lemma}
$\FTwoT$ is a transformation group.
\label{ftwot-transformation-group}
\end{lemma}
\begin{proof}
In order to prove that $\FTwoT$ is a transformation group, we need to prove that it satisfies the conditions of closure, associativity, presence of a universal identity element, and invertibility.
\begin{enumerate}
\item{\bfseries Closure}: Let $f,g\in\FTwoT$, let $p\in\reals^2$ and let $h(p)=f(g(p))$. This implies that $h(p)=(f\circ g)(p)$. From \eqref{ftwot-phi}, it follows that 
\begin{eqnarray}
f(g(p))=(f\circ g)(p) &=& \phi_{\FTwoT}(\sigma_{\FTwoT}(f),\phi_{\FTwoT}(\sigma_{\FTwoT}(g),p))\nonumber\\
&=&p+\sigma_{\FTwoT}(g)+\sigma_{\FTwoT}(f)\nonumber\\
&=&\phi_{\FTwoT}(\sigma_{\FTwoT}(g)+\sigma_{\FTwoT}(f),p)
\label{ftwot-closure}
\end{eqnarray}
Therefore, $\sigma_{\FTwoT}(h)=\sigma_{\FTwoT}(g)+\sigma_{\FTwoT}(f)$ and so $h\in\FTwoT$. $\FTwoT$ is therefore closed under transformation composition.
\item{\bfseries Associativity}: Let $f,g,h\in\FTwoT$ and let $p\in\reals^2$. From \eqref{ftwot-closure},
\begin{eqnarray}
(f\circ g)(p)&=&\phi_{\FTwoT}(\sigma_{\FTwoT}(g)+\sigma_{\FTwoT}(f),p)\text{ and }\nonumber\\
(g\circ h)(p)&=&\phi_{\FTwoT}(\sigma_{\FTwoT}(h)+\sigma_{\FTwoT}(g),p)\,.\nonumber
\end{eqnarray}
Therefore, from \eqref{ftwot-phi},
\begin{eqnarray}
(f\circ g)(h(p)) &=& \phi_{\FTwoT}(\sigma_{\FTwoT}(g)+\sigma_{\FTwoT}(f),\phi_{\FTwoT}(\sigma_{\FTwoT}(h),p))\nonumber\\
&=&p+\sigma_{\FTwoT}(h)+\sigma_{\FTwoT}(g)+\sigma_{\FTwoT}(f)\label{ftwot-assoc-1}
\end{eqnarray}
and
\begin{eqnarray}
f((g\circ h)(p)) &=& \phi_{\FTwoT}(\sigma_{\FTwoT}(f),\phi_{\FTwoT}(\sigma_{\FTwoT}(h)+\sigma_{\FTwoT}(g),p))\nonumber\\
&=&p+\sigma_{\FTwoT}(h)+\sigma_{\FTwoT}(g)+\sigma_{\FTwoT}(f)\label{ftwot-assoc-2}
\end{eqnarray}
From \eqref{ftwot-assoc-1} and \eqref{ftwot-assoc-2} it follows that $(f\circ g)(h(p)) = f((g\circ h)(p))$ and that therefore transformation composition is associative over $\FTwoT$.
\item{\bfseries Identity element} If $\sigma_{\FTwoT}(e)=\langle0,0\rangle$ and $p\in\reals^2$, then, from \eqref{ftwot-phi}, $e(p)=p$, and $e$ is a universal identity element within $\FTwoT$.
\item{\bfseries Invertibility} Let $\sigma_{\FTwoT}(g)=-\sigma_{\FTwoT}(f)$. From \eqref{ftwot-phi} and \eqref{ftwot-closure} it follows that 
\begin{eqnarray}
(f\circ g)(p) = (g \circ f)(p) &=& \sigma_{\FTwoT}(\sigma_{\FTwoT}(g)+\sigma_{\FTwoT}(f),p)\nonumber\\
&=& \sigma_{\FTwoT}(-\sigma_{\FTwoT}(f)+\sigma_{\FTwoT}(f),p)\nonumber\\
&=&\sigma_{\FTwoT}(\langle0,0\rangle,p)=p\,.\nonumber
\end{eqnarray}
 Therefore, in general, $\sigma_{\FTwoT}(f^{-1})=-\sigma_{\FTwoT}(f)$. Every transformation, $f\in\FTwoT$, therefore has a unique inverse.
\end{enumerate}
$\FTwoT$ therefore satisfies the necessary conditions for it to be a transformation group.\qed
\end{proof}

\begin{corollary}
$\mathcal M_{\FTwoT}$ is a transformational equivalence relation.
\label{m-ftwot-tran-eq-rel}
\end{corollary}
\begin{proof}
This follows directly from Lemma~\ref{ftwot-transformation-group} and Corollary~\ref{match-rel-of-tran-grp-is-eq-rel}.\qed
\end{proof}

\begin{corollary}
If $P,D\subset\reals^2$, then $\occsetfun(P,D,\FTwoT)$ is a set of transformationally equivalent patterns.
\end{corollary}
\begin{proof}
From Corollary~\ref{m-ftwot-tran-eq-rel}, we know that $\mathcal M_{\FTwoT}$ is a transformational equivalence relation, and therefore an equivalence relation over $2^{\mathbb R^2}$. From Lemma~\ref{tran-eq-rel-part-lemma}, it follows that $\mathcal M_{\FTwoT}$ partitions $2^{\mathbb R^2}$ exhaustively and exclusively into transformational equivalence classes. It follows that $\occsetfun(P,D,\FTwoT)\subset [P]_{\FTwoT}$ and that therefore $\occsetfun(P,D,\FTwoT)$ is a set of transformationally equivalent patterns.\qed
\end{proof}

\subsubsection{$\FTwoTR$}

We now consider, $\FTwoTR$, the set of 2-dimensional transformations consisting of a translation, optionally followed by a reflection in the $x$-axis.

\begin{lemma}
$\FTwoTR$ is a transformation group.
\label{ftwotr-transformation-group}
\end{lemma}
\begin{proof}
In order to prove that $\FTwoTR$ is a transformation group, we need to prove that it satisfies the conditions of closure, associativity, presence of a universal identity element, and invertibility.
\begin{enumerate}
\item{\bfseries Closure}: Let $f,g\in\FTwoTR$, $p\in\reals^2$, $p=\langle p_x, p_y\rangle$ and $h(p)=f(g(p))$. This implies that $h(p)=(f\circ g)(p)$. Let $\sigma^f = \sigma_{\FTwoTR}(f)$ and let $\sigma^g = \sigma_{\FTwoTR}(g)$. Let $\sigma^x_i=\sigma^x[i]$. From \eqref{phi_F2TR}, it follows that 
\begin{eqnarray}
f(g(p))=(f\circ g)(p) &=& \phi_{\FTwoTR}(\sigma^f,\phi_{\FTwoTR}(\sigma^g,p))\nonumber\\
&=&\phi_{\FTwoTR}(\sigma^f,\langle p_x+\sigma^g_0,\sigma^g_2(p_y+\sigma^g_1)\rangle)\nonumber\\
&=&\langle p_x+\sigma^g_0+\sigma^f_0,\sigma^f_2(\sigma^g_2(p_y+\sigma^g_1)+\sigma^f_1)\rangle\nonumber\\
&=&\langle p_x+\sigma^g_0+\sigma^f_0, \sigma^f_2\sigma^g_2p_y+\sigma^f_2\sigma^g_2\sigma^g_1+\sigma^f_2\sigma^f_1\rangle\nonumber\\
&=&\left\langle p_x+\left(\sigma^g_0+\sigma^f_0\right), \sigma^f_2\sigma^g_2\left(p_y+\left(\sigma^g_1+\frac{\sigma^f_1}{\sigma^g_2}\right)\right)\right\rangle\nonumber\\
&=&\phi_{\FTwoTR}\left(\left\langle\sigma^g_0+\sigma^f_0,\sigma^g_1+\frac{\sigma^f_1}{\sigma^g_2},\sigma^f_2\sigma^g_2\right\rangle,p\right)
\label{ftwotr-closure}
\end{eqnarray}
Therefore, $$\sigma_{\FTwoTR}(h)=\left\langle\sigma^g_0+\sigma^f_0,\sigma^g_1+\frac{\sigma^f_1}{\sigma^g_2},\sigma^f_2\sigma^g_2\right\rangle$$ and so $h\in\FTwoTR$. $\FTwoTR$ is therefore closed under transformation composition.
\item{\bfseries Associativity}: Let $f,g,h\in\FTwoTR$, $p\in\reals^2$ and $p=\langle p_x, p_y\rangle$. 
Let $\sigma^f = \sigma_{\FTwoTR}(f)$, $\sigma^g = \sigma_{\FTwoTR}(g)$ and $\sigma^h = \sigma_{\FTwoTR}(h)$. 
Let $\sigma^x_i=\sigma^x[i]$. 
From \eqref{ftwotr-closure},
\begin{eqnarray}
(f\circ g)(p)&=&\phi_{\FTwoTR}\left(\left\langle\sigma^g_0+\sigma^f_0,\sigma^g_1+\frac{\sigma^f_1}{\sigma^g_2},\sigma^f_2\sigma^g_2\right\rangle,p\right)\label{ftwotr-fog}\\
(g\circ h)(p)&=&\phi_{\FTwoTR}\left(\left\langle\sigma^h_0+\sigma^g_0,\sigma^h_1+\frac{\sigma^g_1}{\sigma^h_2},\sigma^g_2\sigma^h_2\right\rangle,p\right)\label{ftwotr-goh}
\end{eqnarray}
Therefore, from \eqref{phi_F2TR},
\begin{eqnarray}
(f&\circ& g)(h(p))\nonumber\\
&=& \phi_{\FTwoTR}\left(\left\langle\sigma^g_0+\sigma^f_0,\sigma^g_1+\frac{\sigma^f_1}{\sigma^g_2},\sigma^f_2\sigma^g_2\right\rangle,\phi_{\FTwoTR}\left(\sigma_{\FTwoTR}(h),p\right)\right)\nonumber\\
&=& \phi_{\FTwoTR}\left(\left\langle\sigma^g_0+\sigma^f_0,\sigma^g_1+\frac{\sigma^f_1}{\sigma^g_2},\sigma^f_2\sigma^g_2\right\rangle,
				\left\langle p_x+\sigma^h_0,\sigma^h_2\left(p_y+\sigma^h_1\right)\right\rangle\right)\nonumber\\
&=& \left\langle p_x+\sigma^h_0+\sigma^g_0+\sigma^f_0,\sigma^f_2\sigma^g_2\left(\sigma^h_2\left(p_y+\sigma^h_1\right)+\sigma^g_1+\frac{\sigma^f_1}{\sigma^g_2}\right)\right\rangle\nonumber\\
&=& \left\langle p_x+\sigma^h_0+\sigma^g_0+\sigma^f_0,\sigma^f_2\sigma^g_2\sigma^h_2\left(p_y+\sigma^h_1+\frac{\sigma^g_1}{\sigma^h_2}+\frac{\sigma^f_1}{\sigma^g_2\sigma^h_2}\right)\right\rangle\label{ftwotr-assoc-1}
\end{eqnarray}
and
\begin{eqnarray}
f((g&\circ& h)(p))\nonumber\\
&=& \phi_{\FTwoTR}\left(\sigma_{\FTwoTR}(f),(g\circ h)(p)\right)\nonumber\\
&=&\phi_{\FTwoTR}\left(\left\langle\sigma^f_0,\sigma^f_1,\sigma^f_2\right\rangle,\phi_{\FTwoTR}\left(\left\langle\sigma^h_0+\sigma^g_0,\sigma^h_1+\frac{\sigma^g_1}{\sigma^h_2},\sigma^g_2\sigma^h_2\right\rangle,p\right)\right)\nonumber\\
&=&\phi_{\FTwoTR}\left(\left\langle\sigma^f_0,\sigma^f_1,\sigma^f_2\right\rangle,\left\langle p_x+\sigma^h_0+\sigma^g_0,\sigma^g_2\sigma^h_2\left(p_y+\sigma^h_1+\frac{\sigma^g_1}{\sigma^h_2}\right)\right\rangle\right)\nonumber\\
&=&\left\langle p_x+\sigma^h_0+\sigma^g_0+\sigma^f_0,\sigma^f_2\left(\sigma^g_2\sigma^h_2\left(p_y+\sigma^h_1+\frac{\sigma^g_1}{\sigma^h_2}\right)+\sigma^f_1\right)\right\rangle\nonumber\\
&=& \left\langle p_x+\sigma^h_0+\sigma^g_0+\sigma^f_0,\sigma^f_2\sigma^g_2\sigma^h_2\left(p_y+\sigma^h_1+\frac{\sigma^g_1}{\sigma^h_2}+\frac{\sigma^f_1}{\sigma^g_2\sigma^h_2}\right)\right\rangle
\label{ftwotr-assoc-2}
\end{eqnarray}
From \eqref{ftwotr-assoc-1} and \eqref{ftwotr-assoc-2} it follows that $(f\circ g)(h(p)) = f((g\circ h)(p))$ and that therefore transformation composition is associative over $\FTwoTR$.
\item{\bfseries Identity element} If $\sigma_{\FTwoTR}(e)=\langle0,0,1\rangle$, $p\in\reals^2$ and $p=\langle p_x,p_y\rangle$, then, from \eqref{phi_F2TR}, 
\begin{eqnarray}
e(p)&=&\phi_{\FTwoTR}(\langle0,0,1\rangle,\langle p_x,p_y\rangle)\nonumber\\
&=&\langle p_x+0,1(p_y+0)\rangle\nonumber\\
&=&\langle p_x,p_y\rangle=p\nonumber\,.
\end{eqnarray}
Therefore $e(p)=p$ for all $p\in\reals^2$. $e$ is therefore a universal identity element within $\FTwoTR$.
\item{\bfseries Invertibility} Let $q=\phi_{\FTwoTR}(\sigma,p)$, where $\sigma=\langle \sigma_0,\sigma_1,\sigma_2\rangle$, and let $q=\langle q_x,q_y\rangle$ and $p=\langle p_x, p_y\rangle$. From \eqref{phi_F2TR}, it follows that 
$$q=\langle q_x,q_y\rangle=\langle p_x+\sigma_0,\sigma_2(p_y+\sigma_1)\rangle$$
and therefore
$$p=\langle p_x, p_y\rangle=\left\langle q_x - \sigma_0, \frac{q_y}{\sigma_2}-\sigma_1\right\rangle\,.$$
But $\sigma_2\in \lbrace1,-1\rbrace$, therefore $\frac1{\sigma_2}=\sigma_2$, so
\begin{eqnarray}
p&=&\langle p_x, p_y\rangle=\left\langle q_x - \sigma_0, \sigma_2q_y-\sigma_1\right\rangle\nonumber\\
&=&\left\langle q_x - \sigma_0,\sigma_2\left(q_y-\sigma_1\sigma_2\right)\right\rangle\nonumber\\
&=&\phi_{\FTwoTR}\left(\langle-\sigma_0,-\sigma_1\sigma_2, \sigma_2\rangle,q\right)\,.\label{f2tr-inverse-sigma}
\end{eqnarray}
Therefore, if $f\in\FTwoTR$ and $\sigma=\sigma_{\FTwoTR}(f)$ , then $$\sigma_{\FTwoTR}(f^{-1})=\langle-\sigma[0],-\sigma[1]\sigma[2], \sigma[2]\rangle\,.$$ Therefore every transformation in $\FTwoTR$ has a unique inverse.
\end{enumerate}
$\FTwoTR$ is therefore a transformation group.\qed
\end{proof}

\begin{corollary}
$\mathcal M_{\FTwoTR}$ is a transformational equivalence relation.
\label{m-ftwotr-tran-eq-rel}
\end{corollary}
\begin{proof}
This follows directly from Lemma~\ref{ftwotr-transformation-group} and Corollary~\ref{match-rel-of-tran-grp-is-eq-rel}.\qed
\end{proof}

\begin{corollary}
If $P,D\subset\reals^2$, then $\occsetfun(P,D,\FTwoTR)$ is a set of transformationally equivalent patterns.
\end{corollary}
\begin{proof}
From Corollary~\ref{m-ftwotr-tran-eq-rel}, we know that $\mathcal M_{\FTwoTR}$ is a transformational equivalence relation, and therefore an equivalence relation over $2^{\mathbb R^2}$. From Lemma~\ref{tran-eq-rel-part-lemma}, it follows that $\mathcal M_{\FTwoTR}$ partitions $2^{\mathbb R^2}$ exhaustively and exclusively into transformational equivalence classes. It follows that $\occsetfun(P,D,\FTwoTR)\subset [P]_{\FTwoTR}$ and that therefore $\occsetfun(P,D,\FTwoTR)$ is a set of transformationally equivalent patterns.\qed
\end{proof}

\subsubsection{$\FTwoST$ and $\FTwoST'$}

$\FTwoST$ was defined in Section~\ref{ftwost-section} to be the set containing those 2-dimensional transformations that consist of a scaling parallel to the $x$-axis by a real-valued scale factor (which may be zero), followed by a translation. Any transformation in $\FTwoST$ that has an $x$-axis scale factor of zero does not have a unique inverse. Therefore $\FTwoST$ is not a transformation group. 

Let $P,D\subset\reals^2$ and suppose that we have an occurrence set, $\occsetfun(P,D,\FTwoST)$, that contains an occurrence, $Q$, of $P$, such that the $x$-axis scale factor of the transformation that maps $P$ onto $Q$ is zero. Let $S$ be the subset of $\occsetfun(P,D,\FTwoST)$ that contains occurrences of $P$ in $D$ with respect to $\FTwoST$ that can be obtained from $P$ by transformations whose $x$-axis scale factor is {\em not\/} zero. There is then no transformation in $\FTwoST$ that maps $Q$ onto any of the occurrences in $S$, so if $S$ is not empty, then $\occsetfun(P,D,\FTwoST)$ is {\em not\/} a set of transformationally equivalent patterns. Moreover, this shows that $\mathcal M_{\FTwoST}$ is {\em not\/} a transformational equivalence relation.
\begin{figure}[t]
\begin{center}
\resizebox*{.7\columnwidth}{!}{
\begin{minipage}{\linewidth}
\begin{tabbing}
\input{ColumnDefinitions}
$\fun{$\BTSF_{\FTwoST'}$}(\objectBasis,\imageBasis)$\\
1\>$\dmif \imageBasis = \objectBasis$\\
2\>\>$\dmreturn\lbrace\langle1,0,0\rangle\rbrace\dmcomment{Identity transformation}$\\
3\>$p_1\becomes\objectBasis[0],p_2\becomes\objectBasis[1],q_1\becomes\imageBasis[0],q_2\becomes\imageBasis[1]$\\
4\>$\dmif q_2[0]\neq q_1[0]\land p_2[0]=p_1[0]\dmcomment{$\sigma[0]$ is undefined.}$\\
5\>\>$\dmreturn \lbrace\rbrace$\\
6\>$\dmif  q_2[0]= q_1[0] \land p_2[0]=p_1[0] \dmcomment{$\sigma[0]=0/0$}$\\
7\>\>$\sigma_0\becomes1$\\
\>$\dmcommentnotab{So $p_2[0]\neq p_1[0]$.}$\\
8\>$\dmelse\dmif q_2[0]= q_1[0]\dmcomment{Would require $\sigma_0$ to be zero.}$\\
9\>\>$\dmreturn \lbrace\rbrace$\\
10\>$\dmelse$\\
11\>\>$\sigma_0\becomes(q_2[0]-q_1[0])/(p_2[0]-p_1[0])\dmcomment{From \eqref{eleven-F2ST}}$\\
12\>$\sigma_1\becomes q_2[0]-p_2[0]\sigma_0\dmcomment{From \eqref{six-F2ST}}$\\
13\>$\dmif q_1[1]-p_1[1]=q_2[1] - p_2[1]\dmcomment{From \eqref{y-difference-equal-F2ST}}$\\
14\>\>$\sigma_2\becomes q_1[1]-p_1[1]\dmcomment{From \eqref{nine-F2ST}}$\\
15\>\>$\dmreturn\lbrace\langle\sigma_0,\sigma_1,\sigma_2\rangle\rbrace$\\
16\>$\dmreturn\lbrace\rbrace$
\end{tabbing}
\end{minipage}
}
\end{center}
\caption{Algorithm for computing the basis-pair transformation set function for $\FTwoST'$.}
\label{FTwoST-dash-basis-transformation-set-function}
\end{figure}

We define the class $\FTwoST'$ to be the set of 2-dimensional transformations that consist of an $x$-axis scaling by a {\em non-zero\/} real scale factor, followed by a translation. That is, $\FTwoST'$ is the subset of $\FTwoST$ that contains only the  transformations in $\FTwoST$ whose $x$-axis scale factors are non-zero. $\lambda(\FTwoST')=3$ and the transformation, $f\in\FTwoST'$, that consists of an $x$-axis scaling by a scale factor, $s$, followed by a translation by the vector, $\langle v_x, v_y\rangle$, can be uniquely identified by the transformation parameter vector, $\langle s,v_x, v_y\rangle\in\reals^3$, where $s\in\reals\setminus0$ and $v_x,v_y\in\reals$. The transformation class function for $\FTwoST'$ is the same as that for $\FTwoST$---that is,
\begin{equation}
\phi_{\FTwoST'}(\sigma,p)=\langle p[0]\sigma[0]+\sigma[1],p[1]+\sigma[2]\rangle\,.
\label{phi-f2st-dash}
\end{equation}
Also, as for $\FTwoST$, $\beta(\FTwoST')=2$ and $\tau(\FTwoST')=1$. The basis-pair transformation set function for $\FTwoST'$ is as follows:
\begin{equation}
\BTSF_{\FTwoST'}\left(\langle p_1,p_2\rangle,\langle q_1,q_2\rangle\right)=\left\lbrace\left\langle\sigma[0],q_2[0]-p_2[0]\sigma[0],q_1[1]-p_1[1]\right\rangle\right\rbrace\,,
\end{equation}
where
$$q_1[1]-p_1[1]=q_2[1] - p_2[1]$$
and where
\begin{equation}
\sigma[0]=\begin{cases}
1,&\text{ if $p_2[0]=p_1[0]$ and $q_2[0]= q_1[0]$; and}\\
\frac{q_2[0]-q_1[0]}{p_2[0]-p_1[0]},&\text{otherwise, provided that }p_2[0]\neq p_1[0]\text{ and }q_2[0]\neq q_1[0]\,.
\end{cases}
\end{equation}
$\BTSF_{\FTwoST'}$ can be computed using the algorithm in Fig.~\ref{FTwoST-dash-basis-transformation-set-function}. Note that this algorithm is the same as that in Fig.~\ref{FTwoST-basis-transformation-set-function} except that an extra check has been added in lines 8--9 to ensure that the $x$-axis scale factor cannot be zero.

We will now show that $\FTwoST'$ is a transformation group.

\begin{lemma}
$\FTwoST'$ is a transformation group.
\label{f2std-group}
\end{lemma}
\begin{proof}
In order to prove that $\FTwoST'$ is a transformation group, we need to show that it satisfies the conditions of closure, associativity, presence of a unique identity element, and invertibility.
\begin{enumerate}
\item{\bfseries Closure}: 
Let $f,g\in\FTwoST'$, $p\in\reals^2$ and $p=\langle p_x, p_y\rangle$. By definition, $f(g(p))=(f\circ g)(p)$. Let $\sigma^f = \sigma_{\FTwoST'}(f)$ and let $\sigma^g = \sigma_{\FTwoST'}(g)$. Let $\sigma^x_i=\sigma^x[i]$. From \eqref{phi-f2st-dash},
\begin{eqnarray}
f(g(p))&=&(f\circ g)(p)\nonumber\\
&=&\phiftwostd\left(\sigma^f,\left\langle p_x\sigma^g_0+\sigma^g_1,p_y+\sigma^g_2\right\rangle\right)\nonumber\\
&=&\left\langle\left(p_x\sigma^g_0+\sigma^g_1\right)\sigma^f_0+\sigma^f_1,\left(p_y+\sigma^g_2\right)+\sigma^f_2\right\rangle\nonumber\\
&=&\left\langle p_x\left(\sigma^g_0\sigma^f_0\right)+\left(\sigma^g_1\sigma^f_0+\sigma^f_1\right), p_y+\left(\sigma^g_2+\sigma^f_2\right)\right\rangle\nonumber\\
&=&\phiftwostd\left(\left\langle \sigma^g_0\sigma^f_0,\sigma^g_1\sigma^f_0+\sigma^f_1,\sigma^g_2+\sigma^f_2\right\rangle,p\right)\label{ftwost-dash-closure}
\end{eqnarray}
Therefore, 
$$\sigma_{\FTwoST'}(f\circ g)= \left\langle \sigma^g_0\sigma^f_0,\sigma^g_1\sigma^f_0+\sigma^f_1,\sigma^g_2+\sigma^f_2\right\rangle\,,$$
which implies that $f\circ g\in \FTwoST'$, which implies that $\FTwoST'$ is closed under transformation composition.
\item{\bfseries Associativity}: Let $f,g,h\in\FTwoST'$, $p\in\reals^2$ and $p=\langle p_x, p_y\rangle$. 
Let $\sigma^f = \sigma_{\FTwoST'}(f)$, $\sigma^g = \sigma_{\FTwoST'}(g)$ and $\sigma^h = \sigma_{\FTwoST'}(h)$. 
Let $\sigma^x_i=\sigma^x[i]$. 
From \eqref{ftwost-dash-closure},
\begin{eqnarray}
(f\circ g)(p)&=&\phiftwostd\left(\left\langle \sigma^g_0\sigma^f_0,\sigma^g_1\sigma^f_0+\sigma^f_1,\sigma^g_2+\sigma^f_2\right\rangle,p\right)\label{ftwostd-fog}\\
(g\circ h)(p)&=&\phiftwostd\left(\left\langle \sigma^h_0\sigma^g_0,\sigma^h_1\sigma^g_0+\sigma^g_1,\sigma^h_2+\sigma^g_2\right\rangle,p\right)\label{ftwostd-goh}
\end{eqnarray}
Therefore, from \eqref{phi-f2st-dash} and \eqref{ftwostd-fog},
\begin{eqnarray}
(f&\circ& g)(h(p))\nonumber\\
&=&\phiftwostd\left(\left\langle \sigma^g_0\sigma^f_0,\sigma^g_1\sigma^f_0+\sigma^f_1,\sigma^g_2+\sigma^f_2\right\rangle,\phiftwostd\left(\sigma_{\FTwoST'}\left(h\right),p\right)\right)\nonumber\\
&=&\phiftwostd\left(\left\langle \sigma^g_0\sigma^f_0,\sigma^g_1\sigma^f_0+\sigma^f_1,\sigma^g_2+\sigma^f_2\right\rangle,\left\langle p_x\sigma^h_0+\sigma^h_1,p_y+\sigma^h_2\right\rangle\right)\nonumber\\
&=&\left\langle \left(p_x\sigma^h_0+\sigma^h_1\right)\sigma^g_0\sigma^f_0+\sigma^g_1\sigma^f_0+\sigma^f_1, p_y+\sigma^h_2+\sigma^g_2+\sigma^f_2\right\rangle\nonumber\\
&=&\left\langle p_x\sigma^h_0\sigma^g_0\sigma^f_0+\sigma^h_1\sigma^g_0\sigma^f_0+\sigma^g_1\sigma^f_0+\sigma^f_1,p_y+\sigma^h_2+\sigma^g_2+\sigma^f_2\right\rangle\nonumber\\
&=&\phiftwostd\left(\left\langle\sigma^h_0\sigma^g_0\sigma^f_0,\sigma^h_1\sigma^g_0\sigma^f_0+\sigma^g_1\sigma^f_0+\sigma^f_1,\sigma^h_2+\sigma^g_2+\sigma^f_2\right\rangle,p\right)\nonumber\\
\label{f2std-assoc-1}
\end{eqnarray}
and, from \eqref{phi-f2st-dash} and \eqref{ftwostd-goh},
\begin{eqnarray}
f((g&\circ& h)(p))\nonumber\\
&=&\phiftwostd\left(\sigma_{\FTwoST'}(f),\phiftwostd\left(\left\langle \sigma^h_0\sigma^g_0,\sigma^h_1\sigma^g_0+\sigma^g_1,\sigma^h_2+\sigma^g_2\right\rangle,p\right)\right)\nonumber\\
&=&\phiftwostd\left(\left\langle\sigma^f_0,\sigma^f_1,\sigma^f_2\right\rangle,
\left\langle p_x\sigma^h_0\sigma^g_0+\sigma^h_1\sigma^g_0+\sigma^g_1,p_y+\sigma^h_2+\sigma^g_2\right\rangle
\right)\nonumber\\
&=&\left\langle \left(p_x\sigma^h_0\sigma^g_0+\sigma^h_1\sigma^g_0+\sigma^g_1\right)\sigma^f_0+\sigma^f_1,p_y+\sigma^h_2+\sigma^g_2+\sigma^f_2\right\rangle\nonumber\\
&=&\left\langle p_x\sigma^h_0\sigma^g_0\sigma^f_0+\sigma^h_1\sigma^g_0\sigma^f_0+\sigma^g_1\sigma^f_0+\sigma^f_1,p_y+\sigma^h_2+\sigma^g_2+\sigma^f_2\right\rangle\nonumber\\
&=&\phiftwostd\left(\left\langle\sigma^h_0\sigma^g_0\sigma^f_0,\sigma^h_1\sigma^g_0\sigma^f_0+\sigma^g_1\sigma^f_0+\sigma^f_1,\sigma^h_2+\sigma^g_2+\sigma^f_2\right\rangle,p\right)\nonumber\\
\label{f2std-assoc-2}
\end{eqnarray}
From \eqref{f2std-assoc-1} and \eqref{f2std-assoc-2} it follows that $(f\circ g)(h(p))=f((g\circ h)(p))$ and that therefore transformation composition is associative over $\FTwoST'$.
\item{\bfseries Identity element}: Suppose that $e\in\FTwoST'$ and $\sigma_{\FTwoST'}(e)=\langle1,0,0\rangle$. Let $p\in\reals^2$ and $p=\langle p_x, p_y\rangle$. Then, from \eqref{phi-f2st-dash},
\begin{eqnarray}
e(p)&=&\phiftwostd(\langle1,0,0\rangle,\langle p_x,p_y\rangle)\nonumber\\
&=&\langle 1p_x+0,p_y+0\rangle\nonumber\\
&=&\langle p_x,p_y\rangle=p\label{ftwostd-identity}
\end{eqnarray}
Therefore, $e(p)=p$ for all $p\in\reals^2$. $e$ is therefore a universal identity element in $\FTwoST'$.
\item{\bfseries Invertibility}: Let $p,q\in\reals^2$ and let $q=\phiftwostd(\sigma,p)$ where $\sigma=\langle\sigma_0,\sigma_1,\sigma_2\rangle$, $p=\langle p_x,p_y\rangle$ and $q=\langle q_x,q_y\rangle$. Then, from \eqref{phi-f2st-dash},
$$q=\langle p_x\sigma_0+\sigma_1,p_y+\sigma_2\rangle$$
and therefore
$$p=\langle p_x,p_y\rangle=\left\langle\frac{q_x}{\sigma_0}-\frac{\sigma_1}{\sigma_0},q_y-\sigma_2\right\rangle=\phiftwostd\left(\left\langle\frac1{\sigma_0},-\frac{\sigma_1}{\sigma_0},-\sigma_2\right\rangle,q\right)\,.$$
Therefore, if $f\in\FTwoST'$ and $\sigma=\sigma_{\FTwoST'}(f)$, then
$$\sigma_{\FTwoST'}(f^{-1})=\left\langle\frac1{\sigma[0]},-\frac{\sigma[1]}{\sigma[0]},-\sigma[2]\right\rangle\,.$$
Therefore, every transformation in $\FTwoST'$ has a unique inverse.
\end{enumerate}
$\FTwoST'$ is therefore a transformation group.\qed
\end{proof}

\begin{corollary}
$\mathcal M_{\FTwoST'}$ is a transformational equivalence relation.
\label{m-ftwostd-tran-eq-rel}
\end{corollary}
\begin{proof}
This follows directly from Lemma~\ref{f2std-group} and Corollary~\ref{match-rel-of-tran-grp-is-eq-rel}.\qed
\end{proof}

\begin{corollary}
If $P,D\subset\reals^2$, then $\occsetfun(P,D,\FTwoST')$ is a set of transformationally equivalent patterns.
\end{corollary}
\begin{proof}
From Corollary~\ref{m-ftwostd-tran-eq-rel}, we know that $\mathcal M_{\FTwoST'}$ is a transformational equivalence relation, and therefore an equivalence relation over $2^{\mathbb R^2}$. From Lemma~\ref{tran-eq-rel-part-lemma}, it follows that $\mathcal M_{\FTwoST'}$ partitions $2^{\mathbb R^2}$ exhaustively and exclusively into transformational equivalence classes. It follows that $\occsetfun(P,D,\FTwoST')\subset [P]_{\FTwoST'}$ and that therefore $\occsetfun(P,D,\FTwoST')$ is a set of transformationally equivalent patterns.\qed
\end{proof}

\subsubsection{$\FTwoSTR$ and $\FTwoSTR'$}
$\FTwoSTR$ was defined in Section~\ref{example-FTwoSTR} to be the set containing those 2-dimensional transformations that consist of a scaling parallel to the $x$-axis by a real-valued scale factor (which may be zero), followed by a translation, optionally followed by a reflection in the $x$-axis. Any transformation in $\FTwoSTR$ that has an $x$-axis scale factor of zero does not have a unique inverse. Therefore $\FTwoSTR$ is not a transformation group. 

Let $P,D\subset\reals^2$ and suppose that we have an occurrence set, $\occsetfun(P,D,\FTwoSTR)$, that contains an occurrence, $Q$, of $P$, such that the $x$-axis scale factor of the transformation that maps $P$ onto $Q$ is zero. Let $S$ be the subset of $\occsetfun(P,D,\FTwoSTR)$ that contains occurrences of $P$ in $D$ with respect to $\FTwoSTR$ that can be obtained from $P$ by transformations whose $x$-axis scale factor is {\em not\/} zero. There is then no transformation in $\FTwoSTR$ that maps $Q$ onto any of the occurrences in $S$, so if $S$ is not empty, then $\occsetfun(P,D,\FTwoSTR)$ is {\em not\/} a set of transformationally equivalent patterns. Moreover, this shows that $\mathcal M_{\FTwoSTR}$ is {\em not\/} a transformational equivalence relation.

We define the class $\FTwoSTR'$ to be the set of 2-dimensional transformations that consist of an $x$-axis scaling by a {\em non-zero\/} real scale factor, followed by a translation, optionally followed by a reflection in the $x$-axis. That is, $\FTwoSTR'$ is the subset of $\FTwoSTR$ that contains only the  transformations in $\FTwoSTR$ whose $x$-axis scale factors are non-zero. 
$\lambda(\FTwoSTR')=4$ and each transformation $f\in\FTwoSTR'$, can be uniquely identified by a 4-vector, $\langle s, v_x, v_y, b\rangle\in\reals^4$, where $s\in\reals\setminus0$, $v_x,v_y\in\reals$ and $b\in\lbrace-1,1\rbrace$. If $\sigma_{\FTwoSTR'}(f)=\langle s, v_x, v_y, b\rangle$, then $f$ is the transformation that scales parallel to the $x$-axis by a {\em non-zero\/} scale factor, $s$, then translates by the vector, $\langle v_x, v_y\rangle$ and then multiplies the $y$-coordinate of the resulting point by $b$ (which is either $-1$ or $1$). 
\begin{figure}[t]
\begin{center}
\resizebox*{.7\columnwidth}{!}{
\begin{minipage}{\linewidth}
\begin{tabbing}
\input{ColumnDefinitions}
$\fun{$\BTSF_{\FTwoSTR'}$}(\objectBasis,\imageBasis)$\\
1\>$\dmif \imageBasis = \objectBasis$\\
2\>\>$\dmreturn\lbrace\langle1,0,0,1\rangle\rbrace\dmcomment{Identity transformation}$\\
3\>$p_1\becomes\objectBasis[0],p_2\becomes\objectBasis[1],q_1\becomes\imageBasis[0],q_2\becomes\imageBasis[1]$\\
4\>$\dmif q_2[0]\neq q_1[0]\land p_2[0]=p_1[0]\dmcomment{$\sigma[0]$ is undefined.}$\\
5\>\>$\dmreturn \lbrace\rbrace$\\
6\>$\dmif q_2[0]= q_1[0] \land p_2[0]=p_1[0] \dmcomment{$\sigma[0]=0/0$}$\\
7\>\>$\sigma_0\becomes1$\\
\>$\dmcommentnotab{So $p_2[0]\neq p_1[0]$}$\\
8\>$\dmelse\dmif q_2[0]= q_1[0]\dmcomment{Would require $\sigma_0$ to be zero.}$\\
9\>\>$\dmreturn \lbrace\rbrace$\\
10\>$\dmelse$\\
11\>\>$\sigma_0\becomes(q_2[0]-q_1[0])/(p_2[0]-p_1[0])\dmcomment{From \eqref{eleven}}$\\
12\>$S\becomes\lbrace\rbrace$\\
13\>$\sigma_1\becomes q_2[0]-p_2[0]\sigma_0\dmcomment{From \eqref{six}}$\\
14\>$\dmif q_1[1]-p_1[1]=q_2[1] - p_2[1]\dmcomment{From \eqref{y-difference-equal}}$\\
15\>\>$\sigma_3\becomes1$\\
16\>\>$\sigma_2\becomes q_1[1]-p_1[1]\dmcomment{From \eqref{nine}}$\\
17\>\>$S\becomes \lbrace\langle\sigma_0,\sigma_1,\sigma_2,\sigma_3\rangle\rbrace$\\
18\>$\dmif p_1[1]+q_1[1]=p_2[1] + q_2[1]\dmcomment{From \eqref{y-difference-equal-sigma-3-equals-1}}$\\
19\>\>$\sigma_3\becomes-1$\\
20\>\>$\sigma_2\becomes-p_1[1]-q_1[1]\dmcomment{From \eqref{seven}}$\\
21\>\>$S\becomes S\cup\lbrace\langle\sigma_0,\sigma_1,\sigma_2,\sigma_3\rangle\rbrace$\\
22\>$\dmreturn S$
\end{tabbing}
\end{minipage}
}
\end{center}
\caption{Algorithm for computing the basis-pair transformation set function for $\FTwoSTR'$.}
\label{FTwoSTR-dash-basis-transformation-set-function}
\end{figure}

The transformation class function for $\FTwoSTR'$ is the same as that for $\FTwoSTR$---that is,
\begin{equation}
\phi_{\FTwoSTR'}(\sigma,p)=\langle p[0]\sigma[0]+\sigma[1],\sigma[3](p[1]+\sigma[2])\rangle\,.
\label{phi-f2str-dash}
\end{equation}
Also, as for $\FTwoSTR$, $\beta(\FTwoSTR')=2$ and $\tau(\FTwoSTR')=2$. The basis-pair transformation set function for $\FTwoSTR'$ is as follows:
\begin{multline}
\BTSF_{\FTwoSTR'}\left(\langle p_1,p_2\rangle,\langle q_1,q_2\rangle\right)
=\\
\left\lbrace\left\langle\sigma[0],q_2[0]-p_2[0]\sigma[0],q_1[1]-p_1[1],1\right\rangle\right.,\\
\left.\left\langle\sigma[0],q_2[0]-p_2[0]\sigma[0],-p_1[1]-q_1[1],-1\right\rangle\right\rbrace\,,
\label{equation-BTSF-FTwoSTR-dash}
\end{multline}
where 
$$q_1[1]-p_1[1]=q_2[1] - p_2[1] \text{ when } \sigma[3]=1\,,$$
and 
$$p_1[1]+q_1[1]=p_2[1] + q_2[1]\text{ when }\sigma[3]=-1\,,$$
and where 
\begin{equation}
\sigma[0]=\begin{cases}
1,&\text{ if $p_2[0]=p_1[0]$ and $q_2[0]= q_1[0]$; and}\\
\frac{q_2[0]-q_1[0]}{p_2[0]-p_1[0]},&\text{otherwise, provided that }p_2[0]\neq p_1[0]\text{ and }q_2[0]\neq q_1[0]\,.
\end{cases}
\end{equation}
$\BTSF_{\FTwoSTR'}$ can be computed using the algorithm in Fig.~\ref{FTwoSTR-dash-basis-transformation-set-function}. Note that this algorithm is the same as that in Fig.~\ref{FTwoSTR-basis-transformation-set-function} except that an extra check has been added in lines 9--10 to ensure that the $x$-axis scale factor cannot be zero.

We will now show that $\FTwoSTR'$ is a transformation group.

\begin{lemma}
$\FTwoSTR'$ is a transformation group.
\label{f2strd-group}
\end{lemma}
\begin{proof}
In order to prove that $\FTwoSTR'$ is a transformation group, we need to show that it satisfies the conditions of closure, associativity, presence of a unique identity element, and invertibility.
\begin{enumerate}
\item{\bfseries Closure}: 
Let $f,g\in\FTwoSTR'$, $p\in\reals^2$ and $p=\langle p_x, p_y\rangle$. By definition, $f(g(p))=(f\circ g)(p)$. Let $\sigma^f = \sigma_{\FTwoSTR'}(f)$ and let $\sigma^g = \sigma_{\FTwoSTR'}(g)$. Let $\sigma^x_i=\sigma^x[i]$. 
From \eqref{phi-f2str-dash},
\begin{eqnarray}
f(g(p))&=&(f\circ g)(p)\nonumber\\
&=&\phiftwostrd\left(\sigma^f,\left\langle p_x\sigma^g_0 + \sigma^g_1,\sigma^g_3\left(p_y+\sigma^g_2\right)\right\rangle\right)\nonumber\\
&=&\left\langle\left(p_x\sigma^g_0+\sigma^g_1\right)\sigma^f_0+\sigma^f_1,\sigma^f_3\left(\sigma^g_3\left(p_y+\sigma^g_2\right)+\sigma^f_2\right)\right\rangle\nonumber\\
&=&\left\langle p_x\sigma^g_0\sigma^f_0+\sigma^g_1\sigma^f_0+\sigma^f_1,\sigma^f_3\sigma^g_3\left(p_y+\sigma^g_2+\frac{\sigma^f_2}{\sigma^g_3}\right)\right\rangle\nonumber\\
&=&\phiftwostrd\left(\left\langle \sigma^g_0\sigma^f_0, \sigma^g_1\sigma^f_0+\sigma^f_1, \sigma^g_2+\frac{\sigma^f_2}{\sigma^g_3}, \sigma^f_3\sigma^g_3\right\rangle,p\right)\nonumber\\
\label{ftwostr-dash-closure}
\end{eqnarray}
Therefore 
$$
\sigma_{\FTwoSTR'}\left(f\circ g\right)=\left\langle \sigma^g_0\sigma^f_0, \sigma^g_1\sigma^f_0+\sigma^f_1, \sigma^g_2+\frac{\sigma^f_2}{\sigma^g_3}, \sigma^f_3\sigma^g_3\right\rangle\,,
$$
which implies that $f\circ g\in \FTwoSTR'$, which implies that $\FTwoSTR'$ is closed under transformation composition.
\item{\bfseries Associativity}: Let $f,g,h\in\FTwoSTR'$, $p\in\reals^2$ and $p=\langle p_x, p_y\rangle$. 
Let $\sigma^f = \sigma_{\FTwoSTR'}(f)$, $\sigma^g = \sigma_{\FTwoSTR'}(g)$ and $\sigma^h = \sigma_{\FTwoSTR'}(h)$. 
Let $\sigma^x_i=\sigma^x[i]$. 
From \eqref{ftwostr-dash-closure},
\begin{eqnarray}
(f\circ g)(p)&=&\phiftwostrd\left(\left\langle \sigma^g_0\sigma^f_0, \sigma^g_1\sigma^f_0+\sigma^f_1, \sigma^g_2+\frac{\sigma^f_2}{\sigma^g_3}, \sigma^f_3\sigma^g_3\right\rangle,p\right)\label{ftwostrd-fog}\\
(g\circ h)(p)&=&\phiftwostrd\left(\left\langle \sigma^h_0\sigma^g_0, \sigma^h_1\sigma^g_0+\sigma^g_1, \sigma^h_2+\frac{\sigma^g_2}{\sigma^h_3}, \sigma^g_3\sigma^h_3\right\rangle,p\right)\label{ftwostrd-goh}
\end{eqnarray}
Therefore, from \eqref{phi-f2str-dash} and \eqref{ftwostrd-fog},
\begin{eqnarray}
(f&\circ& g)(h(p))\nonumber\\
&=&\phiftwostrd\left(\left\langle \sigma^g_0\sigma^f_0, \sigma^g_1\sigma^f_0+\sigma^f_1, \sigma^g_2+\frac{\sigma^f_2}{\sigma^g_3}, \sigma^f_3\sigma^g_3\right\rangle,h(p)\right)\nonumber\\
&=&\phiftwostrd\left(\left\langle \sigma^g_0\sigma^f_0, \sigma^g_1\sigma^f_0+\sigma^f_1, \sigma^g_2+\frac{\sigma^f_2}{\sigma^g_3}, \sigma^f_3\sigma^g_3\right\rangle,\left\langle p_x\sigma^h_0 + \sigma^h_1,\sigma^h_3\left(p_y+\sigma^h_2\right)\right\rangle\right)\nonumber\\
&=&	\left\langle 
		\left(p_x\sigma^h_0 + \sigma^h_1\right)\sigma^g_0\sigma^f_0+ \sigma^g_1\sigma^f_0+\sigma^f_1,
		\sigma^f_3\sigma^g_3\left(\sigma^h_3\left(p_y+\sigma^h_2\right)+\sigma^g_2+\frac{\sigma^f_2}{\sigma^g_3}\right)
	\right\rangle
	\nonumber\\
&=&\left\langle
	p_x\sigma^h_0\sigma^g_0\sigma^f_0
	+ \sigma^h_1\sigma^g_0\sigma^f_0+\sigma^g_1\sigma^f_0+\sigma^f_1,
	\sigma^f_3\sigma^g_3\sigma^h_3\left(
		p_y + \sigma^h_2 + \frac{\sigma^g_2}{\sigma^h_3} + \frac{\sigma^f_2}{\sigma^g_3\sigma^h_3}
	\right)
\right\rangle
\nonumber\\
&=&\phiftwostrd\left(
	\left\langle
		\sigma^h_0\sigma^g_0\sigma^f_0,
		\sigma^h_1\sigma^g_0\sigma^f_0+\sigma^g_1\sigma^f_0+\sigma^f_1,
		\sigma^h_2 + \frac{\sigma^g_2}{\sigma^h_3} + \frac{\sigma^f_2}{\sigma^g_3\sigma^h_3},
		\sigma^f_3\sigma^g_3\sigma^h_3
	\right\rangle,
	p
\right)\nonumber\\\label{f2strd-assoc-1}
\end{eqnarray}
and, from \eqref{phi-f2str-dash} and \eqref{ftwostrd-goh},
\begin{eqnarray}
f((g&\circ& h)(p))\nonumber\\
&=&\phiftwostrd\left(
	\sigma_{\FTwoSTR'}(f),
	\phiftwostrd\left(
		\sigma_{\FTwoSTR'}\left(g\circ h\right),
		p
	\right)
\right)\nonumber\\
&=&\phiftwostrd\left(
	\sigma_{\FTwoSTR'}(f),
	\phiftwostrd\left(
	\left\langle 
		\sigma^h_0\sigma^g_0, 
		\sigma^h_1\sigma^g_0+\sigma^g_1, 
		\sigma^h_2+\frac{\sigma^g_2}{\sigma^h_3}, 
		\sigma^g_3\sigma^h_3
	\right\rangle,
	p
	\right)
\right)\nonumber\\
&=&\phiftwostrd\left(
	\sigma_{\FTwoSTR'}(f),
	\left\langle
		p_x\sigma^h_0\sigma^g_0 + \sigma^h_1\sigma^g_0+\sigma^g_1,
		\sigma^g_3\sigma^h_3\left(
			p_y + \sigma^h_2+\frac{\sigma^g_2}{\sigma^h_3}
		\right)
	\right\rangle
\right)\nonumber\\
&=&\left\langle
	\sigma^f_0\left(
		p_x\sigma^h_0\sigma^g_0 + \sigma^h_1\sigma^g_0+\sigma^g_1
	\right)
	+\sigma^f_1,
	\sigma^f_3\left(
		\sigma^g_3\sigma^h_3\left(
			p_y + \sigma^h_2+\frac{\sigma^g_2}{\sigma^h_3}
		\right)
		+\sigma^f_2
	\right)
\right\rangle\nonumber\\
&=&\left\langle
	p_x\sigma^f_0\sigma^h_0\sigma^g_0 + \sigma^f_0\sigma^h_1\sigma^g_0+\sigma^f_0\sigma^g_1
	+\sigma^f_1,
		\sigma^f_3\sigma^g_3\sigma^h_3\left(
			p_y + \sigma^h_2+\frac{\sigma^g_2}{\sigma^h_3}
			+ \frac{\sigma^f_2}{\sigma^g_3\sigma^h_3}
		\right)
\right\rangle
\nonumber\\
&=&\phiftwostrd\left(
	\left\langle
		\sigma^f_0\sigma^h_0\sigma^g_0,
		\sigma^f_0\sigma^h_1\sigma^g_0+\sigma^f_0\sigma^g_1+\sigma^f_1,
		\sigma^h_2+\frac{\sigma^g_2}{\sigma^h_3} + \frac{\sigma^f_2}{\sigma^g_3\sigma^h_3},
		\sigma^f_3\sigma^g_3\sigma^h_3
	\right\rangle,
	p
\right)
\nonumber\\
\label{f2strd-assoc-2}
\end{eqnarray}
From \eqref{f2strd-assoc-1} and \eqref{f2strd-assoc-2}, it follows that $(f\circ g)(h(p))=f((g\circ h)(p))$ and that therefore transformation composition is associative over $\FTwoSTR'$.
\item{\bfseries Identity element}: Suppose that $e\in\FTwoSTR'$ and $\sigma_{\FTwoSTR'}(e)=\langle1,0,0,1\rangle$. Let $p\in\reals^2$ and $p=\langle p_x, p_y\rangle$. Then, from \eqref{phi-f2str-dash},
\begin{eqnarray}
e(p)&=&\phiftwostrd(\langle1,0,0,1\rangle,\langle p_x,p_y\rangle)\nonumber\\
&=&\langle 1p_x+0,1(p_y+0)\rangle\nonumber\\
&=&\langle p_x,p_y\rangle=p\label{ftwostrd-identity}
\end{eqnarray}
Therefore, $e(p)=p$ for all $p\in\reals^2$. $e$ is therefore a universal identity element in $\FTwoSTR'$.
\item{\bfseries Invertibility}: Let $p,q\in\reals^2$ and let $q=\phiftwostrd(\sigma,p)$ where $\sigma=\langle\sigma_0,\sigma_1,\sigma_2,\sigma_3\rangle$, $p=\langle p_x,p_y\rangle$ and $q=\langle q_x,q_y\rangle$. Then, from \eqref{phi-f2str-dash},
$$q=\langle p_x\sigma_0+\sigma_1,\sigma_3(p_y+\sigma_2)\rangle$$
and therefore
$$p=\langle p_x,p_y\rangle=\left\langle
	\frac{q_x}{\sigma_0}-\frac{\sigma_1}{\sigma_0},
	\frac{1}{\sigma_3}\left(q_y-\sigma_3\sigma_2\right)
\right\rangle=\phiftwostrd\left(
	\left\langle\frac1{\sigma_0}, -\frac{\sigma_1}{\sigma_0},-\sigma_3\sigma_2,\frac1{\sigma_3}\right\rangle,
	q\right)\,.$$
Recall that $\sigma_3\in\left\lbrace1,-1\right\rbrace$. Therefore, if $f\in\FTwoSTR'$ and $\sigma=\sigma_{\FTwoSTR'}(f)$, then
$$\sigma_{\FTwoSTR'}(f^{-1})=\left\langle\frac1{\sigma_0}, -\frac{\sigma_1}{\sigma_0},-\sigma_3\sigma_2,\sigma_3\right\rangle\,.$$
Therefore, every transformation in $\FTwoSTR'$ has a unique inverse.
\end{enumerate}
$\FTwoSTR'$ is therefore a transformation group.\qed
\end{proof}

\begin{corollary}
$\mathcal M_{\FTwoSTR'}$ is a transformational equivalence relation.
\label{m-ftwostrd-tran-eq-rel}
\end{corollary}
\begin{proof}
This follows directly from Lemma~\ref{f2strd-group} and Corollary~\ref{match-rel-of-tran-grp-is-eq-rel}.\qed
\end{proof}

\begin{corollary}
If $P,D\subset\reals^2$, then $\occsetfun(P,D,\FTwoSTR')$ is a set of transformationally equivalent patterns.
\end{corollary}
\begin{proof}
From Corollary~\ref{m-ftwostrd-tran-eq-rel}, we know that $\mathcal M_{\FTwoSTR'}$ is a transformational equivalence relation, and therefore an equivalence relation over $2^{\mathbb R^2}$. From Lemma~\ref{tran-eq-rel-part-lemma}, it follows that $\mathcal M_{\FTwoSTR'}$ partitions $2^{\mathbb R^2}$ exhaustively and exclusively into transformational equivalence classes. It follows that $\occsetfun(P,D,\FTwoSTR')\subset [P]_{\FTwoSTR'}$ and that therefore $\occsetfun(P,D,\FTwoSTR')$ is a set of transformationally equivalent patterns.\qed
\end{proof}

\subsubsection{$\FTwoSTRMod{\mu}$ and $\FTwoSTRMod{\mu}'$}

For some specified value of $\mu$, $\FTwoSTRMod{\mu}$ was defined in Section~\ref{example-FTwoSTRmodMu} above to be the set of 2-dimensional transformations, such that each transformation, $f\in\FTwoSTRMod{\mu}$, takes a point, $p=\langle x, y\rangle$, multiplies $x$ by a real-valued scale factor, $s$, (which may be zero) to get a new point, $p'=\langle sx,y\rangle$, then translates $p'$ by a vector, $v= \langle v_x, v_y\rangle$, with the $y$-component of $p'$ being added modulo $\mu$ to the $y$-component of $v$, to give a new point, $p''= \langle sx+v_x,(y+v_y)\bmod\mu\rangle$. Finally, $p''$ may be reflected in the $x$-axis, modulo $\mu$, to give a point, $q=\langle sx+v_x,(b(y+v_y))\bmod\mu\rangle$, where $b=-1$ if $p''$ is reflected and $b=1$ if it is not. As was the case with $\FTwoST$ and $\FTwoSTR$, any transformation in $\FTwoSTRMod{\mu}$ that has an $x$-axis scale factor of zero does not have a unique inverse. Therefore $\FTwoSTRMod{\mu}$ is not a transformation group.

Let $P,D\subset\reals^2$ and suppose that we have an occurrence set, $\occsetfun(P,D,\FTwoSTRMod{\mu})$, for a given value of $\mu$, that contains an occurrence, $Q$, of $P$, such that the $x$-axis scale factor of the transformation that maps $P$ onto $Q$ is zero. Let $S$ be the subset of $\occsetfun(P,D,\FTwoSTRMod{\mu})$ that contains occurrences of $P$ in $D$ with respect to $\FTwoSTRMod{\mu}$ that can be obtained from $P$ by transformations whose $x$-axis scale factor is {\em not\/} zero. There is then no transformation in $\FTwoSTRMod{\mu}$ that maps $Q$ onto any of the occurrences in $S$, so if $S$ is not empty, then $\occsetfun(P,D,\FTwoSTRMod{\mu})$ is {\em not\/} a set of transformationally equivalent patterns. Moreover, this shows that $\mathcal M_{\FTwoSTRMod{\mu}}$ is {\em not\/} a transformational equivalence relation.

We therefore define that, for a given value of $\mu$, $\FTwoSTRMod{\mu}'$ is the subset of $\FTwoSTRMod{\mu}$ that contains only the  transformations in $\FTwoSTRMod{\mu}$ whose $x$-axis scale factors are non-zero. 
$\lambda\left(\FTwoSTRMod{\mu}'\right)=4$ and each transformation, $f\in\FTwoSTRMod{\mu}'$, can be uniquely identified by a 4-vector, $\langle s, v_x, v_y, b\rangle$, where $s\in\reals\setminus0$, $v_x\in\reals$, $v_y\in\reals_{\mu}$ and $b\in\lbrace-1,1\rbrace$. If $\sigma_{\FTwoSTRMod{\mu}'}(f)=\langle s, v_x, v_y, b\rangle$, then $f(p)=\langle sx+v_x,(b(y+v_y))\bmod\mu\rangle$ and the transformation class function is the same as for $\ftwostrmu$, that is,
\begin{equation}
\phi_{\FTwoSTRMod{\mu}'}(\sigma,p)=\left\langle p[0]\sigma[0]+\sigma[1],\left(\sigma[3]\left(p[1]+\sigma[2]\right)\right)\bmod\mu\right\rangle\,,
\label{phi-f2str-mod-dash}
\end{equation}
where $p\in\reals\times\reals_{\mu}$ and there must exist a transformation, $f\in\FTwoSTRMod{\mu}'$, such that $\sigma=\sigma_{\FTwoSTRMod{\mu}'}(f)$.

We can define that $\beta\left(\FTwoSTRMod{\mu}'\right)=2$ and that $\tau\left(\FTwoSTRMod{\mu}'\right)=2$. Given an object basis, $\langle p_1,p_2\rangle$, and an image basis, $\langle q_1,q_2\rangle$, the basis pair transformation function for $\ftwostrmu'$ is defined as follows:
\begin{multline}
\BTSF_{\FTwoSTRMod{\mu}'}\left(\langle p_1,p_2\rangle,\langle q_1,q_2\rangle\right)=\\
\left\lbrace\left\langle\sigma[0],q_2[0]-p_2[0]\sigma[0],(q_1[1]-p_1[1])\bmod\mu,1\right\rangle\right.,\\
\left.\left\langle\sigma[0],q_2[0]-p_2[0]\sigma[0],(-p_1[1]-q_1[1])\bmod\mu,-1\right\rangle\right\rbrace
\label{equation-BTSFFTwoSTRModMuDash}
\end{multline}
where
\begin{equation}
(q_1[1]-p_1[1])\bmod\mu=(q_2[1] - p_2[1])\bmod\mu\text{ when }\sigma[3]=1\,,
\label{tfstrmodmu-dash-condition-1}
\end{equation}
and
\begin{equation}
(p_1[1]+q_1[1])\bmod\mu=(p_2[1] + q_2[1])\bmod\mu\text{ when }\sigma[3]=-1\,,
\label{tfstrmodmu-dash-condition-2}
\end{equation}
and
\begin{equation}
\sigma[0]=\begin{cases}
1,&\text{ if $p_2[0]=p_1[0]$ and $q_2[0]= q_1[0]$; and}\\
\frac{q_2[0]-q_1[0]}{p_2[0]-p_1[0]},&\text{otherwise, provided that }p_2[0]\neq p_1[0]\text{ and }q_2[0]\neq q_1[0]\,.
\end{cases}
\end{equation}
$\BTSF_{\FTwoSTRMod{\mu}'}$ can be computed using the algorithm in Fig.~\ref{FTwoSTRModMu-dash-basis-transformation-set-function}. Note that this algorithm is the same as that in Fig.~\ref{FTwoSTRModMu-basis-transformation-set-function} except that an extra check has been added in lines 8--9 to ensure that the $x$-axis scale factor cannot be zero.

\begin{figure}[t]
\begin{center}
\resizebox*{.75\columnwidth}{!}{
\begin{minipage}{\linewidth}
\begin{tabbing}
\input{ColumnDefinitions}
$\fun{$\BTSF_{\FTwoSTRMod{\mu}'}$}(\objectBasis,\imageBasis)$\\
1\>$\dmif \imageBasis = \objectBasis$\\
2\>\>$\dmreturn\lbrace\langle1,0,0,1\rangle\rbrace\dmcomment{Identity transformation}$\\
3\>$p_1\becomes\objectBasis[0],p_2\becomes\objectBasis[1],q_1\becomes\imageBasis[0],q_2\becomes\imageBasis[1]$\\
4\>$\dmif q_2[0]\neq q_1[0]\land p_2[0]=p_1[0]\dmcomment{$\sigma[0]$ is undefined.}$\\
5\>\>$\dmreturn \lbrace\rbrace$\\
6\>$\dmif  q_2[0]= q_1[0]\land p_2[0]=p_1[0] \dmcomment{$\sigma[0]=0/0$}$\\
7\>\>$\sigma_0\becomes1$\\
\>$\dmcommentnotab{So $p_2[0]\neq p_1[0]$}$\\
8\>$\dmelse\dmif q_2[0]= q_1[0]\dmcomment{Would require $\sigma_0$ to be zero.}$\\
9\>\>$\dmreturn \lbrace\rbrace$\\
10\>$\dmelse$\\
11\>\>$\sigma_0\becomes(q_2[0]-q_1[0])/(p_2[0]-p_1[0])\dmcomment{From \eqref{f2strmod-eq8}}$\\
12\>$S\becomes\lbrace\rbrace$\\
13\>$\sigma_1\becomes q_2[0]-p_2[0]\sigma_0\dmcomment{From \eqref{f2strmod-eq10}}$\\
14\>$\dmif (q_1[1]-p_1[1])\bmod\mu=(q_2[1] - p_2[1])\bmod\mu\dmcomment{From \eqref{f2strmod-y-difference-equal}}$\\
15\>\>$\sigma_3\becomes1$\\
16\>\>$\sigma_2\becomes (q_1[1]-p_1[1])\bmod\mu\dmcomment{From \eqref{f2strmod-eq5}}$\\
17\>\>$S\becomes\lbrace\langle\sigma_0,\sigma_1,\sigma_2,\sigma_3\rangle\rbrace$\\
18\>$\dmif (p_1[1]+q_1[1])\bmod\mu=(p_2[1] + q_2[1])\bmod\mu\dmcomment{From \eqref{f2strmod-y-difference-equal-sigma-3-equals-1}}$\\
19\>\>$\sigma_3\becomes-1$\\
20\>\>$\sigma_2\becomes(-p_1[1]-q_1[1])\bmod\mu\dmcomment{From \eqref{f2strmod-eq12}}$\\
21\>\>$S\becomes S\cup\lbrace\langle\sigma_0,\sigma_1,\sigma_2,\sigma_3\rangle\rbrace$\\
22\>$\dmreturn S$
\end{tabbing}
\end{minipage}
}
\end{center}
\caption{Algorithm for computing the basis-pair transformation set function for $\FTwoSTRMod{\mu}'$.}
\label{FTwoSTRModMu-dash-basis-transformation-set-function}
\end{figure}

We will now show that $\ftwostrmu'$ is a transformation group.
\begin{lemma}
For any real value of $\mu>0$, $\FTwoSTRMod{\mu}'$ is a transformation group.
\label{f2str-mod-dash-group}
\end{lemma}
\begin{proof}
In order to prove that $\FTwoSTRMod{\mu}'$ is a transformation group, we need to show that it satisfies the conditions of closure, associativity, presence of a unique identity element, and invertibility.
\begin{enumerate}
\item{\bfseries Closure}: 
Let $f,g\in\FTwoSTRMod{\mu}'$, $p\in\reals\times\reals_\mu$ and $p=\langle p_x, p_y\rangle$. By definition, $f(g(p))=(f\circ g)(p)$. Let $\sigma^f = \sigma_{\FTwoSTRMod{\mu}'}(f)$ and let $\sigma^g = \sigma_{\FTwoSTRMod{\mu}'}(g)$. Let $\sigma^x_i=\sigma^x[i]$. 
From \eqref{phi-f2str-mod-dash},
\begin{eqnarray}
f(g(p))&=&(f\circ g)(p)\nonumber\\
&=&\phi_{\ftwostrmu}'\left(\sigma^f,\left\langle p_x\sigma^g_0 + \sigma^g_1,\left(\sigma^g_3\left(p_y+\sigma^g_2\right)\right)\bmod\mu\right\rangle\right)\nonumber\\
&=&\left\langle\left(p_x\sigma^g_0+\sigma^g_1\right)\sigma^f_0+\sigma^f_1,\left(\sigma^f_3\left(\sigma^g_3\left(p_y+\sigma^g_2\right)+\sigma^f_2\right)\right)\bmod\mu\right\rangle\nonumber\\
&=&\left\langle p_x\sigma^g_0\sigma^f_0+\sigma^g_1\sigma^f_0+\sigma^f_1,\left(\sigma^f_3\sigma^g_3\left(p_y+\sigma^g_2+\frac{\sigma^f_2}{\sigma^g_3}\right)\right)\bmod\mu\right\rangle\nonumber\\
&=&\phi_{\ftwostrmu'}\left(\left\langle \sigma^g_0\sigma^f_0, \sigma^g_1\sigma^f_0+\sigma^f_1, \sigma^g_2+\frac{\sigma^f_2}{\sigma^g_3}, \sigma^f_3\sigma^g_3\right\rangle,p\right)\nonumber\\
\label{ftwostrmud-closure}
\end{eqnarray}
Therefore 
$$
\sigma_{\ftwostrmu'}\left(f\circ g\right)=\left\langle \sigma^g_0\sigma^f_0, \sigma^g_1\sigma^f_0+\sigma^f_1, \sigma^g_2+\frac{\sigma^f_2}{\sigma^g_3}, \sigma^f_3\sigma^g_3\right\rangle\,,
$$
which implies that $f\circ g\in \ftwostrmu'$, which implies that $\ftwostrmu'$ is closed under transformation composition.
\item{\bfseries Associativity}: Let $f,g,h\in\ftwostrmu'$, $p\in\reals\times\reals_\mu$ and $p=\langle p_x, p_y\rangle$. 
Let $\sigma^f = \sigma_{\ftwostrmu'}(f)$, $\sigma^g = \sigma_{\ftwostrmu'}(g)$ and $\sigma^h = \sigma_{\ftwostrmu'}(h)$. 
Let $\sigma^x_i=\sigma^x[i]$. 
From \eqref{ftwostrmud-closure},
\begin{eqnarray}
(f\circ g)(p)&=&\phi_{\ftwostrmu'}\left(\left\langle \sigma^g_0\sigma^f_0, \sigma^g_1\sigma^f_0+\sigma^f_1, \sigma^g_2+\frac{\sigma^f_2}{\sigma^g_3}, \sigma^f_3\sigma^g_3\right\rangle,p\right)\nonumber\\
\label{ftwostrmud-fog}\\
(g\circ h)(p)&=&\phi_{\ftwostrmu'}\left(\left\langle \sigma^h_0\sigma^g_0, \sigma^h_1\sigma^g_0+\sigma^g_1, \sigma^h_2+\frac{\sigma^g_2}{\sigma^h_3}, \sigma^g_3\sigma^h_3\right\rangle,p\right)\nonumber\\\label{ftwostrmud-goh}
\end{eqnarray}
Therefore, from \eqref{phi-f2str-mod-dash} and \eqref{ftwostrmud-fog},
\begin{eqnarray}
(f&\circ& g)(h(p))\nonumber\\
&=&\phi_{\ftwostrmu'}\left(\left\langle \sigma^g_0\sigma^f_0, \sigma^g_1\sigma^f_0+\sigma^f_1, \sigma^g_2+\frac{\sigma^f_2}{\sigma^g_3}, \sigma^f_3\sigma^g_3\right\rangle,h(p)\right)\nonumber\\
&=&\phi_{\ftwostrmu'}\left(\left\langle \sigma^g_0\sigma^f_0, \sigma^g_1\sigma^f_0+\sigma^f_1, \sigma^g_2+\frac{\sigma^f_2}{\sigma^g_3}, \sigma^f_3\sigma^g_3\right\rangle,
\left\langle p_x\sigma^h_0 + \sigma^h_1,\left(\sigma^h_3\left(p_y+\sigma^h_2\right)\right)\bmod\mu\right\rangle\right)\nonumber\\
&=&	\left\langle 
		\left(p_x\sigma^h_0 + \sigma^h_1\right)\sigma^g_0\sigma^f_0+ \sigma^g_1\sigma^f_0+\sigma^f_1,
		\left(\sigma^f_3\sigma^g_3\left(\sigma^h_3\left(p_y+\sigma^h_2\right)+\sigma^g_2+\frac{\sigma^f_2}{\sigma^g_3}\right)\right)\bmod\mu
	\right\rangle
	\nonumber\\
&=&\left\langle
	p_x\sigma^h_0\sigma^g_0\sigma^f_0
	+ \sigma^h_1\sigma^g_0\sigma^f_0+\sigma^g_1\sigma^f_0+\sigma^f_1,
	\left(\sigma^f_3\sigma^g_3\sigma^h_3\left(
		p_y + \sigma^h_2 + \frac{\sigma^g_2}{\sigma^h_3} + \frac{\sigma^f_2}{\sigma^g_3\sigma^h_3}
	\right)\right)\bmod\mu
\right\rangle
\nonumber\\
&=&\phi_{\ftwostrmu'}\left(
	\left\langle
		\sigma^h_0\sigma^g_0\sigma^f_0,
		\sigma^h_1\sigma^g_0\sigma^f_0+\sigma^g_1\sigma^f_0+\sigma^f_1,
		\sigma^h_2 + \frac{\sigma^g_2}{\sigma^h_3} + \frac{\sigma^f_2}{\sigma^g_3\sigma^h_3},
		\sigma^f_3\sigma^g_3\sigma^h_3
	\right\rangle,
	p
\right)\nonumber\\\label{f2strmud-assoc-1}
\end{eqnarray}
and, from \eqref{phi-f2str-mod-dash} and \eqref{ftwostrmud-goh},
\begin{eqnarray}
f((g&\circ& h)(p))\nonumber\\
&=&\phi_{\ftwostrmu'}\left(
	\sigma^f,
	\phi_{\ftwostrmu'}\left(
		\sigma_{\ftwostrmu'}\left(g\circ h\right),
		p
	\right)
\right)\nonumber\\
&=&\phi_{\ftwostrmu'}\left(
	\sigma^f,
	\phi_{\ftwostrmu'}\left(
	\left\langle 
		\sigma^h_0\sigma^g_0, 
		\sigma^h_1\sigma^g_0+\sigma^g_1, 
		\sigma^h_2+\frac{\sigma^g_2}{\sigma^h_3}, 
		\sigma^g_3\sigma^h_3
	\right\rangle,
	p
	\right)
\right)\nonumber\\
&=&\phi_{\ftwostrmu'}\left(
	\sigma^f,
	\left\langle
		p_x\sigma^h_0\sigma^g_0 + \sigma^h_1\sigma^g_0+\sigma^g_1,
		\left(\sigma^g_3\sigma^h_3\left(
			p_y + \sigma^h_2+\frac{\sigma^g_2}{\sigma^h_3}
		\right)\right)\bmod\mu
	\right\rangle
\right)\nonumber\\
&=&\left\langle
	\sigma^f_0\left(
		p_x\sigma^h_0\sigma^g_0 + \sigma^h_1\sigma^g_0+\sigma^g_1
	\right)
	+\sigma^f_1,
	\left(\sigma^f_3\left(
		\sigma^g_3\sigma^h_3\left(
			p_y + \sigma^h_2+\frac{\sigma^g_2}{\sigma^h_3}
		\right)
		+\sigma^f_2
	\right)\right)\bmod\mu
\right\rangle\nonumber\\
&=&\left\langle
	p_x\sigma^f_0\sigma^h_0\sigma^g_0 + \sigma^f_0\sigma^h_1\sigma^g_0+\sigma^f_0\sigma^g_1
	+\sigma^f_1,
		\left(\sigma^f_3\sigma^g_3\sigma^h_3\left(
			p_y + \sigma^h_2+\frac{\sigma^g_2}{\sigma^h_3}
			+ \frac{\sigma^f_2}{\sigma^g_3\sigma^h_3}
		\right)\right)\bmod\mu
\right\rangle
\nonumber\\
&=&\phi_{\ftwostrmu'}\left(
	\left\langle
		\sigma^f_0\sigma^h_0\sigma^g_0,
		\sigma^f_0\sigma^h_1\sigma^g_0+\sigma^f_0\sigma^g_1+\sigma^f_1,
		\sigma^h_2+\frac{\sigma^g_2}{\sigma^h_3} + \frac{\sigma^f_2}{\sigma^g_3\sigma^h_3},
		\sigma^f_3\sigma^g_3\sigma^h_3
	\right\rangle,
	p
\right)
\nonumber\\
\label{f2strmud-assoc-2}
\end{eqnarray}
From \eqref{f2strmud-assoc-1} and \eqref{f2strmud-assoc-2}, it follows that $(f\circ g)(h(p))=f((g\circ h)(p))$ and that therefore transformation composition is associative over $\ftwostrmu'$.
\item{\bfseries Identity element}: Suppose that $e\in\ftwostrmu'$ and $\sigma_{\ftwostrmu'}(e)=\langle1,0,0,1\rangle$. Let $p\in\reals\times\reals_{\mu}$ and $p=\langle p_x, p_y\rangle$. Then, from \eqref{phi-f2str-mod-dash},
\begin{eqnarray}
e(p)&=&\phi_{\ftwostrmu'}(\langle1,0,0,1\rangle,\langle p_x,p_y\rangle)\nonumber\\
&=&\langle 1p_x+0,(1(p_y+0))\bmod\mu\rangle\nonumber\\
&=&\langle p_x,p_y\rangle=p\label{ftwostrmud-identity}
\end{eqnarray}
Therefore, $e(p)=p$ for all $p\in\reals^2$. $e$ is therefore a universal identity element in $\ftwostrmu'$.
\item{\bfseries Invertibility}: Let $p,q\in\reals\times\reals_\mu$ and let $q=\phi_{\ftwostrmu'}(\sigma,p)$ where $\sigma=\langle\sigma_0,\sigma_1,\sigma_2,\sigma_3\rangle$, $p=\langle p_x,p_y\rangle$ and $q=\langle q_x,q_y\rangle$. Then, from \eqref{phi-f2str-mod-dash},
$$q=\langle p_x\sigma_0+\sigma_1,(\sigma_3(p_y+\sigma_2))\bmod\mu\rangle$$
and therefore
\begin{eqnarray}
p&=&\langle p_x,p_y\rangle\nonumber\\
&=&\left\langle
	\frac{q_x}{\sigma_0}-\frac{\sigma_1}{\sigma_0},
	\left(\frac{1}{\sigma_3}\left(q_y-\sigma_3\sigma_2\right)\right)\bmod\mu
\right\rangle\nonumber\\
&=&\phi_{\ftwostrmu'}\left(
	\left\langle\frac1{\sigma_0}, -\frac{\sigma_1}{\sigma_0},-\sigma_3\sigma_2,\frac1{\sigma_3}\right\rangle,
	q\right)\nonumber
\end{eqnarray}
Recall that $\sigma_3\in\left\lbrace1,-1\right\rbrace$. Therefore, if $f\in\ftwostrmu'$ and $\sigma=\sigma_{\ftwostrmu'}(f)$, then
$$\sigma_{\ftwostrmu'}(f^{-1})=\left\langle\frac1{\sigma_0}, -\frac{\sigma_1}{\sigma_0},-\sigma_3\sigma_2,\sigma_3\right\rangle\,.$$
Therefore, every transformation in $\ftwostrmu'$ has a unique inverse.
\end{enumerate}
$\ftwostrmu'$ is therefore a transformation group.\qed
\end{proof}

\begin{corollary}
$\mathcal M_{\ftwostrmu'}$ is a transformational equivalence relation.
\label{m-ftwostrmud-tran-eq-rel}
\end{corollary}
\begin{proof}
This follows directly from Lemma~\ref{f2str-mod-dash-group} and Corollary~\ref{match-rel-of-tran-grp-is-eq-rel}.\qed
\end{proof}

\begin{corollary}
If $P,D\subset\reals\times\reals_\mu$, then $\occsetfun\left(P,D,\ftwostrmu'\right)$ is a set of transformationally equivalent patterns.
\end{corollary}
\begin{proof}
From Corollary~\ref{m-ftwostrmud-tran-eq-rel}, we know that $\mathcal M_{\ftwostrmu'}$ is a transformational equivalence relation, and therefore an equivalence relation over $2^{\mathbb R^2}$. From Lemma~\ref{tran-eq-rel-part-lemma}, it follows that $\mathcal M_{\ftwostrmu'}$ partitions $2^{\mathbb R^2}$ exhaustively and exclusively into transformational equivalence classes. It follows that $\occsetfun\left(P,D,\ftwostrmu'\right)\subset [P]_{\ftwostrmu'}$ and that therefore $\occsetfun(P,D,\ftwostrmu')$ is a set of transformationally equivalent patterns.\qed
\end{proof}

\section{Summary}
In this paper,
we have presented a framework of concepts, theory, notation and terminology for geometric pattern matching.
Specifically, we began by providing a formal definition of the concept of a transformation class, and derived precise definitions of the transformation classes $\FTwoT$, $\FTwoTR$, $\FTwoST$, $\FTwoSTR$ and $\FTwoSTRMod{\mu}$. For each of these transformation classes, we derived the transformation class function and the basis-pair transformation set function. Where necessary, we also provided algorithms for computing these basis-pair transformation set functions.

We then provided precise definitions of the concepts of a match and an occurrence, preparing the way for defining the concepts of a maximal match and a maximal pattern and proving a lemma that characterises the relationship between the two. We then focused on maximal matches and maximal patterns for transformations that are bijections. In particular, we showed that the maximal pattern for a pattern, a dataset and a bijection is equal to the intersection of the pattern and the dataset transformed by the inverse of the bijection. We also proved a theorem that fully describes the relationships between the maximal patterns and maximal matches for two patterns with respect to a bijection and its inverse. We noted that maximal patterns correspond to partial symmetries.

We then defined the concepts of occurrence set and heterogeneous occurrence set, along with the notions of the covered set of an occurrence set and an occurrence set's compression factor. 

Finally, we examined in detail the concepts of transformational equivalence and transformation groups and analysed the group-theoretic and transformation-equivalence properties of the specific transformation classes defined earlier in the paper.
In particular, we showed that $\FTwoST$, $\FTwoSTR$ and $\FTwoSTRMod{\mu}$ do not form transformation groups, but that the subsets of these classes containing transformations with non-zero $x$-axis scale factors {\em do\/} form transformation groups. We name these subsets $\FTwoST'$, $\FTwoSTR'$ and $\FTwoSTRMod{\mu}'$, respectively.

\bibliographystyle{apalike}
\bibliography{/Users/susanne/Repos/bibliography/bibliography}

%% file: ColumnDefinitions.tex
\firstColWidth\=\colWidth\=\colWidth\=\colWidth\=\colWidth\=\colWidth\=\colWidth\=\colWidth\=\colWidth\=\colWidth\=\\

%% file: RAVEL-HAYDN-GROUND-TRUTH-CPO-F2STRBook-IPTG.tex
\begin{sidewaystable}
\caption{Inter-pattern transformation graph for the set of patterns in the CPO representation in Fig.~\ref{RAVEL-HAYDN-GROUND-TRUTH-CPO-GT} and the transformation class, $F_{\mathrm{2STR}}$.}
\label{RAVEL-HAYDN-GROUND-TRUTH-CPO-IPTG}
\resizebox{\linewidth}{!}{
\begin{tabularx}{1.33\linewidth}{ll|lllllllllllllllllllllllllllllll}
&&\multicolumn{31}{c}{{\bfseries To pattern}}\\
&&1&2&3&4&5&6&7&8&9&10&11&12&13&14&15&16&17&18&19&20&21&22&23&24&25&26&27&28&A&B&C\\\cline{1-33}
&1 &$I$ & & & & &$T_{\langle48,0\rangle}$ & & & & & & & & & & & & & & & & & & & & & & & & & \\\cline{2-33}
&2 & &$I$ & & & & &$T_{\langle48,0\rangle}$ & & & & & & & & & & & & & & & & & & & & & & & & \\\cline{2-33}
&3 & & &$I$ &$T_{\langle6,0\rangle}$ & & & &$T_{\langle48,0\rangle}$ &$T_{\langle54,0\rangle}$ & & & & & & & & & & & & & & & & & & & & & & \\\cline{2-33}
&4 & & &$T_{\langle-6,0\rangle}$ &$I$ & & & &$T_{\langle42,0\rangle}$ &$T_{\langle48,0\rangle}$ & & & & & & & & & & & & & & & & & & & & & & \\\cline{2-33}
&5 & & & & &$I$ & & & & &$T_{\langle48,0\rangle}$ & & & & & & & & & & & & & & & & & & & & & \\\cline{2-33}
&6 &$T_{\langle-48,0\rangle}$ & & & & &$I$ & & & & & & & & & & & & & & & & & & & & & & & & & \\\cline{2-33}
&7 & &$T_{\langle-48,0\rangle}$ & & & & &$I$ & & & & & & & & & & & & & & & & & & & & & & & & \\\cline{2-33}
&8 & & &$T_{\langle-48,0\rangle}$ &$T_{\langle-42,0\rangle}$ & & & &$I$ &$T_{\langle6,0\rangle}$ & & & & & & & & & & & & & & & & & & & & & & \\\cline{2-33}
&9 & & &$T_{\langle-54,0\rangle}$ &$T_{\langle-48,0\rangle}$ & & & &$T_{\langle-6,0\rangle}$ &$I$ & & & & & & & & & & & & & & & & & & & & & & \\\cline{2-33}
&10 & & & & &$T_{\langle-48,0\rangle}$ & & & & &$I$ & & & & & & & & & & & & & & & & & & & & & \\\cline{2-33}
&11 & & & & & & & & & & &$I$ & & & & & & & & & & & & & & & & & & & & \\\cline{2-33}
&12 & & & & & & & & & & & &$I$ & & & & & &$T_{\langle60,0\rangle}$ & &$T_{\langle84,-7\rangle}$ & & & & & & & & & & & \\\cline{2-33}
&13 & & & & & & & & & & & & &$I$ & & & & & & & & & & & & & & & & & & \\\cline{2-33}
&14 & & & & & & & & & & & & & &$I$ & & & & & & & & & & & & & & & & & \\\cline{2-33}
&15 & & & & & & & & & & & & & & &$I$ & & & & & & & & & & & & & & & & \\\cline{2-33}
&16 & & & & & & & & & & & & & & & &$I$ & & &$T_{\langle24,-12\rangle}$ & & & & & & & & & & & & \\\cline{2-33}
&17 & & & & & & & & & & & & & & & & &$I$ & & & & & & & & & & & & & & \\\cline{2-33}
&18 & & & & & & & & & & & &$T_{\langle-60,0\rangle}$ & & & & & &$I$ & &$T_{\langle24,-7\rangle}$ & & & & & & & & & & & \\\cline{2-33}
&19 & & & & & & & & & & & & & & & &$T_{\langle-24,12\rangle}$ & & &$I$ & & & & & & & & & & & & \\\cline{2-33}
&20 & & & & & & & & & & & &$T_{\langle-84,7\rangle}$ & & & & & &$T_{\langle-24,7\rangle}$ & &$I$ & & & & & & & & & & & \\\cline{2-33}
&21 & & & & & & & & & & & & & & & & & & & & &$I$ & & & & & & & & & & \\\cline{2-33}
&22 & & & & & & & & & & & & & & & & & & & & & &$I$ & & & & & & & & & \\\cline{2-33}
&23 & & & & & & & & & & & & & & & & & & & & & & &$I$ & & & & & & & & \\\cline{2-33}
&24 & & & & & & & & & & & & & & & & & & & & & & & &$I$ & & & & & & & \\\cline{2-33}
&25 & & & & & & & & & & & & & & & & & & & & & & & & &$I$ & & & & & & \\\cline{2-33}
&26 & & & & & & & & & & & & & & & & & & & & & & & & & &$I$ &$T_{\langle6,0\rangle}$ & & & & \\\cline{2-33}
&27 & & & & & & & & & & & & & & & & & & & & & & & & & &$T_{\langle-6,0\rangle}$ &$I$ & & & & \\\cline{2-33}
&28 & & & & & & & & & & & & & & & & & & & & & & & & & & & &$I$ & & & \\\cline{2-33}
&A & & & & & & & & & & & & & & & & & & & & & & & & & & & & &$I$ & & \\\cline{2-33}
&B & & & & & & & & & & & & & & & & & & & & & & & & & & & & & &$I$ & \\\cline{2-33}
\multirow{-31}{*}{\rotatebox[origin=c]{90}{{\bfseries From pattern}}}&C & & & & & & & & & & & & & & & & & & & & & & & & & & & & & & &$I$ \\\cline{2-33}
\end{tabularx}
}
\end{sidewaystable}